\documentclass[a4paper]{amsart}

\usepackage[foot]{amsaddr}

\usepackage{amsmath}
\usepackage{stmaryrd}
\usepackage{tikz}
\usetikzlibrary{shapes.geometric}
\usetikzlibrary{backgrounds,calc,positioning}
\usepackage{forest}
\usetikzlibrary{arrows}
\usepackage{microtype}
\usepackage{xparse}
\usepackage{enumerate}

\usepackage{hyperref}

\theoremstyle{plain}
\newtheorem{thm}{Theorem}[section]
\newtheorem{prop}[thm]{Proposition}
\newtheorem{lem}[thm]{Lemma}
\newtheorem{cor}[thm]{Corollary}
\theoremstyle{remark}
\newtheorem{rem}[thm]{Remark}

\theoremstyle{definition}
\newtheorem{defi}[thm]{Definition}
\newtheorem{ex}[thm]{Example}

\newcommand{\A}{\mathcal{A}}
\newcommand{\roundlog}[1]{\lceil \log_2 \|#1\| \rceil}
\newcommand{\plainlog}[1]{\log_2 \|#1\|}
\newcommand{\B}{\mathcal{B}}
\newcommand{\calC}{\mathcal{C}}
\newcommand{\D}{\mathcal{D}}
\newcommand{\F}{\mathcal{F}}
\newcommand{\calG}{\mathcal{G}}
\newcommand{\V}{\mathcal{V}}
\newcommand{\X}{\mathcal{X}}
\newcommand{\Y}{\mathcal{Y}}
\newcommand{\Z}{\mathcal{Z}}
\newcommand{\N}{\mathbb{N}}
\newcommand{\So}{\mathcal{S}}
\newcommand{\T}{\mathcal{T}}

\newcommand{\valX}[1]{\llbracket #1 \rrbracket}
\newcommand{\valXG}[2]{\valX{#2}_{#1}}
\newcommand{\concv}{\mathbin{\varobar}}
\newcommand{\conch}{\mathbin{\varominus}}
\newcommand{\lin}{\text{lin}}
\newcommand{\mergv}{\mathbin{\varobar}}
\newcommand{\mergh}{\mathbin{\varominus}}
\newcommand{\depth}{\text{depth}}

\newcommand{\type}{\text{type}}

\DeclareMathOperator{\rmq}{RMQ}

\usepackage{color}
\definecolor{myYellow}{rgb}{0.9,0.9,0}
\definecolor{myGreen}{rgb}{0.1,1,0}

\begin{document}

\title{Balancing Straight-Line Programs}

\author{Moses Ganardi}
\email{ganardi@mpi-sws.org}

\author{Artur Je\.z}
\email{aje@cs.uni.wroc.pl}

\author{Markus~Lohrey}
\thanks{Markus Lohrey has been supported by the DFG research project LO 748/10-1.}
\email{lohrey@eti.uni-siegen.de}

\address[Moses~Ganardi]{Max Planck Institute for Software Systems (MPI-SWS), Germany}
\address[Artur Je\.z]{University of Wroc{\l}aw, Poland}
\address[Markus~Lohrey]{Universit{\"a}t Siegen, Germany}

\begin{abstract}
We show that a context-free grammar of size $m$ that produces a single string $w$ of length $n$
(such a grammar is also called a string straight-line program)
can be transformed in linear time into a context-free grammar for $w$ of size 
$\mathcal{O}(m)$, whose unique derivation tree has depth $\mathcal{O}(\log n)$.
This solves an open problem in the area of grammar-based compression,
improves many results in this area
and greatly simplifies many existing constructions.
Similar results are shown for two formalisms for grammar-based tree
compression: top dags and forest straight-line programs.
These balancing results can be all deduced from a single
meta theorem stating that the depth of an algebraic circuit over an algebra with a certain finite base property can be reduced 
to $\mathcal{O}(\log n)$ with the cost of a constant multiplicative size increase. Here, $n$ refers to the size of the 
unfolding (or unravelling) of the circuit.
In particular, this results applies to standard arithmetic circuits over (noncommutative) semirings.
\end{abstract}

\maketitle

\section{Introduction}

\paragraph{Grammar-based string compression}
In {\em grammar-based compression} a~combinatorial object is compactly represented
using a~grammar of an~appropriate type.
Such a~grammar can be up to exponentially smaller 
than the object itself.
A~well-studied example of this general idea is grammar-based string compression using 
context-free grammars that produce only one string,
which are also known as {\em straight-line programs}.
Since the term ``straight-line programs'' is used in the literature for different kinds of objects (e.g.\ arithmetic 
straight-line programs) and we will also deal with different types of straight-line programs, we use the term 
{\em string straight-line program}, SSLP for short.
Grammar-based string compression is tightly related to dictionary-based compression: the famous LZ78 algorithm can be viewed 
as a~particular grammar-based compressor, the number of phrases in the LZ77-factorization is a~lower bound for the smallest
SSLP for a~string \cite{Ryt03},
and an~LZ77-factorization of length $m$ can be converted to an~SSLP of size $\mathcal{O}(m \cdot \log n)$ where $n$ is the length of the string~\cite{Ryt03,CLLLPPSS05,Jez15,simplegrammar}.
For various other aspects of grammar-based string compression see \cite{CLLLPPSS05,lohrey_survey}.

\paragraph{Balancing string straight-line programs}
The~two important measures for an~SSLP are size and depth. To define these measures, it is convenient to assume that all
right-hand sides of the grammar have length two (as in Chomsky normal form). Then, the size $|\calG|$ of an~SSLP 
$\calG$ is the number of variables (nonterminals) of $\calG$ and the depth of $\calG$ ($\depth(\calG)$ 
for short) is the depth of the unique derivation tree of $\calG$. It is straightforward to show that any string $s$ of length $n$ 
can be produced by an~SSLP of size $\mathcal{O}(n)$ and depth $\mathcal{O}(\log n)$. A~more difficult problem is to balance a~%
given SSLP: Assume that the SSLP $\calG$ produces a~string of length $n$. Several authors have shown that one can 
restructure $\calG$ in time $\mathcal{O}(|\calG| \cdot \log n)$ into an~equivalent SSLP $\mathcal{H}$ of size 
$\mathcal{O}(|\calG| \cdot \log n)$ and depth $\mathcal{O}(\log n)$~\cite{Ryt03,CLLLPPSS05,simplegrammar}.

Finding SSLPs of small size and small depth is important in many algorithmic applications.
A~prominent example is the {\em random access problem for grammar-compressed strings}:
For a~given SSLP $\calG$ that produces the string $s$ of length 
$n$ and a~given position $p \in [1,n]$ one wants to access the $p$-th symbol in $s$.
As observed in \cite{BilleLRSSW15} one can solve this problem in time $\mathcal{O}(\depth(\calG))$
(assuming arithmetic operations on numbers from the interval $[0,n]$ use constant time).
Combined with one of the known SSLP balancing procedures \cite{Ryt03,CLLLPPSS05} one obtains
access time $\mathcal{O}(\log n)$, but one has to pay with an increased SSLP size of $\mathcal{O}(|\calG| \cdot \log n)$.
Using sophisticated data structures, the following result was shown in \cite{BilleLRSSW15}:

\begin{thm}[random access to grammar-compressed strings, cf.~\cite{BilleLRSSW15}] \label{cor-random-access}
From a given SSLP $\calG$ of size $m$ that generates the string $s$ of length $n$, 
one can construct in time $\mathcal{O}(m)$ a data structure of size $\mathcal{O}(m)$ (measured in words of bit length $\log n$)  that allows to answer 
random access queries in time  $\mathcal{O}(\log n)$.
\end{thm}
Our main result for string straight-line programs states that SSLP balancing is in fact possible with a~constant blow-up in size.

\begin{thm}
\label{thm-balance-SSLP}
Given an~SSLP $\calG$ producing a~string of length $n$
one can construct in linear time an~equivalent SSLP
$\mathcal{H}$ of size $\mathcal{O}(|\calG|)$ and depth $\mathcal{O}(\log n)$.
\end{thm}
As a~corollary we obtain a~very simple and clean proof of Theorem~\ref{cor-random-access}.
We can also obtain an~algorithm for the random access problem
with running time $\mathcal{O}(\log n / \log\log n)$ using $\mathcal{O}(m \cdot \log^{\epsilon} n)$ 
words of bit length $\log n$; previously this bound was only shown for balanced SSLPs~\cite{BelazzouguiCPT15}. 
Section~\ref{sec-applications}
contains a~list of further applications of Theorem~\ref{thm-balance-SSLP},
which include the following problems on SSLP-compressed strings:
rank and select queries \cite{BelazzouguiCPT15}, subsequence matching \cite{BilleCG17},
computing Karp-Rabin fingerprints \cite{BilleGCSVV17}, computing
runs, squares, and palindromes \cite{IMSIBTNS15}, real-time traversal \cite{GasieniecKPS05,LohreyMR18} and range-minimum queries~\cite{GaJoMoWe19}.
In all these applications we either improve existing results or significantly simplify existing proofs by 
replacing $\depth(\calG)$ by $\mathcal{O}(\log n)$ in time/space bounds.


Let us say a few words over the underlying computational model in Theorem~\ref{thm-balance-SSLP}.
Our balancing procedure involves (simple) arithmetic on lengths, i.e., numbers of order $n$.
Thus the linear running time can be achieved assuming
that machine words have $\Omega(\log n)$ bits.
Otherwise the running time increases by a~multiplicative $\log n$ factor.
Note that such an~assumption is realistic and standard in the field since machine words of bit length $\Omega(\log n)$
are needed, say, for indexing positions in the represented string.
On the other hand, our procedure works in the pointer model regime.

\paragraph{Balancing forest straight-line programs and top dags}
Grammar-based compression has been generalized from strings to ordered ranked node-labelled trees.
In fact, the representation of a~tree $t$ by its smallest directed acyclic graph (DAG) is a~form of 
grammar-based tree compression. This DAG is obtained by merging nodes where the same subtree of $t$ is rooted. 
It can be seen as a~regular tree grammar that produces only $t$. A~drawback of DAG-compression is that the size 
of the DAG is lower-bounded by the height of the tree $t$. Hence, for deep narrow trees (like for instance 
caterpillar trees), the DAG-representation cannot achieve good compression. This can be overcome by representing a~
tree $t$ by a~linear context-free tree grammar that produces only $t$.
Such grammars are also known as {\em tree straight-line programs} in the case of ranked trees
\cite{BuLoMa07,Lohrey15dlt,LohreyMM13} and {\em forest straight-line programs} in the case of 
unranked trees \cite{GLMRS18}. The~latter are tightly related to {\em top dags} \cite{BilleGLW15,BilleFG17,DudekG18,GLMRS18,Hubschle-Schneider15},
which are another tree compression formalism, also akin to grammars.
Our balancing technique works similarly for those types of compression:

\begin{thm} \label{thm-balancing-TSLP}
Given a~top dag\,/\,forest straight-line program\,/\,tree straight-line program $\calG$
producing the tree $t$ 
one can compute in time $\mathcal{O}(|\calG|)$ a~top dag\,/\,forest straight-line program\,/\,tree straight-line program $\mathcal{H}$ for $t$
of size $\mathcal{O}(|\calG|)$ and depth $\mathcal{O}(\log |t|)$.
\end{thm}

For top dags, this solves an~open problem from \cite{BilleGLW15},
where it was proved that from a~tree $t$ of size $n$, whose minimal DAG has size $m$ (measured in number
of edges in the DAG), one can 
construct in linear time a~top dag for $t$ of size $\mathcal{O}(m \cdot \log n)$ and depth 
$\mathcal{O}(\log n)$.
It remained open whether one can get rid of the factor $\log n$ in the size bound.
For the specific top dag constructed in \cite{BilleGLW15}, it was shown in \cite{BilleFG17} that the factor $\log n$
in the size bound $\mathcal{O}(m \cdot \log n)$ cannot be avoided. On the other hand, our results
yield another top dag of size $\mathcal{O}(m)$ and depth $\mathcal{O}(\log n)$. To see this note that one
can easily convert the minimal DAG of $t$ into a~top dag of roughly the same size, which can then be balanced.
This also gives an~alternative proof of a~result from \cite{DudekG18}, according to which
one can construct in linear time a~top dag of size $\mathcal{O}(n / \log_\sigma n)$ and depth $\mathcal{O}(\log n)$
for a~given tree of size $n$ containing $\sigma$ many different node labels.


\paragraph{Balancing circuits over algebras}
Our balancing results for SSLPs, top dags, forests straight-line programs
and tree straight-line programs are all instances of a~general balancing
result that applies to a~large class of circuits over algebraic structures.
To see the connection between circuits and straight-line programs,
consider SSLPs as an~example. An~SSLP is the same thing as a~bounded fan-in circuit over a~free monoid.
The~circuit gates compute the concatenation of their inputs and correspond to the variables of the 
SSLP. In general, for any algebra one can define straight-line programs,
which coincide with the classic notion of a~circuit.

The~definition of a~class of algebras,
to which our general balancing technique applies,
uses {\em unary linear term functions},
which were also used for instance in the context of efficient parallel evaluation 
of expression trees~\cite{MillerT97}. 
Fix an~algebra $\A$ (a set together with finitely many operations of possibly different arities).
For some of our applications we have to allow multi-sorted algebras that have several carrier sets (think 
for instance of a~vector space, where the two carrier sets are an~abelian group and a~field of scalars).
A~unary linear term function is a~unary function on $\A$ that is computed by a~term (or algebraic expression) that contains
a single variable $x$ (which stands for the function argument) and, moreover, $x$ occurs exactly once
in the term. For instance, a~unary linear term function over a~commutative ring is of the form $x \mapsto ax+b$
for ring elements $a,b$. A~{\em subsumption base} for an~algebra $\A$ is, roughly speaking, a~finite set $C(\A)$ of unary linear term
functions that are described by terms with parameters such that every unary linear term function can be obtained from one of the 
terms in $C(\A)$ by instantiating the parameters. In the above example for a~commutative ring the set $C(\A)$ 
consists of the single term $ax+b$, where $a$ and $b$ are the parameters.

Our general balancing result needs one more concept, namely the {\em unfolded size} of a~circuit $\mathcal{G}$.
It can be conveniently defined as follows: we replace in $\mathcal{G}$ every input gate by the number $1$, and 
we replace every internal gate by an~addition gate. The~unfolded size of $\mathcal{G}$ is the value of this
additive circuit.
In other words, this is the size of the tree obtained by unravelling $\mathcal{G}$ into a~tree.
Note that the size of this unfolding can be exponential in the circuit size.
Now we can state the general balancing result in a slightly informal way (the precise statement can be found in Theorem~\ref{cor-balance-dag}):

\begin{thm}[informal statement] \label{thm-general-balancing}
Let $\A$ be a~multi-sorted algebra with a~finite number of operations (of arbitrary arity) 
such that $\A$ has a~finite subsumption base.
Given a~circuit $\calG$ over $\A$ whose unfolded size is $n$,
one can compute in time $\mathcal{O}(|\calG|)$
a circuit $\mathcal{H}$ evaluating to the same element of $\A$
such that $|\mathcal{H}| \in \mathcal{O}(|\calG|)$ and $\depth(\mathcal{H}) \in \mathcal{O}(\log n)$.
\end{thm}

Theorems~\ref{thm-balance-SSLP} and \ref{thm-balancing-TSLP} are immediate corollaries of Theorem~\ref{thm-general-balancing}.
Theorem~\ref{thm-general-balancing} can be also applied to not necessarily commutative
semirings, as every semiring has a~finite subsumption base.
Hence, for every semiring circuit one can reduce with a~linear size blow-up the depth to 
$\mathcal{O}(\log n)$, where $n$ is the size of the circuit unfolding. 

Note that in the depth bound $\mathcal{O}(\log n)$ in our balancing result for string straight-line programs
(Theorem~\ref{thm-balance-SSLP}), $n$ refers to the length of the produced string. A~string straight-line program
can be viewed as a~circuit for a~non-commutative semiring circuit that produces a~single monomial (the symbols in the string
correspond to the non-commuting variables). If one considers arbitrary circuits over non-commutative semirings (that produce
a sum of more than one monomial), depth reduction is not possible in general by a~result of 
Kosaraju \cite{Kosaraju90}. For circuits over commutative semirings depth reduction is possible 
by a~seminal result of Valiant, Skyum, Berkowitz and Rackoff 
\cite{ValiantSBR83}: for any commutative semiring, every circuit of size $m$ and formal 
degree $d$ can be transformed into an~
equivalent circuit of depth $\mathcal{O}(\log m \log d)$ and size polynomial in $m$ and $d$. This result led to many further investigations
on depth reduction for bounded degree circuits over various classes of commutative as well as non-commutative semirings \cite{AllenderJMV98}. 
If one drops the restriction to bounded degree circuits, then depth reduction gets even harder. For general Boolean circuits, the best known result states that every Boolean circuit of size 
$m$ is equivalent to a~Boolean circuit of depth $\mathcal{O}(m / \log m)$ \cite{PatersonV76}.

\paragraph{Proof strategy}
The~proof of Theorem~\ref{thm-balance-SSLP}
consists of two main steps (the general result Theorem~\ref{thm-general-balancing}
is shown similarly). Take an~SSLP $\calG$ for the string $s$ of length $n$ and let 
$m$ be the size of $\calG$. 
We consider the derivation tree $t$ for $\calG$; it has size $\mathcal{O}(n)$.
The~SSLP $\calG$ can be viewed as a~DAG for $t$ of size $m$.
We decompose this DAG into node-disjoint paths such that each path from the root to a~leaf intersects $\mathcal{O}(\log n)$ paths from the decomposition
(Section~\ref{sec-centroid}).
Each path from the decomposition is then viewed as a~string of integer-weighted symbols,
where the weights 
are the lengths of the strings derived from nodes that branch off from the path.
For this weighted string we construct an~SSLP of linear size that produces all suffixes of the path 
in a~weight-balanced way (Section~\ref{sec-suffixes}).
Plugging these SSLPs together yields the final balanced SSLP.

Some of the concepts of our construction can be traced back to the area of parallel algorithms:
the path decomposition for DAGs from Section~\ref{sec-centroid} is related
to the centroid decomposition of trees~\cite{ColeV88}, where it is the key technique in several 
parallel algorithms on trees. Moreover, the SSLP of linear size that produces all suffixes of a~weighted string
with (Section~\ref{sec-suffixes}) can be seen as a~weight-balanced version of the 
optimal prefix sum algorithm.

For the general result Theorem~\ref{thm-general-balancing} we need another ingredient:
when the above construction is used for circuits over algebras,
the corresponding procedure produces a~tree straight-line program for the unfolding of the circuit.
We show that if the underlying algebra $\A$ has a~finite subsumption base,
then one can compute from a~tree straight-line program
an equivalent circuit over $\A$. 
Moreover, the size and depth of this circuit 
are linearly bounded in the size and depth of the tree straight-line program.
This construction was used before for the special cases of semirings
and regular expressions~\cite{GHJLN17,GL18}.

\section{Part I: Balancing of string straight-line programs} \label{sec-part-I}

The goal of the first part of the paper is to prove Theorem~\ref{thm-balance-SSLP}.
This result can be also derived from our general balancing theorem (Theorem~\ref{thm-general-balancing}), which
will be shown in the second part of the paper (Section~\ref{sec-part-II}). The techniques that we introduce in part I
will be also needed in Section~\ref{sec-part-II}. Moreover, we believe that it helps the reader to first
see the simpler balancing procedure for string straight-line programs before going into the details of the general
balancing result. Finally, the reader who is only interested in SSLP balancing can ignore part II completely.

We start with the afore-mentioned new decomposition technique for DAGs that we call 
symmetric centroid decomposition. The technical heart of our string balancing procedure is the  
linear-size SSLP that produces all suffixes of the path 
in a weight-balanced way (Section~\ref{sec-suffixes}).
Section~\ref{sec-suffixes} concludes the proof of Theorem~\ref{thm-balance-SSLP}, and Section~\ref{sec-applications} presents applications.

\subsection{The symmetric centroid decomposition of a DAG} \label{sec-centroid}

We start with a new decomposition of a DAG (directed acyclic graph) into disjoint paths. We believe that this decomposition might have
further applications. For trees, several decompositions into disjoint paths with the additional property that every path
from the root to a leaf only intersects a logarithmic number of paths from the decomposition exist. Examples are 
the heavy path decomposition \cite{HarelT84} and centroid decomposition \cite{ColeV88}. These decompositions can be also defined
for DAGs but a technical problem is that the resulting paths are no longer disjoint and form, in general, a subforest
of the DAG, see e.g.~\cite{BilleLRSSW15}. 

\definecolor{dgreen}{rgb}{0,0.6,0}
\newcommand{\dgreen}{\color{dgreen}}

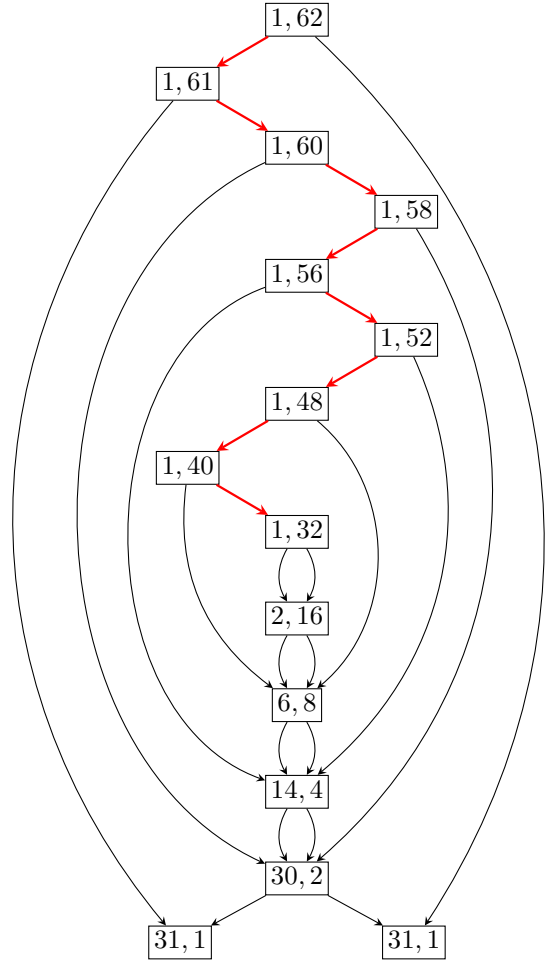
\begin{figure}[t]
		\centering
		\tikzstyle{lts} = [->, >=stealth]
		\tikzstyle{state} = [draw, inner sep = .7mm]
		
		\scalebox{1}{
			\begin{tikzpicture}[lts]
			
			\node [state] (0) {$1,62$};
			\node [state, below left = .4cm and .6cm of 0] (1) {$1,61$};
			\node [state, below right = .4cm and .6cm of 1] (2) {$1,60$};
			\node [state, below right = .4cm and .6cm of 2] (3) {$1,58$};
			\node [state, below left = .4cm and .6cm of 3] (4) {$1,56$};
			\node [state, below right = .4cm and .6cm of 4] (5) {$1,52$};
			\node [state, below left = .4cm and .6cm of 5] (6) {$1,48$};
			\node [state, below left = .4cm and .6cm of 6] (7) {$1,40$};
			\node [state, below right = .4cm and .6cm of 7] (8) {$1,32$};
			
			\node [state, below = .7cm of 8] (9) {$2,16$};
			\node [state, below = .7cm of 9] (10) {$6,8$};
			\node [state, below = .7cm of 10] (11) {$14,4$};					
			\node [state, below =  .7cm of 11] (12) {$30,2$};
			
			\node [state, below left = .4cm and .7cm of 12] (13a) {$31,1$};
			\node [state, below right = .4cm and .7cm of 12] (13b) {$31,1$};
			
			\draw (0) [line width = .3mm,red] to node[above]{} (1);
			\draw (1) [line width = .3mm,red] to node[above]{} (2);
			\draw (2) [line width = .3mm,red] to node[above]{} (3);
			\draw (3) [line width = .3mm,red] to node[above]{} (4);
			\draw (4) [line width = .3mm,red] to node[above]{} (5);
			\draw (5) [line width = .3mm,red] to node[above]{} (6);
			\draw (6) [line width = .3mm,red]  to node[above]{} (7);
			\draw (7) [line width = .3mm,red] to node[above]{} (8);
			\draw (8) to[bend right=30] node[left=-.5mm]{} (9);
			\draw (8) to[bend left=30] node[right=-.5mm]{} (9);
			\draw (9) to[bend right=30] node[left=-.5mm]{} (10);
			\draw (9) to[bend left=30] node[right=-.5mm]{} (10);
			\draw (10) to[bend right=30] node[left=-.5mm]{} (11);
			\draw (10) to[bend left=30] node[right=-.5mm]{} (11);
			\draw (11) to[bend right=30] node[left=-.5mm]{} (12);
			\draw (11) to[bend left=30] node[right=-.5mm]{} (12);
			
			\draw (12) to node[above]{} (13a);
			\draw (12) to node[above]{} (13b);
			
			\draw (0) to[bend left=40] node[right=-.5mm]{} (13b);
			\draw (1) to[bend right=40] node[left=-.5mm]{} (13a);
			\draw (2) to[bend right=65] node[left=-.5mm]{} (12);
			\draw (3) to[bend left=40] node[right=-.5mm]{} (12);
			\draw (4) to[bend right=70] node[left=-.5mm]{} (11);
			\draw (5) to[bend left=37] node[right=-.5mm]{} (11);
			\draw (6) to[bend left=50] node[right=-.5mm]{} (10);
			\draw (7) to[bend right=30] node[left=-.5mm]{} (10);
			\end{tikzpicture}
		}
		\caption{A DAG and its symmetric centroid decomposition.}
		\label{fig:scp}
	\end{figure}
	
Our new decomposition can be seen as a symmetric form of the centroid decomposition of \cite{ColeV88}.
Consider a DAG $\D = (V,E)$ with node set $V$ and the set of multi-edges $E$, i.e.,
$E$ is a finite subset of $V \times \N \times V$ 
such that $(u, d, v) \in E$ implies that for every $1 \le i<d$ there exists $v' \in V$ with $(u,i,v') \in E$.
Intuitively, $(u,d,v)$ is the $d$-th outgoing edge of $u$. We assume 
that there is a single root node $r \in V$, i.e., $r$ is the unique node with no incoming edges. Hence, all
nodes are reachable from $r$. A path from $u \in V$ to $v \in V$ is a sequence of edges
$(v_0, d_1, v_1), (v_1, d_2, v_2), \ldots, (v_{p-1}, d_p, v_p)$ where $u = v_0$ and $v = v_p$. 
We also allow the empty path from $u$ to $u$. 
With $\pi(u,v)$ we denote the number of paths from $u$ to $v$, and for $V' \subseteq V$ let
$\pi(u,V') = \sum_{v \in V'} \pi(u,v)$.
Let $W \subseteq V$ be the set of sink nodes of $\D$, i.e., those nodes without outgoing edges, 
and let $n(\D) = \pi(r,W)$. This is the number of leaves in the tree obtained by unfolding  $\D$ into a tree.
With a node $v \in V$ we assign the pair
$\lambda_{\D}(v) = (\lfloor \log_2 \pi(r,v) \rfloor,  \lfloor \log_2 \pi(v,W) \rfloor)$.
If $\lambda_{\D}(v) = (k,\ell)$, then 
$k, \ell \leq \lfloor \log_2 n(\D) \rfloor$ because $\pi(r,v)$ and $\pi(v,W)$ are both bounded by $n(\D)$.
Let us now define the edge set
$E_{\text{scd}}(\D)$ (``scd'' stands for symmetric centroid decomposition) as
$E_{\text{scd}}(\D) = \{ (u,i,v) \in E \mid \lambda_{\D}(u) = \lambda_{\D}(v) \}$.

\begin{ex}
\label{ex:scd}
Figure~\ref{fig:scp} shows the symmetric centroid decomposition of a DAG. The numbers in a node $v$
are the values $\pi(r,v)$ and $\pi(v,W)$ where $r$ is the root and $W$ consists of the two sink nodes. 
Edges that belong to a symmetric centroid path are drawn in red. Note that the 9 topmost nodes
form a symmetric centroid path since $\lfloor \log_2 \pi(r,v) \rfloor = 0$ and $\lfloor \log_2 \pi(v,W) \rfloor = 5$
for each of these nodes.
In this example the symmetric centroid decomposition consists of one 
path of length 8 (number of edges); all other nodes form 
symmetric centroid paths of length zero.
\end{ex}

\begin{lem}   \label{lem-decomposition}
Let $\D = (V,E)$ be a DAG with $n = n(\D)$.
Then every node has at most one outgoing and at most one incoming edge from $E_{\text{scd}}(\D)$.
Furthermore, every path from the root $r$ to a sink node contains at most $2  \log_2 n$ edges
that do not belong to $E_{\text{scd}}(\D)$.
\end{lem}

\begin{proof}
Consider a node $v\in V$ with two different outgoing edges $(u,i,v), (u,j,w) \in E_{\text{scd}}(\D)$.
Hence, $\lambda(u) = \lambda(v) = \lambda(w)$. Let $\lambda(u) = (k,\ell)$.
If $W$ is the set of sinks, we get 
$\pi(u,W) \geq \pi(v,W) + \pi(w,W)$ (since we consider paths of multi-edges, this inequality also holds for $v=w$). 
W.l.o.g.~assume that $\pi(w,W) \ge \pi(v,W)$ and thus
$\pi(u,W) \geq 2 \pi(v,W)$.
We get 
$$
\lfloor \log_2 \pi(u,W) \rfloor  \ge 1+\lfloor \log_2 \pi(v,W) \rfloor   = 1+\lfloor \log_2 \pi(u,W)\rfloor,
$$
where the last equality follows from $\lambda(u) = \lambda(v)$. This is a contradiction and proves
the claim for outgoing edges. Incoming edges are treated similarly,
this time using $\pi(r,v)$.

For the second claim of the Lemma, consider a 
path
\[
(v_0, d_1, v_1), (v_1, d_2, v_2), \ldots, (v_{p-1}, d_p, v_p),
\]
where $v_0$ is the root
and $v_p$ is a sink. Let $\lambda(v_i) = (k_i,\ell_i)$. We must have $k_i \le k_{i+1}$ and 
$\ell_i \ge \ell_{i+1}$ for all $0 \le i \le p-1$. Moreover, $k_0 = \ell_p = 0$ and $\ell_0, k_p \le \lfloor \log_2 n \rfloor$.
Consider now an edge $(v_i, d_i, v_{i+1}) \in E \setminus E_{\text{scd}}(\D)$.
Since $\lambda(v_i) \neq \lambda(v_{i+1})$, we have $k_i < k_{ji+1}$  or $\ell_i > \ell_{i+1}$.
Hence, there can be at most $2 \lfloor \log_2 n \rfloor  \le 2 \log_2 n$ edges from $E \setminus E_{\text{scd}}(\D)$
on the path.
\end{proof}
Lemma~\ref{lem-decomposition} implies that the subgraph $(V, E_{\text{scd}}(\D))$
is a disjoint union of possibly empty paths, called {\em symmetric centroid paths} of $\D$.
It is straight-forward to compute the edge set $E_{\text{scd}}(\D)$ in time $\mathcal{O}(|\D|)$,
where $|\D|$ is defined as the number of edges of the DAG: 
By traversing $\D$ in both directions (from the root to the sinks and from the sinks to the root)
one can compute all pairs $\lambda(v)$ for $v \in V$ in linear time.

One can use Lemma~\ref{lem-decomposition} in order to simplify the original proof of Theorem~\ref{cor-random-access}
from \cite{BilleLRSSW15}: in \cite{BilleLRSSW15}, the authors use the heavy-path decomposition of the derivation tree of an 
SSLP. In the SSLP (viewed
as a DAG that defines the derivation tree), these heavy paths lead to a forest, called the heavy path forest~\cite{BilleLRSSW15}. 
The important property used in \cite{BilleLRSSW15} is the fact that any path from the root of the DAG to a sink node contains only
$\mathcal{O}(\log n)$ edges that do not belong to a heavy path, where $n$ is the length of string produced by the SSLP. 
Using Lemma~\ref{lem-decomposition},
one can replace this heavy path forest by the decomposition into symmetric centroid paths. The fact that the latter is a disjoint union
of paths in the DAG simplifies the technical details in \cite{BilleLRSSW15} a lot. 
On the other hand, Theorem~\ref{cor-random-access} follows directly from 
Theorem~\ref{thm-balance-SSLP}, see Section~\ref{sec-applications}.

\subsection{Straight-line programs and suffixes of weighted strings} \label{sec-suffixes}

Given an alphabet of symbols $\Sigma$, $\Sigma^*$ denotes the set of all finite words
over the alphabet $\Sigma$, including the empty word $\varepsilon$. The set of non-empty
words is denoted by $\Sigma^+ = \Sigma^* \setminus \{\varepsilon\}$.
The length of a word $w$ is denoted with $|w|$.

 Let $\Sigma$ be a finite alphabet of terminal symbols.
A string straight-line program (SSLP for short) over the alphabet
$\Sigma$ is a triple $\calG = (\V,\rho,S)$, where $\V$ is a finite set of variables,
$S \in \V$ is the start variable, and $\rho \colon \V \to (\Sigma \cup \V)^*$ 
(the right-hand side mapping) has the property that the binary relation
$E(\calG) = \{ (X,Y) \in \V \times \V \colon Y \text{ occurs in } \rho(X) \}$ is acyclic.
This allows to define for every variable $X\in \V$ a string $\valXG{\calG}{X}$ 
as follows: if $\rho(X) = u_0 X_1 u_1 X_2 \cdots u_{n-1} X_n u_n$
with $u_0, u_1, \ldots, u_n \in \Sigma^*$ and $X_1, \ldots, X_n \in \V$ then
$\valXG{\calG}{X} = u_0 \valXG{\calG}{X_1} u_1 \valXG{\calG}{X_2} \cdots u_{n-1}
\valXG{\calG}{X_n} u_n$. We omit the subscript $\calG$ if $\calG$ is clear from the context. 
Finally, we define $\valX{\calG} = \valX{S}$. 

An SSLP $\calG$ can be seen as a context-free grammar that produces 
the single string $\valX{\calG}$. Quite often, one assumes that all right-hand
sides $\rho(X)$ are from $\Sigma \cup \V\V$. This corresponds to the Chomsky
normal form. For every SSLP $\calG$ with $\valX{\calG} \neq \varepsilon$ 
one can construct in linear time an equivalent SSLP in Chomsky normal form
by replacing every right-hand side by a balanced binary derivation tree.

Fix an SSLP $\calG = (\V,\rho,S)$. We define the size $|\calG|$ of $\calG$ as 
$\sum_{X \in \V} |\rho(X)|$. Let $d$ be the length of a longest path in the 
DAG $(\V,E(\calG))$ and $r = \max\{ |\rho(X)| \colon X \in \V\}$. We define 
the depth of $\calG$ as $\depth(\calG) = d \cdot  \lceil \log_2 r\rceil$. These definitions ensure
that depth and size only increase by fixed constants 
when an SSLP is transformed into Chomsky normal form. Note that for an SSLP in Chomsky normal form,
the definition of the depth simplifies to $\depth(\calG) = d$.

A {\em weighted string} is a string $s \in \Sigma^*$ equipped with a weight function
$\| \cdot \| \colon \Sigma \to \N \setminus \{0\}$, which is extended to a homomorphism
$\| \cdot \| \colon \Sigma^* \to \N$ by $\| a_1 a_2 \cdots a_n \| = \sum_{i=1}^n \| a_i \|$.
If $X$ is a variable in an SSLP $\calG$, we also write $\|X\|$ for the weight of the string $\valXG{\calG}{X}$
derived from $X$.
Moreover, when we speak of suffixes of a string, we always mean non-empty suffixes.

\begin{prop}
	 \label{prop-all-suffixes}
	For every non-empty weighted string $s$ of length $n$ one can construct in linear time an SSLP $\calG$ 
	with the following properties:
	\begin{itemize}
        \item $\calG$ contains at most $3n$ variables,
        \item all right-hand sides of $\calG$ have length at most 4,
        \item $\calG$ contains \emph{suffix variables} $S_1, \dots, S_n$ producing all suffixes of $s$, and
        \item every path from $S_i$ to some terminal symbol $a$ in the derivation tree of $\calG$
	has length at most $3 + 2(\log_2 \|S_i\|  - \log_2 \|a\|)$.
        \end{itemize}
\end{prop}

\begin{proof}
First, the presented algorithm never uses the fact that some letters of $s$ may be equal.
	Thus it is more convenient to assume that letters in $s$ are pairwise different---in this way
	the path from a variable $S_i$ to a terminal symbol $a$ in the last condition
	is defined uniquely.
	
	For the sake of an inductive proof, the constructed SSLP
	will satisfy a slightly stronger and more technical variant of the last condition:
	every path from $S_i$ to some terminal symbol $a$ in the derivation tree of $\calG$
		has length at most $1 + 2(\lceil \log_2 \|S_i\| \rceil  - \log_2 \|a\|)$.
	The trivial estimation $\lceil \log_2 \|S_i\| \rceil \le 1 + \log_2 \|S_i\|$ then yields the announced variant.
 
    We first show how to construct $\calG$ with the desired properties and then prove that the 
    construction can be done in linear time.

	The case $n = 1$ is trivial. Now assume that $n \ge 2$ and
	let
	\[
		s = a_1 \cdots a_k \, c \, b_1 \cdots b_m
	\]
	where $c b_1 \cdots b_m$ is the shortest suffix of $s$ such that
	$\roundlog{cb_1 \cdots b_m} = \roundlog{s}$.
	Clearly such a suffix exists (in the extreme cases it is the entire string $s$ or a single letter).
	Note that 
         \begin{equation} \label{eq-suffix1}
         \roundlog{cb_1 \cdots b_m} = \roundlog{a_i \cdots a_k cb_1 \cdots b_m}
         \end{equation}
         for $1 \le i \le k+1$. Moreover, the following inequalities hold:
	\begin{align}
	\label{eq:weight_suff}
	\roundlog{cb_1 \cdots b_m} &\ge  \roundlog{b_1 \cdots b_m} + 1 \\
	\label{eq:weight_pref}
	\roundlog{cb_1 \cdots b_m} &\ge  \roundlog{a_1 \cdots a_k} + 1	
	\end{align}
	(here, we define $\log_2(0) = -\infty$).
	The former is clear from the definition of $cb_1 \cdots b_m$,
	as $b_1 \cdots b_m$ satisfies $\roundlog{b_1 \cdots b_m} < \roundlog{s} = \roundlog{cb_1 \cdots b_m}$.
	If \eqref{eq:weight_pref} does not hold then both 
$a_1\cdots a_k$ and $cb_1 \cdots b_m$ have weights strictly more than $2^{\roundlog{s}-1}$
and so their concatenation $s$ has weight strictly more than $2^{\roundlog{s}} \geq \|s\|$, which is a
contradiction.

	Recall that the symbols $a_1, \ldots, a_k,c,b_1, \ldots, b_m$ 
	are pairwise different by the convention from the first paragraph of the proof.

	For $b_1\cdots b_m$ we make a recursive call (if $m=0$ we do nothing at this step) and include the produced SSLP
	in the output SSLP $\calG$.
	Let  $V_{1}, V_{2}, \dots, V_{m}$ be the variables such that 
	\[
			\valXG{\calG}{V_{i}} = b_{i} \cdots b_m .
	\]
	By the inductive assumption, every path $V_i \xrightarrow{*} b_j$
	in the derivation tree has length at most
	\[
	1 + 2\lceil\log_2 \|V_i\|\rceil - 2\log_2 \|a_j\| .
	\]
	Add a variable $V_0$ with right-hand side $cV_1$ (or $c$ if $m=0$),
	which derives the suffix $cb_1\cdots b_m$.
	The path from $V_0$ to $c$ in the derivation tree has length $1$, which is fine,
	and the path $V_0 \xrightarrow{*} a_j$ is one larger than the path $V_1 \xrightarrow{*}a_j$  and hence
	has length at most 
	\[
	1 + 1 + 2\lceil\log_2 \|V_1\|\rceil - 2\log_2 \|a_j\|
	\le 2\lceil\log_2 \|V_0\|\rceil - 2\log_2 \|a_j\| ,
	\]
	as $1 + \roundlog{V_1} \le \roundlog{V_0}$ by \eqref{eq:weight_suff}.

	Next we decompose the prefix $a_1 \cdots a_k$ into $\lfloor k/2 \rfloor$
	many blocks of length two and, when $k$ is odd, one block of length $1$.
	We add to the output SSLP $\calG$ new variables $X_1, \dots, X_{\lfloor k/2 \rfloor}$ and define
	their right-hand sides by
	\[
		\rho(X_i) = a_{2i-1} a_{2i}.
	\]
	The number of variables in $\calG$ is $\lfloor k/2 \rfloor$.
	For ease of presentation, when $k$ is odd, define $X_{\lceil k/2 \rceil} = a_k$,
	this is not a new variable, rather just a notational convention to streamline the presentation.
	Note that for even $k$ we have $\lceil k/2 \rceil = \lfloor k/2 \rfloor$ and
	in this case $X_{\lceil k/2 \rceil}$ is already defined.
	Viewing $X_1\cdots X_{\lceil k/2 \rceil}$ as a weighted string of length $\lceil k/2 \rceil$
	over the alphabet $\{X_1, \ldots, X_{\lceil k/2 \rceil}\}$,
	we obtain inductively an SSLP $\calG_X$ with at most $3 \lceil k/2 \rceil$ variables 
	and right-hand sides of length at most 4 (if $k=0$ we do nothing at this step).
	Moreover, $\calG_X$ contains variables $U_{1}, U_{2}, \dots, U_{\lceil k/2 \rceil}$
	with
	\[
		\valXG{\calG_X}{U_{i}} = X_i X_{i+1} \cdots X_{\lceil k/2 \rceil}
	\]
	such that any path of the form $U_{i} \xrightarrow{*} X_j$
	in the derivation tree of $\calG_X$ has length at most
	\[
		1 + 2\lceil\log_2 \|U_i\|\rceil - 2\log_2 \|X_j\| .
	\]
	By adding all variables and right-hand side definitions from $\calG_X$ to $\calG$
	(where all symbols $X_i$ are variables, except $X_{\lceil k/2 \rceil}$ when $k$ is odd,
	in which case $X_{\lceil k/2 \rceil} = a_k$)
	we obtain
	\[
		\valXG{\calG}{U_{i}} = a_{2i-1} a_{2i} \cdots a_k
	\]
	for all $1 \le i \le {\lceil k/2 \rceil}$.
	Any path $U_{i} \xrightarrow{*} a_j$ in the derivation tree of $\calG$
	has length at most
\begin{equation}
\label{eq:U_aj_estimation}
	2 + 2\lceil\log_2 \|U_i\|\rceil - 2\log_2 \|a_j\| .
\end{equation}
	Now, every suffix of $s$ that includes some letter of $a_1 \cdots a_k$
	(note that we already have variables for all other suffixes)
	can be defined by a right-hand side of the form
	$U_{i} c  V_{1}$ or  $a_{2i-2} U_{i} c V_{1}$ ($U_{i} c$ or  $a_{2i-2} U_{i} c$ if $m=0$).
	As in the statement of the lemma, denote those variables by $S_1, \ldots, S_k$.
	Let us next verify the condition on the path lengths for derivations from those variables.
	All paths $S_i \xrightarrow{*} c$ have length one. 
	Now consider a path $S_i \xrightarrow{*} a_j$.
	If the path has length one then we are done.
	Otherwise, the path must be of the form $S_i \to U_{l} \xrightarrow{*} a_j$.
	Therefore, by~\eqref{eq:U_aj_estimation}
	the path length is at most
\begin{align*}
	3 + 2\roundlog{U_l} - 2 \plainlog{a_j} &\leq 3 + 2\roundlog{U_1} - 2 \plainlog{a_j}\\
	&\leq
	1 + 2\roundlog{cb_1 \cdots b_m} - 2 \plainlog{a_j}\\
	& = 
	1 + 2\roundlog{S_i} - 2 \plainlog{a_j},
\end{align*}	
	where the second inequality follows from~\eqref{eq:weight_pref}
	and the equality at the end follows from~\eqref{eq-suffix1}.

	Paths of the form $S_i \xrightarrow{*} b_j$ can be treated similarly:
	they are of the form $S_i \to V_1 \xrightarrow{*} b_j$,
	where the path $V_1 \xrightarrow{*} b_j$ is of length at most
	$1 +  2\roundlog{V_1} - 2 \plainlog{a_j}$ by the inductive assumption.
	Thus, the whole path is of length at most
        \begin{align*}
	2 + 2\roundlog{V_1} - 2 \plainlog{b_j}
	&\leq
	2\roundlog{cb_1 \cdots b_m} - 2 \plainlog{b_j}\\
	& = 
	2\roundlog{S_i} - 2 \plainlog{b_j},
        \end{align*}		
	which follows from~\eqref{eq:weight_suff} and~\eqref{eq-suffix1}.

	The SSLP $\calG$ consists of $\lfloor k/2 \rfloor$ variables $X_i$,
	$3(\lceil k/2 \rceil)$ variables from the recursive call for the weighted string
	$X_1 \cdots X_{\lceil k/2 \rceil}$, $3m = 3(n-k-1)$ variables from the recursive call for $b_1\cdots b_m$,
	and $1 + k$ new suffix variables for suffixes beginning at $a_1 \cdots  a_k c$
	(note that those beginning at $b_1\cdots b_m$ are taken care of by the recursive call).
	Therefore $\calG$ contains at most
	\begin{align*}
	\lfloor k/2 \rfloor + 3\lceil k/2 \rceil + 3(n-k-1) + 1 + k =
	3n + 2\lceil k/2 \rceil - k - 2 < 3n
	\end{align*}
	variables. Also note that all right-hand sides of $\calG$ have length at most four.

	It remains to show that the construction works in linear time.
	To this end we need a small trick:
	we assume that when the algorithm is called on $s$, we supply the algorithm with the value $\|s\|$.
	More formally, the main algorithm applied to a string $s$
	computes $\|s \|$ in linear time by going through $s$ and adding weights.
	Then it calls a subprocedure $\text{main}'(s, \|s\|)$, which performs the actions described above.
	To find the appropriate symbol $c$, $\text{main}'$ computes the weights of consecutive prefixes $s_1s_2\cdots s_i$,
	until it finds the first such that $\roundlog{s} > \lceil \log_2(\|s\| -\|s_1\cdots s_i\| )\rceil$.
	Then $k = i-1$ and so $a_1 \cdots a_k = s_1\cdots s_{i-1}$, $c = s_{i}$, $b_1\cdots b_m = s_{i+1}\cdots s_{|s|}$.
	Moreover, we can compute $\|a_1 \cdots a_k\|$ and $\|b_1 \cdots b_m\|$ for the recursive calls of $\text{main}'$ in constant time.

	Let $T(n)$ be the running time of $\text{main}'$ on a word of length $n$.
	Then all operations of $\text{main}'$, except the recursive calls,
	take at most $\alpha (k+1)$ time for some constant $\alpha \ge 1$,
	where $s$ is represented as $a_1 \cdots a_k c b_1\cdots b_m$.
	Thus $T(n)$ satisfies $T(1) = 1$ and 
    \[
    T(n) = T(\lceil k/2 \rceil) + T(n-k-1) + \alpha(k+1).
    \]
    We claim that $T(n) \leq 2 \alpha n$. 
	This is true for $n=1$ and inductively for $n \ge 2$ we get
\begin{align*}
	T(n)
		&\le
	2 \alpha (\lceil k/2 \rceil) + 2\alpha(n-k-1) + \alpha(k+1)\\
		&\le
	2 \alpha \frac{k+1}{2} + 2\alpha n -\alpha(k+1)\\
		&=
	2\alpha n .
\end{align*}
This concludes the proof of the lemma.
\end{proof}

\subsection{Proof of Theorem~\ref{thm-balance-SSLP}}\label{sec:balancing-of-sslps}

We now prove Theorem~\ref{thm-balance-SSLP}.
Let $\calG = (\V,\rho_\calG,S)$. W.l.o.g. we can assume that $\calG$ is in Chomsky normal form (the case that $\valX{G}=\varepsilon$
is trivial).
Note that the graph $(\V,E(\calG))$ is a directed acyclic graph (DAG). We can assume that every
variable is reachable from the start variable $S$.
Consider a~variable $X$ with $\rho_\calG(X) = YZ$. Then $X$ has the two outgoing edges
$(X,Y)$ and $(X,Z)$ in $(\V,E(\calG))$. We replace these two edges by the triples
$(X,1,Y)$ and $(X,2,Z)$. Hence, $\D := (\V,E(\calG))$ becomes a~DAG with multi-edges
(triples from $\V \times \{1,2\} \times \V$).
Figure~\ref{fig:scp2} shows the DAG $\D$ for an example SSLP; it is the same DAG as in Example~\ref{ex:scd}; see Figure~\ref{fig:scp}.
The right-hand sides for the two sink variables $X_{13}$ and $X_{14}$ are terminal
symbols (the concrete terminals are not relevant for us). The start variable $S$ is $X_0$. 

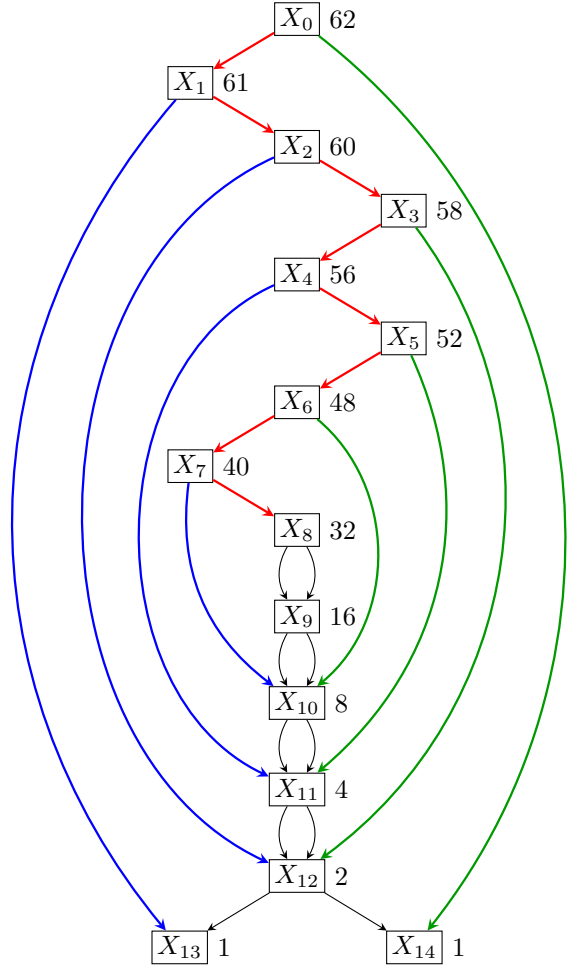
\begin{figure}[t]
		\centering
		\tikzstyle{lts} = [->, >=stealth]
		\tikzstyle{state} = [draw, inner sep = .7mm]
		
		\scalebox{1}{
			\begin{tikzpicture}[lts]
			\node [state, right = 7.25cm  of 0, label={right:62}] (0) {$X_0$};
			\node [state, below left = .4cm and .8cm of 0, label={right:61}] (1) {$X_1$};
			\node [state, below right = .4cm and .8cm of 1, label={right:60}] (2) {$X_2$};
			\node [state, below right = .4cm and .8cm of 2, label={right:58}] (3) {$X_3$};
			\node [state, below left = .4cm and .8cm of 3, label={right:56}] (4) {$X_4$};
			\node [state, below right = .4cm and .8cm of 4, label={right:52}] (5) {$X_5$};
			\node [state, below left = .4cm and .8cm of 5, label={right:48}] (6) {$X_6$};
			\node [state, below left = .4cm and .8cm of 6, label={right:40}] (7) {$X_7$};
			\node [state, below right = .4cm and .8cm of 7, label={right:32}] (8) {$X_8$};
			
			\node [state, below = .7cm of 8, label={right:16}] (9) {$X_9$};
			\node [state, below = .7cm of 9, label={right:8}] (10) {$X_{10}$};
			\node [state, below = .7cm of 10, label={right:4}] (11) {$X_{11}$};					
			\node [state, below =  .7cm of 11, label={right:2}] (12) {$X_{12}$};
			
			\node [state, below left = .5cm and .8cm of 12, label={right:1}] (13a) {$X_{13}$};
			\node [state, below right = .5cm and .8cm of 12, label={right:1}] (13b) {$X_{14}$};
			\draw (0) [red,line width = .3mm] to node[above]{} (1);
			\draw (1) [red,line width = .3mm] to node[above]{} (2);
			\draw (2) [red,line width = .3mm] to node[above,pos=.7]{} (3);
			\draw (3) [red,line width = .3mm] to node[above]{} (4);
			\draw (4) [red,line width = .3mm] to node[above,pos=.7]{} (5);
			\draw (5) [red,line width = .3mm] to node[above]{} (6);
			\draw (6) [red,line width = .3mm]  to node[above]{} (7);
			\draw (7) [red,line width = .3mm] to node[above,pos=.7]{} (8);
			\draw (8) to[bend right=30] node[left]{} (9);
			\draw (8) to[bend left=30] node[right]{} (9);
			\draw (9) to[bend right=30] node[left]{} (10);
			\draw (9) to[bend left=30] node[right]{} (10);
			\draw (10) to[bend right=30] node[left]{} (11);
			\draw (10) to[bend left=30] node[right]{} (11);
			\draw (11) to[bend right=30] node[left]{} (12);
			\draw (11) to[bend left=30] node[right]{} (12);
			
			\draw (12) to node[above,pos=.6]{} (13a);
			\draw (12) to node[above,pos=.6]{} (13b);
			
			\draw (0) [dgreen,line width = .3mm] to[bend left=45] node[right]{} (13b);
			\draw (1) [blue,line width = .3mm] to[bend right=40] node[left]{} (13a);
			\draw (2) [blue,line width = .3mm] to[bend right=65] node[left]{} (12);
			\draw (3) [dgreen,line width = .3mm] to[bend left=45] node[right]{} (12);
			\draw (4) [blue,line width = .3mm] to[bend right=65] node[left]{} (11);
			\draw (5) [dgreen,line width = .3mm] to[bend left=37] node[right]{} (11);
			\draw (6) [dgreen,line width = .3mm] to[bend left=50] node[right]{} (10);
			\draw (7) [blue,line width = .3mm] to[bend right=30] node[left]{} (10);

			\end{tikzpicture}
		}
		\caption{The DAG for an SSLP.}
		\label{fig:scp2}
	\end{figure}

We define for every $X \in \V$ the weight $\| X \|$ as the length of the string 
$\valXG{\calG}{X}$. Moreover, for a string $w = X_1 X_2 \cdots X_n$ we define 
the weight $\| w \| = \sum_{i=1}^n \| X_i \|$. Note that $\| S \|=n$ is the length of the derived string
$\valX{\calG}$ and that this also the value $n(\D)$ defined in  Section~\ref{sec-centroid}.

We compute in linear time the edges from symmetric centroid decomposition of the DAG $\D$, see~Lemma~\ref{lem-decomposition}.
In Figure~\ref{fig:scp2} these are the red edges. The weights $\| X_i \|$ of the variables
are written next to the corresponding nodes; these weights  can be found as the second components in Figure~\ref{fig:scp}.
Hence, we have $\| X_0 \|=62$, $\| X_1 \|=61$, $\| X_2 \|=60$, $\| X_3 \|=58$, etc. 
 
Consider a symmetric centroid path 
\begin{equation} \label{sym-centroid-path2}
(X_0, d_0, X_1), (X_1, d_1, X_2), \ldots, (X_{p-1}, d_{p-1}, X_p)
\end{equation}
in $\D$, where all $X_i$ belong to $\V$ and $d_i \in \{1,2\}$.
Thus, for all $0 \le i \le p-1$, the right-hand side of $X_i$ in $\calG$ has the form 
 $\rho_{\calG}(X_i) = X_{i+1}X'_{i+1}$ (if $d_i = 1$) or
 $\rho_{\calG}(X_i) = X'_{i+1}X_{i+1}$ (if $d_i = 2$)
for  some $X'_{i+1} \in \V$. Note that we can have $X'_i = X'_j$ for $i \neq j$.
The right-hand side $\rho_{\calG}(X_p)$ belongs to $\Sigma \cup \V\V$.
Note that the variables $X'_i$ ($1 \le i \le p$) and the variables in $\rho_{\calG}(X_p)$ (if they exist) 
belong to other symmetric centroid paths.  We will introduce $\mathcal{O}(p)$ many 
variables in the SSLP $\mathcal{H}$ to be constructed. Moreover, all right-hand sides
of $\mathcal{H}$ have length at most four.
By summing over all  symmetric centroid paths, this yields the size 
bound $\mathcal{O}(|\calG|)$ for $\mathcal{H}$.

We now define the right-hand sides of the variables $X_0, \ldots, X_p$ in $\mathcal{H}$.
We write $\rho_{\mathcal{H}}$ for the right-hand side mapping of $\mathcal{H}$.
For $X_p$ we set $\rho_{\mathcal{H}}(X_p) = \rho_{\calG}(X_p)$.
For the variables $X_0, \ldots, X_{p-1}$ we have to ``accelerate'' the derivation somehow
in order to get the depth bound $\mathcal{O}(\log n)$ at the end.
For this, we apply Proposition~\ref{prop-all-suffixes}.
Let $L_1 \cdots L_s$ be the subsequence obtained from $X'_1 X'_2 \cdots X'_p$
by keeping only those $X'_i$ with $d_i = 2$ and let 
$R_1 \cdots R_t$ be the subsequence obtained from the reversed sequence $X'_p X'_{p-1} \cdots X'_1$
by keeping only those $X'_i$ with $d_i = 1$.
Take for instance the red symmetric centroid path consisting of the nodes $X_0, X_2, \ldots, X_8$ (hence, $p=8$) from our
 running example in Figure~\ref{fig:scp2}.
 We have $L_1 \cdots L_s = X_{13} X_{12} X_{11} X_{10}$ (the target nodes of the blue edges) and
 $R_1 \cdots R_t = X_{10} X_{11} X_{12} X_{14}$ (the target nodes of the green edges).

Note that every string $\valX{X_i}$ ($0 \le i \le p-1$) can be derived in $\calG$ from a word
$w_\ell X_p w_r$, where $w_\ell$ is a suffix of $L_1 \cdots L_s$ and $w_r$ is a prefix of $R_1 \cdots R_t$.
For instance, $\valX{X_2}$ can be derived from $(X_{12} X_{11} X_{10}) X_8 (X_{10} X_{11} X_{12})$
in our running example, so $w_\ell = X_{12} X_{11} X_{10}$ and $w_r = X_{10} X_{11} X_{12}$.
We now apply Proposition~\ref{prop-all-suffixes} to the sequence $L_1 \cdots L_s$
in order to get an SSLP $\calG_\ell$ of size $\mathcal{O}(s) \le \mathcal{O}(p)$ that contains variables 
$S_1 \ldots, S_s$ for the non-empty suffixes of $L_1 \cdots L_s$. Moreover, 
every path from a variable $S_i$ to some $L_j$ in the derivation tree
has length at most $3 + 2\log_2 \|S_i\| - 2\log_2\|L_j\|$, 
where $\|S_i\|$ is the weight of $\valXG{\calG_\ell}{S_i}$.
Analogously, we obtain an 
SSLP $\calG_r$ of size $\mathcal{O}(t) \le \mathcal{O}(p)$ that contains variables 
$P_1 \ldots, P_t$ for the non-empty prefixes of $R_1 \cdots R_t$. Moreover, 
every path from a variable $P_i$ to some $R_j$ in the derivation tree
has length at most $3 + 2\log_2 \|P_i\| - 2\log_2\|R_j\|$.
We can then define every right-hand side $\rho_{\mathcal{H}}(X_i)$ as 
$S_j X_p P_k$, $X_p P_k$, $S_j X_p$, or $X_p$  for suitable $j$ and $k$. 
Moreover, we add all variables and right-hand side definitions of $\calG_\ell$
and  $\calG_r$ to $\mathcal{H}$.

We make the above construction for all symmetric centroid paths of the DAG $\D$.
This concludes the construction of $\mathcal{H}$. In our running example 
we set $\rho_{\mathcal{H}}(X_i) = \rho_{\calG}(X_i)$ for $8 \le i \le 14$.
Since we introduce $\mathcal{O}(p)$ 
many variables for every symmetric centroid path of length $p$ and 
all right-hand sides of $\mathcal{H}$ have length at most four,  we obtain the size 
bound $\mathcal{O}(|\calG|)$ for $\mathcal{H}$.

It remains to show that the depth of the SSLP $\mathcal{H}$ is $\mathcal{O}(\log n)$.
Let us first consider the symmetric centroid path \eqref{sym-centroid-path2} and a path in the derivation tree
of $\mathcal{H}$ from a variable $X_i$  ($0 \le i \le p$) to a variable $Y$, where
$Y$ is 
\begin{enumerate}[(a)]
\item a variable in $\rho_{\calG}(X_p) = \rho_{\mathcal{H}}(X_p)$ or
\item a variable $X'_j$ for some $i < j \le p$.
\end{enumerate}
In case (a), the path $X_i \xrightarrow{*} Y$ has length at most two.
In case (b) the path $X_i \xrightarrow{*} Y$ is of the form
$X_i \to S_k \xrightarrow{*} X'_j = Y$ or $X_i \to P_k \xrightarrow{*} X'_j = Y$.
Here, $S_k \xrightarrow{*} X'_j$ (resp., $P_k \xrightarrow{*} X'_j$) is a path 
in $\calG_\ell$ (resp., $\calG_r$) and therefore has length 
$3 + 2\log_2 \|S_k\| - 2\log_2\|Y\|$ (resp., $3 + 2\log_2 \|P_k\| - 2\log_2\|Y\|$).
In both cases, we can bound the length of the path $X_i \xrightarrow{*} Y$ by
$4 + 2\log_2 \|X_i\| - 2\log_2\|Y\|$. 

Consider a maximal path in the derivation tree of $\mathcal{H}$ that starts in the root $S$ and ends in a leaf.
We can factorize this path as
\begin{equation} \label{full-path2}
S = X_0 \xrightarrow{*} X_1 \xrightarrow{*} X_2 \xrightarrow{*} \cdots \xrightarrow{*} X_k
\end{equation}
where all variables $X_i$ belong to the original SSLP and every subpath $X_i  \xrightarrow{*} X_{i+1}$ 
is of the form $X_i \xrightarrow{*} Y$ considered in the previous paragraph.
The right-hand side of $X_k$ is a single symbol from $\Sigma$.
In the DAG $\D$ we have a corresponding path $X_i  \xrightarrow{*} X_{i+1}$, which is contained
in a single symmetric centroid path except for the last edge leading to $X_{i+1}$. 
By the above consideration, the length of the path \eqref{full-path2} is bounded by
$$
\sum_{i=0}^{k-1} (4 + 2\log_2 \| X_i\| - 2\log_2 \| X_{i+1}\|) \le 4k + 2\log_2 \|S\|
= 4k + 2\log_2 n.
$$
By the second claim of 
Lemma~\ref{lem-decomposition} we have $k \le 2\log_2 n$ which shows that the length of the path 
\eqref{full-path2} is bounded by $6\log_2 n$. \qed

\subsection{Applications of Theorem~\ref{thm-balance-SSLP}} \label{sec-applications}

There are several algorithmic applications of Theorem~\ref{thm-balance-SSLP}
with always the same idea:
let $\calG$ be an SSLP of size $m$ for a string $s$ of length $n$.
In many algorithms for SSLP-compressed strings the running time or space consumption depends
on $\depth(\calG)$, which can be $m$ in the worst case.
Theorem~\ref{thm-balance-SSLP} shows that we can replace $\depth(\calG)$ by $\mathcal{O}(\log n)$.
This is the best we can hope for since $\depth(\calG) \ge \Omega(\log n)$ for every SSLP $\calG$.
Moreover, SSLPs that are produced by practical grammar-based compressors (e.g., LZ78 or RePair) are 
in general unbalanced in the sense that $\depth(\calG) \ge \omega(\log n)$.

The time bounds in the following results refer to the RAM model, where arithmetic operations on numbers 
from the interval $[0,n]$ need time $\mathcal{O}(1)$. The size of a data structure is measured in the number
of words of bit length $\log_2 n$.

As a first application of Theorem~\ref{thm-balance-SSLP} we can present a very simple new proof of Theorem~\ref{cor-random-access}
(random access for grammar-compressed strings) based on the folklore random access algorithm that works in time $\mathcal{O}(\depth(\calG))$.

\begin{proof}[Proof of Theorem~\ref{cor-random-access}]
Using Theorem~\ref{thm-balance-SSLP} we compute in time $\mathcal{O}(m)$ 
an equivalent SSLP $\mathcal{H}$ for $s$ of size $\mathcal{O}(m)$ and depth $\mathcal{O}(\log n)$. 
By a single pass over $\mathcal{H}$ we compute for every variable $X$ of $\mathcal{H}$ 
the length of the word $\valX{X}$.
Using these lengths one can descend in the derivation tree $\valX{\mathcal{H}}$ from the root to the $i$-th leaf node (which is labelled
with the $i$-th symbol of $s$) in time 
$\mathcal{O}(\depth(\mathcal{H})) \le \mathcal{O}(\log n)$.
\end{proof}

\begin{rem}
It is easy to see that the balancing algorithm from 
Theorem~\ref{thm-balance-SSLP} can be implemented on a pointer machine, see
\cite{Tarjan83} for a discussion of the pointer machine model. This yields a pointer machine implementation
of the random access data structure from Theorem~\ref{cor-random-access}.  
In contrast, the random access data structure from \cite{BilleLRSSW15} needs the RAM model
(for the pointer machine model only preprocessing time and size $\mathcal{O}(m \cdot \alpha_k(m))$ for any fixed $k$, where 
$\alpha_k$ is the $k$-th inverse Ackermann function, is shown in \cite{BilleLRSSW15}).
On the other hand, recently, in \cite{BiGoGaLaWe19}, the $\mathcal{O}(m)$-space data structure from \cite{BilleLRSSW15}
has been modified so that it can be implemented on a pointer machine as well.
\end{rem}

Using fusion trees \cite{FredmanW93} one can improve the time bound in
Theorem~\ref{cor-random-access} to $\mathcal{O}(\log n/\log \log n)$ at the cost of an additional
factor of $\mathcal{O}(\log^{\epsilon} n)$ in the size bound.
The following result has been shown in \cite[Theorem~2]{BelazzouguiCPT15} under the assumption 
that the input SSLP has depth $\mathcal{O}(\log n)$. We can enforce this bound with 
Theorem~\ref{thm-balance-SSLP}.

\begin{cor} \label{cor-fusion}
Fix an arbitrary constant $\epsilon>0$.
From a given SSLP $\calG$ of size $m$ such that the string $s=\valX{\calG}$ has length $n$, 
one can construct in time $\mathcal{O}(m \cdot \log^{\epsilon}n)$ a 
data structure of size $\mathcal{O}(m \cdot \log^{\epsilon}n)$  that allows to answer random access 
queries in time $\mathcal{O}(\log n/\log \log n)$.
\end{cor}

\begin{proof}
The proof is exactly the same as for \cite[Theorem~2]{BelazzouguiCPT15}.  There, the authors have to assume 
that the input SSLP has depth $\mathcal{O}(\log n)$, which we can enforce by 
Theorem~\ref{thm-balance-SSLP}. Roughly speaking, the idea in \cite{BelazzouguiCPT15} is to reduce
the depth of the SSLP to $\mathcal{O}(\log n/\log \log n)$ by expanding right-hand sides to length
$\mathcal{O}(\log^{\epsilon} n)$. Then for each right-hand side a fusion tree is constructed, which allows
to spend constant time at each variable during the navigation to the $i$-th symbol. 

Let us also remark that the size bound for the computed data structure in \cite{BelazzouguiCPT15} is given in bits,
which yields $\mathcal{O}(m \cdot \log^{1+\epsilon}n)$ bits since numbers from $[0,n]$ have to 
be encoded with $\log_2 n$ bits.
\end{proof}

Given a string $s \in \Sigma^*$, a rank query gets a position $1 \le i \le |s|$ and a symbol $a\in\Sigma$ and 
returns the number of $a$'s in the prefix of $s$ of length $i$. A select query gets a symbol
$a\in \Sigma$ and returns the position of the $i$-th $a$ in $s$ (if it exists).

\begin{cor} \label{cor-rank/select}
Fix an arbitrary constant $\epsilon>0$.
From a given SSLP $\calG$ of size $m$ such that the string $s=\valX{\calG}$ has length $n$, 
one can construct in time $\mathcal{O}(m \cdot |\Sigma| \cdot \log^{\epsilon}n)$ a 
data structure of size $\mathcal{O}(m \cdot |\Sigma| \cdot \log^{\epsilon}n)$  that allows to answer rank and select
queries in time $\mathcal{O}(\log n/\log \log n)$.
\end{cor}

\begin{proof}
Again we follow the proof \cite[Theorem~2]{BelazzouguiCPT15} but first apply
Theorem~\ref{thm-balance-SSLP} in order to reduce the depth of the SSLP to
$\mathcal{O}(\log n)$.
\end{proof}

Corollary~\ref{cor-rank/select} improves \cite[Theorem~2]{BelazzouguiCPT15}, where
the query time is $\mathcal{O}(\log n)$ and the space is $\mathcal{O}(m \cdot |\Sigma| \cdot \log n)$.

Our balancing result also yields an improvement for the {\em compressed subsequence problem} \cite{BilleCG17}.
Bille et al. \cite{BilleCG17} present an algorithm based on a {\em labelled successor} data structure.
Given a string $s = a_1 \cdots a_n \in \Sigma^*$,
a labelled successor query gets a position $1 \le i \le n$ and a symbol $a\in\Sigma$ and 
returns the minimal position $j > i$ with $a_j = a$ (or rejects if it does not exist).
The following result is an improvement over \cite{BilleCG17}, where the authors present
two algorithms for the compressed subsequence problem: one with 
$\mathcal{O}(m + m \cdot |\Sigma|/w)$ preprocessing time and $\mathcal{O}(\log n \cdot \log w)$ query time,
and another algorithm with $\mathcal{O}(m + m \cdot |\Sigma| \cdot \log w / w)$ preprocessing time and $\mathcal{O}(\log n)$ query time.

\begin{cor}
	\label{cor:lab-suc}
	There is a data structure supporting labelled successor (and predecessor) queries
	on a string $s \in \Sigma^*$ of length $n$ compressed by an SSLP of size $m$
	in the word RAM model with word size $w \ge \log_2 n$ using $\mathcal{O}(m + m \cdot |\Sigma|/w)$ space,
	$\mathcal{O}(m + m \cdot |\Sigma|/w)$ preprocessing time, and $\mathcal{O}(\log n)$ query time.
\end{cor}

\begin{proof}
	In the preprocessing phase we first reduce the depth of the given SSLP to
	$\mathcal{O}(\log n)$ using Theorem~\ref{thm-balance-SSLP}.
	We compute for every variable $X$ the length of $\valX{X}$
	in time and space $\mathcal{O}(m)$
	as in the proof of Theorem~\ref{cor-random-access}.
	Additionally for every variable $X$ we compute a bitvector of length $|\Sigma|$
	which encodes the set of symbols $a \in \Sigma$ that occur in $\valX{X}$.
	Notice that this information takes $\mathcal{O}(m \cdot |\Sigma|)$ bits and fits into $\mathcal{O}(m \cdot |\Sigma|/w)$ memory words.
	If $\rho(X) = YZ$ then the bitvector of $X$ can be computed
	from the bitvectors of $Y$ and $Z$ by
	$\mathcal{O}(|\Sigma|/w)$ many bitwise OR operations.
	Hence in total all bitvectors can be computed in time $\mathcal{O}(m \cdot |\Sigma|/w)$.
	
	A labelled successor query (for position $i$ and symbol $a)$ can now be answered in $\mathcal{O}(\log n)$ time in a straightforward way:
	First we compute the path $(X_0,X_1,\dots,X_\ell)$ in the derivation tree from the root $X_0$ to the symbol at the $i$-th position.
	Then we follow the path starting from the leaf upwards to find the maximal $k$ such that
	$\rho(X_k) = X_{k+1} Y$ and $\valX{Y}$ contains the symbol $a$, or reject if no such $k$ exists.
	Finally, starting from $Y$ we navigate in time $\mathcal{O}(\log n)$ to the leftmost leaf in the derivation tree
	which produces the symbol $a$.
\end{proof}

A {\em minimal subsequence occurrence} of a string $p = a_1 a_2 \cdots a_k$ in a string $s = b_1b_2 \cdots b_l$ is given by 
two positions $i,j$ with $1 \le i \le j \le l$ such that $p$ is a subsequence of $b_i b_{i+1} \cdots b_j$ (i.e., $b_i b_{i+1} \cdots b_j$ belongs
to the language $\Sigma^* a_1 \Sigma^* a_2 \cdots \Sigma^* a_k \Sigma^*$) but $p$ is neither a subsequence of 
$b_{i+1} \cdots b_j$ nor of $b_i \cdots b_{j-1}$.
Following the proof of \cite[Theorem~1]{BilleCG17} we obtain:

\begin{cor} \label{cor-subsequence}
	Given an SSLP $\calG$ of size $m$ producing a string $s \in \Sigma^*$ of length $n$
	and a pattern $p \in \Sigma^*$
	one can compute all minimal subsequence occurrences of $p$ in $s$
	in space $\mathcal{O}(m + m \cdot |\Sigma|/w)$
	and time $\mathcal{O}(m + m \cdot |\Sigma|/w + |p| \cdot \log n \cdot \text{occ})$
	where $w \ge \log n$ is the word size and $\text{occ}$ is the number of minimal subsequence occurrences of $p$ in $s$.
\end{cor}
Corollary~\ref{cor-subsequence} improves \cite[Theorem~1]{BilleCG17}, which states the existence of two algorithms for the computation of all
minimal subsequence occurrences with the following running times (the space bounds are the same as in Corollary~\ref{cor-subsequence}):
\begin{itemize}
\item $\mathcal{O}(m + m \cdot |\Sigma|/w + |p| \cdot \log n \cdot \log w \cdot \text{occ})$,
\item $\mathcal{O}(m + m \cdot |\Sigma| \cdot \log w / w + |p| \cdot \log n \cdot \text{occ})$.
\end{itemize}
Let us briefly mention some other application of Theorem~\ref{thm-balance-SSLP}. As before 
let  $\calG$ be an SSLP of size $m$ for a string $s$ of length $n$.

\paragraph{Computing Karp-Rabin fingerprints for compressed strings.} This problem has been studied in 
\cite{BilleGCSVV17}, where the reader can also finde the definition of finger prints).
Given two positions $i \leq j$ in $s$ one 
wants to compute the Karp-Rabin fingerprint of the factor of $s$ that starts at position $i$ and ends at
position $j$. In \cite{BilleGCSVV17} it was shown that one can compute from $\calG$ a data structure of size $\mathcal{O}(m)$ 
that allows to compute fingerprints in time $\mathcal{O}(\log n)$. First, the authors of \cite{BilleGCSVV17} present a very simple 
data structure of size $\mathcal{O}(m)$ that allows to compute fingerprints in time $\mathcal{O}(\depth(\calG))$.
With Theorem~\ref{thm-balance-SSLP}, we can use this data structure to obtain a 
$\mathcal{O}(\log n)$-time solution. This simplifies the proof in \cite{BilleGCSVV17} 
considerably.

\paragraph{Computing runs, squares, and palindromes in SSLP-compressed strings.} It is shown in 
\cite{IMSIBTNS15} that certain compact representations of the set of all runs, squares and palindromes in $s$ (see
\cite{IMSIBTNS15} for precise definitions) 
can be computed in time $\mathcal{O}(m^3 \cdot \depth(\calG))$. 
With Theorem~\ref{thm-balance-SSLP} we can improve the time bound 
to $\mathcal{O}(m^3 \cdot \log n)$.

\paragraph{Real time traversal for SSLP-compressed strings.} One wants to output the
symbols of $s$ from left to right and thereby spend constant time per symbol. A solution can be found in 
\cite{GasieniecKPS05}; a two-way version (where one can navigate in each step to the left or right neighboring position 
in $s$) can be found in \cite{LohreyMR18}. The drawback of these solutions is that they need space 
$\mathcal{O}(\depth(\calG))$. With Theorem~\ref{thm-balance-SSLP} we can reduce this to space
$\mathcal{O}(\log n)$.

\paragraph{Compressed range minimum queries.}
Range minimum data structure preprocesses a given string $s$ of integers
so that the following queries can be efficiently answered:
given $i \le j$, what is the minimum element in $s_i, \ldots, s_j$
(the substring of $s$ from position $i$ to $j$).
We are interested in the variant of the problem, in which the input is given as an SSLP $\mathcal G$.
It is known, that after a preprocessing taking $\mathcal O(|\mathcal G|)$ time,
one can answer range minimum queries in time $\mathcal O (\log n)$~\cite[Theorem~1.1]{{GaJoMoWe19}}.
This implementation extends the data structure for random access for SSLP~\cite{BilleLRSSW15}
with some additional information,
which includes in particular adding standard range minimum data structures for subtrees leaving the heavy path
and extending the original analysis.
Using the balanced SSLP the same running time can be easily obtained,
without the need of hacking into the construction of the balanced SSLP.
To this end for each variable $X$ we store the length $\ell_X$ of the derived word $\valX{X}$
as well the minimum value in $\valX{X}$.
In the following, let $\rmq(X,i,j)$ be the range minimum query called on $\valX{X}$ for interval $[i, j]$.
Given $\rmq(X,i,j)$, with the rule for $X$ being $X \to YZ$ we proceed as follows:
\begin{itemize}
	\item If the query asks about the minimum in the whole $\valX{X}$,
	i.e., $i = 1$ and $j = \ell_X$,
	then we return the minimum of $\valX{X}$; we call this case trivial in the following.
	\item If the whole range is within the substring generated by the first variable in the rule,
	i.e., $j \le \ell_Y$,
	then we call $\rmq(Y,i,j)$.
	\item If the whole range is within the substring generated by the second nonterminal in a rule,
	i.e., $i > \ell_Y$,
	then we call $\rmq(Z,i-\ell_Y, j-\ell_Y)$.
	\item Otherwise,
	i.e., when $i \leq \ell_Y$ and $j > \ell_Y$
	and $(i, j) \neq (1, \ell_X)$,
	the range spans over the substrings generated by both nonterminals.
	Thus we compute the queries for two substrings and take their minimum,
	i.e., we return the minimum of 
	$\rmq(Y,i,\ell_Y)$ and $\rmq(Z,1,j - \ell_Y)$.
\end{itemize}
To see that the running time is $\mathcal O(\depth (\mathcal G)) = \mathcal O(\log n)$
observe first that the cost of trivial cases can be charged to the
function that called them.
Thus it is enough to estimate the number of nontrivial recursive calls.
In the second and third case there is only one recursive call for a variable
that is deeper in the derivation tree of the SSLP.
In the fourth case there are two calls, 
but two nontrivial calls are made at most once
during the whole computation:
if two nontrivial calls are made in the fourth case
then one of them asks for the $\rmq$ of a suffix of $\valX{Y}$
and the other call asks for the $\rmq$ of a prefix of $\valX{Z}$.
Moreover, every recursive call on a prefix of some string $\valX{X'}$
leads to at most one nontrivial call, which is again on a prefix of
some string $\valX{X''}$; and analogously for suffixes.

\paragraph{Lifshits' algorithm for compressed pattern matching~\cite{lifshitsmatching}.}
The input consists of an SSLP $\mathcal{P}$ for a~pattern $p$ and an SSLP $\mathcal{T}$ for a 
text $t$ and the question is whether $p$ occurs in $t$. Lifshits' algorithm has a 
running time of $\mathcal{O}(|\mathcal P|\cdot |\mathcal T|^2)$.
It was conjectured by the author that the running time could be improved to
$\mathcal{O}(|\mathcal P|\cdot |\mathcal T| \cdot \log |t|)$.
This follows easily from Theorem~\ref{thm-balance-SSLP}:
the algorithm fills a table of size $|\mathcal P| \cdot |\mathcal T|$
and on each entry it calls a recursive subprocedure, whose running time is at most $\depth(\mathcal T)$.
By Theorem~\ref{thm-balance-SSLP} we can bound the running time by $\mathcal{O}(\log |t|)$,
which proves Lifshits' conjecture.
Note, that in the meantime a faster algorithm with running time
$\mathcal{O}(|\mathcal T| \cdot \log |p|)$~\cite{FCPM}
was found.

\paragraph{Smallest grammar problem,}
We conclude Part I of the paper with a remark on the so-called {\em smallest grammar problem} for strings.
In this problem one wants to compute for a given string $w$ a smallest SSLP defining $w$.
The decision variant of this problem is NP-hard, the best known approximation lower bound is $\frac{8569}{8568}$~\cite{CLLLPPSS05},
and the best known approximation algorithms have an approximation ratio of $\mathcal{O}(\log n)$, where $n$ is the 
length of the input string \cite{CLLLPPSS05,Ryt03,Jez15,simplegrammar}. Except for \cite{Jez15},
all these algorithms produce SSLPs of depth $\mathcal{O}(\log n)$. It was discussed in \cite{Jez15}
that the reason for the lack of constant-factor approximation algorithms
might be the fact that smallest SSLPs can have larger than logarithmic depth.
Theorem~\ref{thm-balance-SSLP} refutes this approach.

\section{Part II: Balancing circuits over algebras} \label{sec-part-II}

In the second part of the paper we prove our general balancing result Theorem~\ref{thm-general-balancing}.
This requires some technical machinery concerning (multi-sorted) algebra, terms, and straight-line programs over
algebras. This machinery is introduced in Sections~\ref{sec-trees}--\ref{sec-TSLP->SLP}.
In Sections~\ref{sec-main-result-balancing} and~\ref{sec-proof} we prove (a reformulation of) Theorem~\ref{thm-general-balancing}.
Finally in Section~\ref{sec-FSLP} and \ref{sec-cluster}  we apply Theorem~\ref{thm-general-balancing} to forest straight-line programs
and top dags, which yields Theorem~\ref{thm-balancing-TSLP} from the introduction.

\subsection{Algebras and their straight line programs} \label{sec-algebras+SLPs}

\subsubsection{Ranked trees} \label{sec-trees}

Let us fix a finite set $\So$ of {\em sorts}. Later, we will assign to each sort $i \in \So$ a set $A_i$ (of elements of sort $i$).
An {\em $\So$-sorted signature} is a set of symbols $\Gamma$ and a mapping
$\type \colon \Gamma \to \So^+$ that assigns to each symbol from $\Gamma$ a non-empty word
over the alphabet $\So$.
The number $|\type(f)|-1 \ge 0$ is also called the {\em rank} of $f$.
Let $\Gamma_i \subseteq \Gamma$ ($i \ge 0$) be the set of all symbols in $\Gamma$ of rank $i$.

Let us also fix a second (infinite) $\So$-sorted signature $\X$, where every $x \in \X$ has rank zero.
Elements of $\X$ are called {\em variables}. Since $x \in \X$ has rank zero, $\type(x)$ is an element of
$\So$. For $p \in \So$ let $\X_p = \{x \mid \type(x)=p\}$. 
We assume that every set $\X_p$ is infinite.
We will always work with a finite subset $\Y$ of $\X$. 
Take such a set $\Y$.
For each sort $p \in \So$ we define the set of {\em terms} 
$\T_p(\Gamma,\Y)$ of sort $p$ by simultaneous induction  as 
the smallest set such that the following holds:
\begin{itemize}
\item Every $x \in \X_p \cap \Y$ belongs to $\T_p(\Gamma,\Y)$.
\item If $f \in \Gamma_n$ with $\type(f) = p_1 \cdots p_n q$ 
and $t_i \in \T_{p_i}(\Gamma,\Y)$ for $1 \leq i \leq n$, 
then $f(t_1, t_2, \ldots, t_n) \in \T_q(\Gamma,\Y)$. 
\end{itemize}
We write $\T_p(\Gamma)$ for $\T_p(\Gamma,\emptyset)$, and call its elements {\em ground terms} (of sort $p$).
Note that if $a \in \Gamma_0$ and 
$\type(a) = p \in \So$
then $a() \in \T_p(\Gamma)$. In this case, we write $a$ for $a()$ and call $a$ a constant of sort $p$.
Let $\T(\Gamma,\Y) = \bigcup_{p \in \So} \T_p(\Gamma,\Y)$.

Elements of $\T(\Gamma,\Y)$ can be viewed as node labeled trees, where leaves are labeled with symbols form $\Gamma_0 \cup \Y$
and every internal node is labeled with a symbol from some $\Gamma_n$ with $n \geq 1$:
The root of the tree corresponding to the term $f(t_1, t_2, \ldots, t_n)$ is labeled with $f$ and its direct subtrees
are the trees corresponding to $t_1, \ldots, t_n$. 
Note that the composition of two functions $f \colon A \to B$ and $g \colon B \to C$ is denoted
by $g \circ f$,
in particular we first apply $f$ followed by $g$.

For a term $t$ we define the size $|t|$ of $t$ as the  number of edges of the corresponding tree.
Equivalently, $|t|$ is inductively defined as follows: If $t = x$ is a variable, then $|t| = 0$.
If $t = f(t_1, t_2, \ldots, t_n)$ for $f \in \Gamma$, then $|t| = n + \sum_{i=1}^n |t_i|$.
The {\em depth} of a term $t$ is denoted by $\depth(t)$ and defined inductively as usual:
If $t = x$ is a variable, then $\depth(t) = 0$.
If $t = f(t_1, t_2, \ldots, t_n)$ for $f \in \Gamma$, 
then $\depth(t) = \max \{1+\depth(t_i) \mid 1 \le i \le n\}$ with $\max \emptyset = 0$.

\begin{defi}[substitutions] \label{def-sub}
A {\em substitution} is a mapping $\eta \colon \Y \to \T(\Gamma,\Z)$ for finite (not necessarily disjoint) subsets
$\Y, \Z \subseteq \X$ such that $y \in \Y \cap \X_p$ implies $\eta(y) \in \T_p(\Gamma,\Z)$.
If $\Z = \emptyset$, we speak of a {\em ground substitution}.
For $t \in \T(\Gamma,\Y)$ we define the term $\eta(t)$ by replacing simultaneously
all occurrences of variables in $t$ by their images under $\eta$. Formally we extend 
$\eta \colon \Y \to \T(\Gamma,\Z)$ to a mapping $\eta \colon \T(\Gamma,\Y) \to \T(\Gamma,\Z)$
by $\eta(f(t_1,\ldots,t_n)) = f(\eta(t_1), \ldots, \eta(t_n))$ (in particular, $\eta(a)=a$ for $a \in \Gamma_0$).
A {\em variable renaming} is a bijective substitution $\eta \colon \Y \to \Z$ for finite variable sets $\Y$ and $\Z$ of
the same size.
\end{defi}

\begin{defi}[contexts] \label{def-contexts}
Let $p,q \in \So$. We define the set of {\em contexts}
$\calC_{pq}(\Gamma,\Y)$ as the set of all terms $t \in \T_q(\Gamma, \Y \cup \{x\})$,
where $x \in \X_p \setminus \Y$ is a fresh variable such that 
(i) $t \neq x$, (ii) and $x$ occurs exactly once in $t$. We call $x$ the {\em main variable} of $t$
and $\Y$ the set of auxiliary variables of $t$.\footnote{Since also $\Y$ may contain a variable $y$
that occurs exactly once in $t$, we explicitly have to declare a variable as the main variable. Most of the
times, the main variable will be denoted with $x$.}
We write $\calC_{pq}(\Gamma)$ for $\calC_{pq}(\Gamma,\emptyset)$. Elements of 
$\calC_{pq}(\Gamma)$ are called {\em ground contexts}. Let $\calC(\Gamma,\Y) = 
\bigcup_{p,q\in\So} \calC_{pq}(\Gamma,\Y)$ and $\calC(\Gamma) = \calC(\Gamma,\emptyset)$.
For $s \in \calC_{qr}(\Gamma,\Y)$ 
and $t \in \T_q(\Gamma,\Z)$ (or $t \in \calC_{pq}(\Gamma,\Z))$
we define $s[t] \in \T_r(\Gamma, \Y \cup \Z)$ ($s[t] \in \calC_{pr}(\Gamma, \Y \cup \Z)$) 
as the result of replacing the unique occurrence of the main variable in $s$ by $t$.
Formally, we can define $s[t]$ as $\eta(s)$ where $\eta$ is the substitution with domain
$\{x\}$ and $\eta(x)=t$, where $x$ is the main variable of $s$.
An {\em atomic context} is a context of the form
$f(y_1,\ldots,y_{k-1},x,y_{k+1},\ldots,y_k)$ where $x$ is the main variable and
the $y_i$ are the auxiliary variables (we can have $y_i = y_j$ for $i \neq j$).
Note that there are only finitely many atomic contexts up to renaming of variables.
\end{defi}

\subsubsection{Algebras} \label{sec-alg}

We will produce strings, trees and forests by ground terms (also called algebraic expressions in this context) over certain (multi-sorted) algebras. These expressions
will be compressed by directed acyclic graphs. In this section, we introduce the generic framework, which will
be reinstantiated several times later on.

Fix a finite $\So$-sorted signature $\Gamma$.
A {\em $\Gamma$-algebra} is a tuple $\A = ( (A_p)_{p \in \So}, (f^{\A})_{f \in \Gamma})$ where
every $A_p$ is a non-empty set (the universe of sort $p$ or the set of elements of sort $p$) and
for every $f \in \Gamma_n$ with $\type(f) = p_1 p_2 \cdots p_n q$, 
$f^{\A} \colon \prod_{1 \le j \le n}  A_{p_j} \to A_q$ is an $n$-ary function.
We also say that $\Gamma$ is the {\em signature} of $\A$. In our settings, the sets $A_p$ will be always
pairwise disjoint, but formally we do not need this.
Quite often, we will identify the function $f^{\A}$ with the symbol $f$. 
Functions of arity zero are elements of some $A_p$. A ground term $t \in \T_p(\Gamma)$ can be viewed as 
algebraic expressions over $\A$ that evaluates to an element $t^{\A} \in A_p$ in the natural way.
For $x \in \bigcup_{p \in \So} A_p$ we also write $x \in \A$ and for $A_p$ we also write $\A_p$.

When we define a $\Gamma$-algebra, we usually will not specify the types of the symbols
in $\Gamma$. Instead, we just list the sets $A_p$ ($p \in \So$) and the functions $f^{\A}$ ($f \in \Gamma$) including
their domains. The latter implicitly determine the types of the symbols in $\Gamma$.

\begin{ex} \label{ex-vector-space}
A~well known example of a~multi-sorted algebra is a~vector space. More precisely, it can be formalized
as a~$\Gamma$-algebra, where $\Gamma = \{ \overline{0}, 0, 1, \oplus, \odot, +, \cdot \}$ is 
a $\So$-sorted signature for $\So = \{ v, s \}$. Here $v$ stands for ``vectors'' and $s$ stands
for ``scalars''. The types of the symbols in $\Gamma$ are defined as follows:
\begin{itemize}
\item $\type(\overline{0}) = v$ (the zero vector), 
\item $\type(0) = s$ (the $0$-element of the scalar field), 
\item $\type(1) = s$ (the $1$-element of the scalar field),
\item $\type(\oplus) = vvv$ (vector addition), 
\item $\type(\odot) = svv$ (multiplication of a~scalar by a~vector),
\item $\type(+) = sss$ (addition in the field of scalars),
\item $\type(\cdot) = sss$ (multiplication in the field of scalars).
\end{itemize} 
Note that we cannot define non-trivial vectors by ground terms. For this, we should add some constants of type $v$ 
to the signature. For the vector space $F^n$ for a field $F$ we might for instance add the constants $e_1, \ldots, e_n$,
where $e_i$ denotes the $i$-th unit vector.
\end{ex}
From the sets $\T_p(\Gamma)$ one can construct the {\em free
term algebra} $$\T(\Gamma) = ( (\T_p(\Gamma))_{p \in \So}, (f)_{f \in \Gamma}),$$ where every ground term
evaluates to itself. For every $\Gamma$-algebra $\A$, the mapping $t \mapsto t^{\A}$
($t \in  \T(\Gamma)$) is a homomorphism from the free term algebra to $\A$.
We need the technical assumption that this homomorphism is surjective, i.e., for 
every $a \in \A$ there exists a ground term $t \in \T(\Gamma)$ with $a = t^{\A}$. 
In our concrete applications this assumption will be satisfied. Moreover, one can always replace
$\A$ by the subalgebra induced by the elements $t^{\A}$ (we will say more about this later).

For a $\Gamma$-algebra $\A = ( (A_p)_{p \in \So}, (f^{\A})_{f \in \Gamma})$,
a variable $x \in \X_p$ and $a \in A_p$,
we define the $(\Gamma \cup \{x\})$-algebra $\A[x/a] = ( (A_p)_{p \in \So}, (f^{\A[x/a]})_{f \in \Gamma\cup\{x\}})$
by $f^{\A[x/a]} = f^{\A}$ for $f \in \Gamma$ and $x^{\A[x/a]} = a$.

\begin{defi}[unary linear term functions] \label{def-term-func}
Given a $\Gamma$-algebra $\A$ and a ground context
$t \in \calC_{pq}(\Gamma)$ with main variable $x$, 
we define the function $t^{\A} \colon A_p \to A_q$ by
$t^{\A}(a) = t^{\A[x/a]}$ for all $a \in A_p$.
We call $t^{\A}$ a {\em unary linear term function}, ULTF for short.
We write $\lin_{pq}(\A)$ for the set of all ULTFs $t^{\A}$ with
$t \in \calC_{pq}(\Gamma)$.
\end{defi}

\begin{ex}
Consider the the vector space $\mathbb{R}^2$ in the context of Example~\ref{ex-vector-space} and let
$t = e_1 \oplus ( (1+1) \odot (x_v \oplus e_2)) \in \mathcal{C}_{vv}$ (recall that $v$ is the sort of vectors). 
The~corresponding ULTF is the affine mapping $x \mapsto 2 x \oplus (1,2)^{\operatorname{T}}$ on $\mathbb{R}^2$.
\end{ex}
As another example note that a~ULTF, where the underlying algebra is a~ring $\mathcal{R}$ (this is a~one-sorted algebra), is nothing else
than a~linear polynomial over $\mathcal{R}$ in a~single variable $x$.

\subsubsection{Straight-line programs} \label{sec-SLP}

Let $\Gamma$ be any $\So$-sorted signature.
A {\em straight-line program} over $\Gamma$ ($\Gamma$-SLP for short)
is a tuple $\calG = (\V, \rho, S)$, where $\V \subseteq \mathcal{X}$ is a finite set of variables,
$S \in \V$ is the {\em start variable} and $\rho \colon \V \to \T(\Gamma,\V)$ is a substitution (the so called
{\em right-hand side mapping})
such that the edge relation $E(\calG) =\{ (y, z) \in \V \times \V \mid z \text{ occurs in } \rho(y) \}$ is acyclic.
This implies that there exists an $n \geq 1$ such that $\rho^n \colon \T(\Gamma,\V) \to \T(\Gamma)$
(the $n$-fold composition of $\rho$) is a
ground substitution (we can choose $n = |\V|$). For this $n$, we write $\rho^*$ for $\rho^n$.
Note that $\rho^* \circ \rho = \rho \circ \rho^* = \rho^*$. 
The term defined by $\calG$ is $\valX{\calG} := \rho^*(S)$; it is also called the {\em derivation tree} of $\calG$.

In many papers on straight-line programs, the variables of a $\Gamma$-SLP are denoted 
by capital letters $X,Y,Z, X'$, etc. We follow this tradition.
For a variable $X \in \V$ we 
also write $\valXG{\calG}{X}$ (or $\valX{X}$ if $\calG$ is clear from the context) for the ground term $\rho^*(X)$. 

Let $\A$ be a $\Gamma$-algebra. A  $\Gamma$-SLP $\calG = (\V, \rho, S)$
is also called an SLP over the algebra $\A$.
We can evaluate every variable $X \in \V$ 
to its value $\rho^*(X)^{\A} = \valX{X}^{\A} \in \A$ in $\A$. 
It is important to distinguish this value from the 
syntactically computed ground term $\rho^*(X)$ (which is the evaluation of $X$ in the free term algebra). 
Also note that in part I of the paper, we used
the notation $\valX{X}$ for variables of string straight-line programs, which are obtained from the above general
definition by taking a free monoid $\Sigma^*$ for the structure $\A$. In other words: a string $\valX{X}$ from the first
part of the paper would be denoted with $\valX{X}^{\Sigma^*}$ in the second part of the paper. The reason for this
change in notation is two-fold. First, we did not want to overload the notation in Part I (especially for readers that are only
interested in the balancing result for strings); hence we decided to omit the 
superscripts $\Sigma^*$ there. Second, in the following sections the ground terms $\valX{X}$ are the important objects,
which justifies a short notation for them.

The term $\rho(X)$ is also called the {\em right-hand side} of the variable $X \in \V$.
By adding fresh variables,
we can transform every $\Gamma$-SLP in linear time into a so-called {\em standard $\Gamma$-SLP},
where all right-hand sides have the form $f(X_1, \ldots, X_{n})$ for variables
$X_1, \ldots, X_n$ (we can have $X_i = X_j$ for $i \neq j$). 
A~standard $\Gamma$-SLP $\calG$ is the same object as a DAG (directed acyclic graph) with $\Gamma$-labelled
nodes: the DAG is $(\V,E(\calG))$ and if $\rho(X) = f(X_1, \ldots, X_{n})$ then node $X$ is labelled with $f$.
Since the order of the edges $(X,X_i)$ ($1 \le i \le n$) is important and we may have $X_i = X_j$ for $i \neq j$
we formally replace the edge $(X,X_i)$ by the triple $(X,i,X_i)$. A $\Gamma$-SLP interpreted over a 
$\Gamma$-algebra $\A$ is also called an {\em algebraic circuit} over $\A$.

Consider a (possibly non-standard) $\Gamma$-SLP $\calG = (\V,\rho,S)$.
We define the {\em size} of $|\calG|$ as $\sum_{X \in \V} |\rho(X)|$.
For a standard $\Gamma$-SLP this is the number of edges of the corresponding DAG $(\V,E(\calG))$.
The {\em depth} of $\calG$ is defined as $\depth(\calG) = \depth(\valX{\calG})$,
i.e.~the depth of the derivation tree of $\calG$.
For a standard $\Gamma$-SLP $\calG$ this is the
maximum length of a directed path in the DAG $(\V,E(\calG))$.
Our definitions of size and depth ensure that both measures do not increase when one transforms 
a given $\Gamma$-SLP into a standard $\Gamma$-SLP.
In this paper, the sizes of the right-hand sides will be always bounded
by a constant that only depends on the underlying algebra $\A$.

\subsubsection{Functional extensions} 

An important concept in this paper is a functional extension $\hat\T(\Gamma)$ of the free term algebra $\T(\Gamma)$. 
We define an algebra $\hat\T(\Gamma)$ over an $\So \cup \So^2$-sorted signature $\hat{\Gamma}$.

\begin{defi}[Signature $\hat{\Gamma}$] \label{def-hatGamma}
Let $\Gamma$ be a $\So$-sorted signature.
The $\So \cup \So^2$-sorted signature $\hat{\Gamma}$ is
\begin{equation} \label{gamma-hat}
\hat\Gamma = \Gamma \uplus \bigcup_{n \ge 1} \{ \hat{f}_i \mid f \in \Gamma_n, 1 \le i \le n\} 
\uplus \{ \gamma_{pqr} \mid p,q,r \in \So \} \uplus \{ \alpha_{pq}  \mid p,q \in \So \}
\end{equation}
where the type function is defined as follows:
\begin{itemize}
\item Symbols from $\Gamma$ have the same types in $\hat \Gamma$.
\item If $\type(f) = p_1 \cdots p_n q$ then $\type(\hat{f}_i) = p_1 \cdots p_{i-1} p_{i+1} \cdots p_n q$.
\item For all $p,q,r \in \So$ we set $\type(\gamma_{pqr}) = (p,q)(q,r)(p,r)$.
\item For all $p,q \in \So$ we set $\type(\alpha_{pq}) = p(p,q)q$.
\end{itemize}
\end{defi}

\begin{defi}[$\hat\Gamma$-algebra  $\hat\T(\Gamma)$] \label{def-hatT}
The $\hat\Gamma$-algebra
$\hat\T(\Gamma) = ( (A_s)_{s \in \So \cup \So^2}, (f^{\hat\T(\Gamma)})_{f \in \hat{\Gamma}})$
is defined
as follows: the sets $A_p$ and $A_{pq}$ for $p,q \in \So$
are defined as
\begin{itemize}
\item $A_p = \T_p(\Gamma)$ and
\item $A_{pq} = \calC_{pq}(\Gamma)$.
\end{itemize}
The operations $g^{\hat\T(\Gamma)}$ ($g \in \hat\Gamma$) are defined as follows, where we write 
$g$ instead of  $g^{\hat\T(\Gamma)}$:
\begin{itemize}
\item For every symbol $f\in\Gamma_n$
the algebra $\hat\T(\Gamma)$ inherits the function $f^{\T(\Gamma)}$ from $\T(\Gamma)$.

\item For every symbol $f\in\Gamma_n$ with $\type(f)=p_1\cdots p_n q$ ($n \geq 1$) and every
$1 \le k \le n$ we define the $(n-1)$-ary operation
\[
\hat{f}_k \colon \prod_{1 \le i \le n \atop i \neq k} \T_{p_i}(\Gamma) \to \calC_{p_kq}(\Gamma)
\]
by $\hat{f}_k(t_1, \ldots, t_{k-1},t_{k+1},\ldots,t_n) = f(t_1, \ldots, t_{k-1},x,t_{k+1},\ldots,t_n)$ for all
$t_i \in \T_{p_i}(\Gamma)$ ($1 \le i \le n$, $i \neq k$).
\item For all $p,q,r \in \So$ the binary operation $\gamma_{pqr} \colon \calC_{pq}(\Gamma) \times \calC_{qr}(\Gamma) \to \calC_{pr}(\Gamma)$ is defined
by $\gamma_{pqr}(t,s) = s[t]$. 
\item For all $p,q \in \So$ the binary operation $\alpha_{pq} \colon \T_{p}(\Gamma) \times \calC_{pq}(\Gamma)  \to \T_{q}(\Gamma)$ is defined
by $\alpha_{pq}(t,s) = s[t]$. 
\end{itemize}
\end{defi}
The definition of the operations $\alpha_{pq}$ and $\gamma_{pqr}$ suggests to write 
$s[t]$ instead of $\alpha_{pq}(t,s)$ or $\gamma_{pq}(t,s)$, which we will do most of the times.

Recall the definition of unary linear term functions (ULTFs) from Definition~\ref{def-term-func}.
An {\em atomic ULTF} is of the form $z \mapsto f^{\A}(a_1,\ldots,a_{k-1},z,a_{k+1},\ldots,a_n)$
for $f \in \Gamma_n$ with $\type(f) = p_1 \cdots p_nq$ and $a_i \in A_{p_i}$ for 
($1 \le i \le n$, $i \neq k$). We denote this function with $f^{\A}(a_1,\ldots,a_{k-1},\cdot,a_{k+1},\ldots,a_n)$
in the following. At this point, we use the assumption that every element of $\A$ can be written as $t^{\A}$
for a ground term $t$. Hence, the elements $a_i$ are defined by terms, which ensures that 
$f^{\A}(a_1,\ldots,a_{k-1},\cdot,a_{k+1},\ldots,a_n)$ is indeed a ULTF.
It is easy to see that every ULTF is the composition of finitely many atomic ULTFs.

\begin{defi}[$\hat\Gamma$-algebra  $\hat\A$] \label{def-hatA}
Given a $\Gamma$-algebra $\A = ((A_p)_{p \in \So}, (f^{\A})_{f \in \Gamma})$
we define the $\hat\Gamma$-algebra $\hat\A = 
((B_s)_{s \in \So \cup \So^2}, (f^{\hat\A})_{f \in \hat{\Gamma}})$ as follows:
The sets $B_p$ and $B_{pq}$ for $p,q \in \So$
are defined as:
\begin{itemize}
\item $B_p = A_p$  and
\item $B_{pq} = \lin_{pq}(\A)$.
\end{itemize}
The operations $g^{\hat\A}$ ($g \in \hat\Gamma$) are defined as follows, where we write 
$g$ instead of  $g^{\hat\A}$.
\begin{itemize}
\item Every $f \in \Gamma$ is interpreted as $f^{\hat \A} = f^{\A}$.
\item For every symbol $f\in\Gamma_n$ with $\type(f)=p_1\cdots p_n q$ ($n \geq 1$) and every
$1 \le k \le n$ we define the $(n-1)$-ary operation
$$
\hat{f}_k : \prod_{1 \le i \le n \atop i \neq k} A_{p_i} \to \lin_{p_kq}(\A)
$$
by $\hat{f}_k(a_1, \ldots, a_{k-1},a_{k+1},\ldots,a_n) = f^{\A}(a_1, \ldots, a_{k-1},\cdot,a_{k+1},\ldots,a_n)$ for all
$a_i \in A_{p_i}$ ($1 \le i \le n$, $i \neq k$).
\item For all $p,q,r \in \So$ the binary operation $\gamma_{pqr} \colon \lin_{pq}(\A) \times \lin_{qr}(\A) \to 
\lin_{pr}(\A)$ is defined as function composition: $\gamma_{pqr}(g, h) = h \circ g$.
\item For all $p,q \in \So$ the binary operation $\alpha_{pq} \colon A_p \times \lin_{pq}(\A)  \to A_q$ is defined
as function application: $\alpha_{pq}(a,g) = g(a)$. 
\end{itemize}
\end{defi}
Note that Definitions~\ref{def-hatT} and~\ref{def-hatA} are consistent in the following sense:
If we apply the construction from Definition~\ref{def-hatA} for $\A = \T(\Gamma)$ (the free term algebra)
then we obtain an isomorphic copy of the algebra $\hat\T(\Gamma)$ from Definition~\ref{def-hatT}, i.e., 
$\widehat{\T(\Gamma)} \cong \hat\T(\Gamma)$. Moreover, the mappings $t \mapsto t^{\A}$ (for ground terms $t$)
and $c \mapsto c^{\A}$ (for ground contexts $c$) yield a canonical surjective morphism from $\hat\T(\Gamma)$ to $\hat\A$
that extends the canonical morphism from the free term algebra $\T(\Gamma)$ to $\A$.

\subsubsection{Tree straight-line programs} 

Recall the definition of 
the $\So \cup \So^2$-sorted signature $\hat\Gamma$ in \eqref{gamma-hat}.
A $\hat\Gamma$-SLP $\calG$ which evaluates in the $\hat\Gamma$-algebra 
$\hat\T(\Gamma)$ to a ground term (i.e., $\valX{\calG}^{\hat\T(\Gamma)} \in \T(\Gamma)$) is also called
a {\em tree straight-line program} over $\Gamma$ ($\Gamma$-TSLP for short)~\cite{GHJLN17,GL18,Lohrey15dlt}.

Recall that $\hat\Gamma$ contains for every $f \in \Gamma_n$ with $n \geq 1$ the unary symbols
$\hat{f}_k$ ($1 \le k \le n$). Right-hand sides of the form $\hat{f}_k(X_1,\ldots,X_{k-1},X_{k+1},\ldots,X_n)$ 
in a $\Gamma$-TSLP are written for better readability as $f(X_1,\ldots,X_{k-1},x,X_{k+1},\ldots,X_n)$. 
This is also the notation used in \cite{GHJLN17,GL18,Lohrey15dlt}.
For right-hand sides of the form $\alpha_{pq}(X,Y)$ or $\gamma_{pqr}(X,Y)$ we write $X[Y]$.

\begin{ex}
Let us assume that $\So$ consists of a single sort.
Consider the $\Gamma$-TSLP 
$$\calG = (\{S,X_1,\ldots, X_7\},\rho,S)$$ with 
$\Gamma_2 = \{f,g\}$, $\Gamma_0 = \{a,b\}$ and 
$\rho(S) = X_1[X_2]$,  
$\rho(X_1) = X_3[X_3]$, 
$\rho(X_2) = X_4[X_5]$,
$\rho(X_3) = f(x,X_7)$, 
$\rho(X_4) = X_6[X_6]$,  
$\rho(X_5) = a$,
$\rho(X_6) = g(X_7,x)$,
$\rho(X_7) = b$.
We get 
\begin{itemize}
\item $\valX{X_6}^{\hat\T(\Gamma)} = \rho^*(X_6)^{\hat\T(\Gamma)} = g(b,x)$,
\item $\valX{X_4}^{\hat\T(\Gamma)} = \rho^*(X_4)^{\hat\T(\Gamma)} = g(b,x)[g(b,x)] =
 g(b,g(b,x))$,
 \item $\valX{X_3}^{\hat\T(\Gamma)} = \rho^*(X_3)^{\hat\T(\Gamma)} = f(x,b)$,
\item $\valX{X_2}^{\hat\T(\Gamma)} = \rho^*(X_2)^{\hat\T(\Gamma)} = 
g(b,g(b,x))[a] = g(b,g(b,a))$,
\item $\valX{X_1}^{\hat\T(\Gamma)} =  \rho^*(X_1)^{\hat\T(\Gamma)} = 
f(x,b) [f(x,b)] =  f(f(x,b),b)$, and
\item $\valX{\calG}^{\hat\T(\Gamma)} = \rho^*(S)^{\hat\T(\Gamma)} =
f(f(x,b),b)[g(b,g(b,a))] =
f(f(g(b,g(b,a)),b),b)$.
\end{itemize}
\end{ex}

\subsubsection{From TSLPs to SLPs}  \label{sec-TSLP->SLP}

Fix a $\Gamma$-algebra $\A$.
Our first goal is to transform a $\Gamma$-TSLP $\calG$ into a $\Gamma$-SLP 
$\mathcal{H}$ of size $\mathcal{O}(|\calG|)$ and depth $\mathcal{O}(\depth(\calG))$ such that
$\valX{\mathcal{H}}^{\A} = \valX{\mathcal{G}}^{\hat\A}$. 
For this, we have to restrict the class of $\Gamma$-algebras. For instance, for
the free term algebra the above transformation cannot be achieved in general:
the chain tree $t_n = f(f(f(\cdots f(a)\cdots)))$ with 
$2^n$ occurrences of $f$ can be easily produced by a $\{a,f\}$-TSLP of size $\mathcal{O}(n)$
but the only DAG (= SLP over the free term algebra $\T(\{a,f\})$) for $t_n$ is $t_n$ itself.
We restrict ourselves to algebras with a finite subsumption base,
as defined below.
Such algebras have been implicitly used in our recent papers \cite{GHJLN17,GL18}. 

\begin{defi}[equivalence and subsumption preorder in $\A$] \label{def-subsumption}
For contexts $s,t \in \calC_{pq}(\Gamma,\Y)$ we say that $s$ and $t$ are {\em equivalent} in $\A$ 
if for every ground substitution $\eta \colon \Y \to \T(\Gamma)$ we have
$\eta(s)^{\A} = \eta(t)^{\A}$ (which is an ULTF).

For contexts $s \in \calC_{pq}(\Gamma,\Y)$ and
$t \in \calC_{pq}(\Gamma,\Z)$ we say that $t$ {\em subsumes} $s$ in $\A$ or that $s$ is subsumed by $t$ in $\A$
($t \leq^{\A} s$ for short) 
if there exists a substitution $\zeta \colon \Z \to \T(\Gamma, \Y)$ such that 
$s$ and $\zeta(t)$ are equivalent in $\A$.

A {\em subsumption base} of $\A$ is a set of (not necessarily ground) contexts $C$ such that
for every context $s$ there exists a context $t \in C$ with $t \le^{\A} s$.
\end{defi}
It is easy to see that $\le^{\A}$ is reflexive and transitive but 
in general not antisymmetric. 
Moreover, the relation $\le^{\A}$ satisfies the following monotonicity property:
\begin{lem}
	\label{lem:sub-com}
	Let $s \in \calC_{qr}(\Gamma,\Y)$, $t_1 \in \calC_{pq}(\Gamma,\Z_1)$ and $t_2 \in \calC_{pq}(\Gamma,\Z_2)$ be contexts
	such that $\Y \cap \Z_1 = \emptyset$ and $\Y \cup \Z_1 \cup \Z_2$ contains none of the main variables of $s$, $t_1$, $t_2$.
	If $t_1 \le^\A t_2$ then $s[t_1] \le^\A s[t_2]$.
\end{lem}

\begin{proof}
	Since $t_1$ subsumes $t_2$ in $\A$ there exists a substitution $\zeta \colon \Z_1 \to \T(\Gamma,\Z_2)$
	such that for every ground substitution $\eta \colon \Z_2 \to \T(\Gamma)$ we have
	\[
		\eta(t_2)^{\A} =   \eta(\zeta(t_1))^{\A}.
	\]
	Define the substitution $\zeta' \colon \Y \cup \Z_1 \to \T(\Gamma,\Y \cup \Z_2)$ by
	\[
		\zeta'(y) = \begin{cases}
		\zeta(y) & \text{if } y \in \Z_1, \\
		y & \text{if } y \in \Y.
		\end{cases}
	\]
	As $\Z_1 \cap \Y = \emptyset$ by the assumption, $\zeta'$ is well defined.
	It satisfies $\zeta'(t_1) = \zeta(t_1)$ and $\zeta'(s) = s$.
	For any ground substitution $\eta \colon \Y \cup \Z_2 \to \T(\Gamma)$ we have:
	\begin{eqnarray*}
		\eta(s[t_2])^\A &=& (\eta(s)[\eta(t_2)])^\A \\
		&=& \eta(s)^\A \circ \eta(t_2)^\A \\
		&=& \eta(s)^\A \circ \eta(\zeta(t_1))^\A \\
		&=& \eta(\zeta'(s))^\A \circ \eta(\zeta'(t_1))^\A \\
		&=& (\eta(\zeta'(s))[\eta(\zeta'(t_1))])^\A \\
		&=& \eta(\zeta'(s[t_1]))^\A .
	\end{eqnarray*}
	This implies $s[t_1] \le^\A s[t_2]$.
\end{proof}
We will be interested in algebras that have a finite subsumption base. 
In order to show that a set $C$ is a finite subsumption base
 we will use the following lemma.

\begin{lem}
	\label{lem:tame-atomic}
	Let $\A$ be a $\Gamma$-algebra and let $C$ be a finite set of contexts with the following properties:
	\begin{itemize}
         \item For every atomic context $s$ there exists $t \in C$ with $t \le^{\A} s$.
         \item For every atomic context $s$
	and every $t \in C$ such that $s[t]$ is defined and 
	$s$ and $t$ do not share auxiliary variables, there exists $t' \in C$ with 
	$t' \leq^{\A} s[t]$.
        \end{itemize}
     	Then $C$ is a subsumption base.
\end{lem}

\begin{proof}
	Assume that the two conditions from the lemma hold. We show by induction on $s$ that for every context $s$ 
	there exists a context $t \in C$ with $t \leq^{\A} s$.
	
	If $s = f(s_1, \dots, s_{i-1},x,s_{i+1}, \dots, s_n)$ for some terms $s_1, \dots, s_{i-1},s_{i+1}, \dots, s_n$
	then $s$ is subsumed in $\A$ by the atomic context $f(y_1, \dots, y_{i-1},x,y_{i+1}, \dots, y_n)$,
	which in turn is subsumed in $\A$ by some $t \in C$.
	If $s = f(s_1, \dots, s_{i-1},s',s_{i+1}, \dots, s_n)$ for some terms $s_1, \dots, s_n$ and some context $s'$
	then $f(y_1, \dots, y_{i-1},s',y_{i+1}, \dots, y_n) \le^{\A} s$ 
	for fresh auxiliary variables $y_1, \dots, y_{i-1},y_{i+1}, \dots, y_n$ (that neither occur in $s'$ nor any context from $C$). 
	By induction there exists $t' \in C$ with $t' \le^{\A} s'$. 
	By Lemma~\ref{lem:sub-com} we have $f(y_1, \dots, y_{i-1},t',y_{i+1}, \dots, y_n) \le^{\A} f(y_1, \dots, y_{i-1},s',y_{i+1}, \dots, y_n)$.
	By the second assumption from the lemma, we have that 
	$t'' \le^{\A} f(y_1, \dots, y_{i-1},t',y_{i+1}, \dots, y_n)$ for some $t'' \in C$.
	We get $t'' \le^{\A} s$ by transitivity of $\le^{\A}$.
\end{proof}

\begin{rem} \label{rem-tame-2}
Recall that we made the technical assumption that every element $a$ of $\A$ can be written as 
$t^{\A}$ for a ground term $\A$. Let $\B$ be the subalgebra of $\A$ that is induced by all elements $t^{\A}$ for $t \in \T(\A)$. It is 
obvious that every subsumption base of $\A$ is also a subsumption base of $\B$. 
\end{rem}

\begin{ex} \label{ex-semirings}
Every semiring $\A = (A,+,\times, a_1, \ldots, a_n)$, where $a_1, \ldots, a_n \in A$ are arbitrary
constants, has a finite subsumption base. Here we do not assume that $\times$ is commutative,
nor do we assume that identity elements with respect to $+$ or $\times$ exist. In other words:
$(A,+)$ is a commutative semigroup, $(A,\times)$ is a semigroup and the left and right distributive
law holds. The finite subsumption base $C(\A)$ consists of the following contexts 
$axb+c, ax+c, xb+c, x+c, axb, ax, xb$, and $x$, where $x$ is the main variable and $a,b,c$ are auxiliary
variables. We write $ab$ instead of $a \times b$ and omit in $axb$ brackets that
are not needed due to the associativity of multiplication. To see that every context $s$ is subsumed in $\A$
by one of the contexts from $C(\A)$, observe that a context defines a linear polynomial in the main variable $x$.
Hence, every context is equivalent in $\A$ to a context of the form 
$sxt+u, sx+u, xt+u, x+u, sxt, sx, xt$ or $x$, where $s,t,u$ are terms that contain the auxiliary
parameters. Each of these contexts is subsumed by a context from $C(\A)$ by the substitution
$\zeta$ with $\zeta(a)=s$, $\zeta(b)=t$, and $\zeta(c)=u$.

Let us remark that the above proof can be adapted to the situation that also $+$ is not commutative.
In that case, we have include the terms $c' + axb + c$,
$c'+ax+c$, $c'+xb+c$, $c'+x+c$, $c'+axb$, $c'+ax$, $c'+xb$ and $c'+x$ to the set $C(\A)$.

On the other hand, if $(A,+)$ has a neutral element, called $0$ in the following,
then $axb+c, ax+c, xb+c, x+c$ is s subsumption base.
To see this observe that, for instance, $axb$ is subsumed by $axb + c$,
which is shown by the substitution $c \mapsto 0$.
\end{ex}

\begin{ex} \label{ex-free}
If $\Gamma$ contains a symbol of rank at least one, then
the free term algebra $\T(\Gamma)$ has no finite subsumption base:
If $C$ were a finite subsumption base of $\T(\Gamma)$, then every ground context could be obtained from some
$t \in C$ by replacing the auxiliary parameters in $t$ by ground terms.
But this replacement does not change the length of the path from the root of the context to its main variable.
Hence, we would obtain a bound for the length of the path from the root to the main variable in a ground context,
which clearly does not exist.
\end{ex}

\begin{lem} \label{lemma-tslp-to-dag}
Assume that the $\Gamma$-algebra $\A$ has a finite subsumption base. Then from a given $\Gamma$-TSLP $\calG$ one can
compute in time $\mathcal{O}(|\calG|)$ a $\Gamma$-SLP $\mathcal{H}$ of size 
$\mathcal{O}(|\calG|)$ and depth $\mathcal{O}(\depth(\calG))$
such that $\valX{\calG}^{\hat\A} = \valX{\mathcal{H}}^{\A}$.
\end{lem}

\begin{proof}
Let $C(\A)$ be a finite subsumption base for $\A$.
We say that a context $s \in \calC(\Gamma,\Y)$  belongs to $C(\A)$ 
{\em up to variable renaming} if there is a variable renaming $\theta : \Y \to \Z$
such that $\theta(s) \in C(\Gamma)$.
Since the algebra $\A$ is fixed, the set $C(\A)$ has size $\mathcal{O}(1)$.
Assume that $s$ and $t$ are contexts with the following properties:
(i) $s[t]$ is defined,  (ii) $s$ and $t$ have no common auxiliary variable, and
(iii) $s$ and $t$ belong to $C(\A)$ up to variable renaming.
We denote with $s \cdot t$ a context from $C(\A)$ with $s \cdot t \le^{\A} s[t]$.
Since $s$ and $t$ have size $\mathcal{O}(1)$ ($C(\A)$ is a fixed set of contexts), we can
compute from $s,t$ in constant time the context 
$s \cdot t$ and a substitution $\zeta$ such that $s[t]$ and $\zeta(s \cdot t)$ 
are equivalent in $\A$.
Similarly, one can compute from a given atomic context $s$ in constant time a context $t \in C(\A)$
and a ground substitution $\zeta$ such that $s$ and $\zeta(t)$ are equivalent in $\A$. 

Let $\calG = (\V,\rho,S)$. 
We define $\V_0 = \{ X \in \V \mid \rho^*(X)^{\hat\T(\Gamma)} \in \T(\Gamma)\}$ and 
$\V_1= \{ X \in \V \mid \rho^*(X)^{\hat\T(\Gamma)} \in \calC(\Gamma)\} = \V \setminus \V_0$. 
The $\Gamma$-SLP  $\mathcal{H}$ to be constructed
will be denoted with $\mathcal{H} = (\V', \tau, S)$. 
We will have $\V_0 \subseteq \V'$.
A variable $X \in \V_1$ is replaced in $\mathcal{H}$ by a finite set $\Y_X$ of variables.
Moreover, we will compute a context $t_X \in \calC(\Gamma,\Y_X)$ that belongs to $C(\A)$
up to variable renaming. We can assume that
$\Y_X \cap \Y_{X'} = \emptyset = \Y_X \cap \V_0$ for all $X,X' \in \V_1$ with
$X \neq X'$. The set of variables of $\mathcal{H}$ is then $\V' = \V_0 \cup \bigcup_{X \in \V_1} \Y_X$.
Moreover, $\mathcal{H}$ will satisfy the following conditions:
\begin{enumerate}[(a)]
\item If $X \in \V_0$ then $\rho^*(X)^{\hat\A} = \tau^*(X)^{\A}$ (which is an element of $\A$).
\item If $X \in \V_1$ then $\rho^*(X)^{\hat\A} = \tau^*(t_X)^{\A}$ (which is a ULTF on $\A$).
\end{enumerate}
We construct $\mathcal{H}$ bottom-up. That means that we process all variables
in $\V$ in a single pass over $\calG$. When we process a variable $X \in \V$ we have already
processed all variables $X'$ that appear in $\rho(X)$. In particular, the set $\Y_{X'}$ and the context 
$t_{X'} \in \calC(\Gamma,\Y_{X'})$ (in case $X' \in \V_1$) are defined. In addition, $X'$
satisfies the above conditions (a) and (b).

We proceed by a case distinction according to the right-hand side $\rho(X)$ of  $X \in \V$.
This right-hand side has one of the following four forms:

\medskip
\noindent
{\em Case 1.} $X \in \V_0$ and $\rho(X) = f(X_1, \dots, X_n)$ for $f \in \Gamma_n$ ($n \ge 0$)
and $X_1,\dots,X_n \in \V_0$.
Then we set $\tau(X) = \rho(X)$. Clearly, the above condition (a) holds.

\medskip
\noindent
{\em Case 2.} $X \in \V_0$ and $\rho(X) = X'[X'']$ with $X' \in \V_1$, $X'' \in \V_0$.
By induction we have $\rho^*(X'')^{\hat\A} = \tau^*(X'')^{\A}$.
Moreover, we have computed a context $t_{X'} \in \calC(\Gamma,\Y_{X'})$ that belongs to $C(\A)$
up to variable renaming and such that $\rho^*(X')^{\hat\A} = \tau^*(t_{X'})^{\A}$.
We define $\tau(X) = t_{X'}[X''] \in \T(\Gamma, \Y_{X'} \cup \{X''\})$ (that is,
we replace the main variable in $t_{X'}$ by $X''$)
and get
\begin{eqnarray*}
\rho^*(X)^{\hat\A} = \rho^*(X')^{\hat\A}( \rho^*(X'')^{\hat\A}) &=&  \tau^*(t_{X'})^{\A}(\tau^*(X'')^{\A}) \\
&=& \tau^*(t_{X'}[X''])^{\A} =  \tau^*(\tau(X))^{\A} = \tau^*(X)^{\A} .
\end{eqnarray*}
\noindent
{\em Case 3.} $X \in \V_1$ and $\rho(X) = f(X_1,\ldots,X_{k-1},x,X_{k+1},\ldots,X_n)$ for $f \in \Gamma_n$ ($n \ge 1$)
and $X_1,\ldots,X_{k-1}$, $X_{k+1}, \ldots,X_n \in \V_0$.
By induction  we have $\rho^*(X_i)^{\hat\A} = \tau^*(X_i)^{\A}$ for $1 \leq i \leq n$, $i \neq k$.
We can view $\rho(X)$ as an atomic context with main variable $x$ and auxiliary variables
$X_1,\ldots,X_{k-1},X_{k+1},\ldots,X_n$.
Hence, we can compute $t_X \in C(\A)$ with $t_X \le^{\A} \rho(X)$. We rename the auxiliary
variables of $t_X$ such that they do not already belong to $\mathcal{H}$. Let $\Y_X$ be the set of 
auxiliary variables of $t_X$. We then add all variables in $\Y_X$ to $\mathcal{H}$.
By the definition of $\le^{\A}$ there is a substitution $\zeta \colon \Y_X \to
 \T(\Gamma,\{X_1,\ldots,X_{k-1},X_{k+1},\ldots,X_n\})$ such that
$$
\rho^*(X)^{\hat\A} = \rho^*(\rho(X))^{\hat\A} = \tau^*(\rho(X))^{\A} = \tau^*(\zeta(t_X))^{\A} .
$$
We define the right-hand side for every new variable $Y \in \Y_X$ by
$\tau(Y) = \zeta(Y)$ and get $\rho^*(X)^{\hat\A} = \tau^*(\zeta(t_X))^{\A} = \tau^*(\tau(t_X))^{\A} = \tau^*(t_X)^{\A}$, which 
is point (b).

\medskip
\noindent
{\em Case 4.} $X \in \V_1$ with $\rho_{\calG}(X) = X'[X'']$ and $X', X'' \in \V_1$.
We have already defined the terms $t_{X'}, t_{X''}$ that belong to $C(\A)$ up to variable renaming. 
The set of auxiliary variables of $t_{X'}$ (resp., $t_{X''}$) is $\Y_{X'}$ (resp., $\Y_{X''}$)
and we have $\Y_{X'} \cap \Y_{X''} = \emptyset$. 
Moreover, by the induction hypothesis for $X'$ and $X''$ we have 
$\rho^*(X')^{\hat\A} = \tau^*(t_{X'})^{\A}$  and $\rho^*(X'')^{\hat\A} = \tau^*(t_{X''})^{\A}$.
We set $t_X := t_{X'} \cdot t_{X''} \in C(\A)$. We rename the auxiliary variables of $t_X$
such that  they do not already belong to $\mathcal{H}$. Let $\Y_X$ be the set of auxiliary 
variables of $t_X$. We then add every $Y \in \Y_X$ to $\mathcal{H}$.
By definition of $t_X$ we have $t_X \le^{\A} t_{X'}[t_{X''}]$, which implies that there is a 
substitution $\zeta \colon \Y_X \to \T(\Gamma,\Y_{X'} \cup \Y_{X''})$ with
\begin{eqnarray*}
\rho^*(X)^{\hat\A} = \rho^*(X')^{\hat\A} \circ \rho^*(X'')^{\hat\A} &=& 
\tau^*(t_{X'})^{\A} \circ \tau^*(t_{X''})^{\A} \\
&=& \tau^*(t_{X'}[t_{X''}])^{\A} = \tau^*(\zeta(t_X))^{\A} .
\end{eqnarray*}
We define the right-hand side for every new variable $Y \in \Y_X$ by
$\tau(Y) = \zeta(Y)$ and get
$\rho^*(X)^{\hat\A} = \tau^*(\zeta(t_X))^{\A} = \tau^*(\tau(t_X))^{\A} = \tau^*(t_X)^{\A}$, which is point (b).

The running time for the construction of $\mathcal{H}$ is $\mathcal{O}(|\calG|)$,
since for each variable $X \in \V$ we only spend constant time 
(see the remark from the first paragraph of the proof). In each step we have to take a constant number
of fresh auxiliary variables. We can take them from a list $Y_1, Y_1, Y_3, \ldots$ and store a pointer 
to the next free variable.
\end{proof}

It is known~\cite{GHJLN17,GL18}
that a ranked tree $t$ of size $n$ can be transformed in linear time into a tree
straight-line program of size $\mathcal{O}(n / \log_\sigma n)$ and depth $\mathcal{O}(\log n)$,
where $\sigma$ is the number of different node labels that appear in $t$.
With Lemma~\ref{lemma-tslp-to-dag} it follows that for 
every algebra $\A$ having a finite subsumption base one can compute in linear time from a~given expression tree of size $n$ 
an equivalent circuit of size $\mathcal{O}(n / \log_\sigma n)$ and depth $\mathcal{O}(\log n)$
($\sigma$ is a constant here, namely the number of operations of the algebra $\A$). 

\subsubsection{Main result for $\Gamma$-straight line programs} \label{sec-main-result-balancing}

We now state the main technical result for $\Gamma$-straight line programs.
Note that for some applications we need a signature $\Gamma$ that is part of the input.

\begin{thm} \label{thm-balance-TSLP}
From a given signature $\Gamma$ and a 
$\Gamma$-SLP $\calG$, which defines the tree $t = \valX{\calG} \in \T_0(\Gamma)$,
one can compute in time $\mathcal{O}(|\calG|)$ a  $\Gamma$-TSLP
$\mathcal{H}$ such that $\valX{\mathcal{H}}^{\hat\T(\Gamma)} = t$, 
$|\mathcal{H}| \in \mathcal{O}(|\calG|)$ and $\depth(\mathcal{H}) \in \mathcal{O}(\log |t|)$.
\end{thm}
We will prove Theorem~\ref{thm-balance-TSLP} in Section~\ref{sec-proof}.
Together with Lemma~\ref{lemma-tslp-to-dag}, Theorem~\ref{thm-balance-TSLP} yields the following 
result: 

\begin{thm} \label{cor-balance-dag}
Take a fixed signature $\Gamma$ and a fixed $\Gamma$-algebra $\A$ that has a finite subsumption base.
From a given $\Gamma$-SLP $\calG$,
which defines the derivation tree $t = \valX{\calG} \in \T_0(\Gamma)$,
one can compute in time $\mathcal{O}(|\calG|)$ a $\Gamma$-SLP 
$\mathcal{H}$ such that $\valX{\mathcal{H}}^{\A} = \valX{\calG}^{\A}$, 
$|\mathcal{H}| \in \mathcal{O}(|\calG|)$ and $\depth(\mathcal{H}) \in \mathcal{O}(\log |t|)$.
\end{thm}

\begin{proof}
Using Theorem~\ref{thm-balance-TSLP} we obtain from $\calG$ in time $\mathcal{O}(|\calG|)$ a $\Gamma$-TSLP
$\calG'$ such that $\valX{\calG'}^{\hat\T(\Gamma)} = t$, 
$|\calG'| \in \mathcal{O}(|\calG|)$ and $\depth(\calG') \in \mathcal{O}(\log |t|)$.
From $\valX{\calG'}^{\hat\T(\Gamma)} = t = \valX{\calG}$ we get $\valX{\calG'}^{\hat\A} = \valX{\calG}^{\A}$.
By Lemma~\ref{lemma-tslp-to-dag} we can compute from $\calG'$ in time 
$\mathcal{O}(|\calG'|) = \mathcal{O}(|\calG|)$ a $\Gamma$-SLP $\mathcal{H}$ of size 
$\mathcal{O}(|\calG'|) = \mathcal{O}(|\calG|)$ and depth $\mathcal{O}(\depth(\calG')) = \mathcal{O}(\log |t|)$
such that $\valX{\mathcal{H}}^{\A} = \valX{\calG'}^{\hat\A} = \valX{\calG}^{\A}$.
\end{proof}
Note that Theorem~\ref{cor-balance-dag} is exactly the same statement as Theorem~\ref{thm-general-balancing}
from the introduction (which is formulated via circuits instead of straight-line programs).

\begin{rem}
Recall that we made the technical assumption that every element $a$ of $\A$ can be written as 
$t^{\A}$ for a ground term $\A$. We can still prove Corollary~\ref{cor-balance-dag} in case
$\A$ does not satisfy this assumption: let $\B$ be the subalgebra of $\A$ that is induced by all elements $t^{\A}$ for $t \in \T(\A)$. 
By Remark~\ref{rem-tame-2}, $\B$ has a finite subsumption base as well. Moreover, for every 
$\Gamma$-SLP $\calG$ we obviously have $\valX{\calG}^{\A} = \valX{\calG}^{\B}$. Hence, Corollary~\ref{cor-balance-dag} applied
to the algebra $\B$ yields the statement for $\A$.
\end{rem}

\begin{rem} \label{rem-variable-algebra}
Theorem~\ref{cor-balance-dag} only holds for a fixed $\Gamma$-algebra because Lemma~\ref{lemma-tslp-to-dag} 
assumes a fixed $\Gamma$-algebra. Nevertheless there are settings,
where we consider a family $\{ \A_i \mid i \in I \}$ with the $\A_i$ being $\Gamma_i$-algebras.
An example is the family of all free monoids 
$\Sigma^*$ for a finite alphabet $\Sigma$ that is part of the input. 
Under certain assumptions, the statement of Theorem~\ref{cor-balance-dag} can be extended
to the uniform setting, where the signature $\Gamma_i$ ($i \in I$) is part of the input and SLPs are evaluated in the algebra $\A_i$. First of all
we have to assume that every symbol $f \in \Gamma_i$ fits into a machine word of the underlying RAM 
model, which is a natural assumption if the signature $\Gamma_i$ is part of the input. For the $\Gamma_i$-algebras
$\A_i$ we need the following assumptions:
\begin{enumerate}[(i)]
\item There is a constant $r$ such that the rank of every symbol $f \in \bigcup_{i \in I} \Gamma_i$
is bounded by $r$.
\item There is a constant $c$ and a finite subsumption base $C(\A_i)$ for every $i \in I$
 such that the size of every context $s \in \bigcup_{i \in I} C(\A_i)$ is bounded by $c$.
With the above assumption on the word size of the RAM this ensures that a context $s \in \bigcup_{i \in I} C(\A_i)$
fits into $\mathcal{O}(1)$ many machine words.
\item There is a constant time algorithm that computes from a given atomic context $s$ over the signature $\Gamma_i$ a context $t \in C(\A_i)$
and a substitution $\zeta$ such that $\zeta(t)$ and $s$ are equivalent in $\A_i$.
\item There is a constant time algorithm that takes two contexts $s$ and $t$ over the signature $\Gamma_i$ such that
$s[t]$ is defined, $s$ and $t$ have no common auxiliary variable, and
$s$ and $t$ belong to $C(\A_i)$ up to variable renaming, and computes the context
$s \cdot t$ (see the first paragraph in the 
proof of Lemma~\ref{lemma-tslp-to-dag}) and a substitution $\zeta$ 
such that $\zeta(s \cdot t)$ and $s[t]$ are equivalent in $\A_i$.
\end{enumerate}
Under these assumptions the construction
from the proof of Lemma~\ref{lemma-tslp-to-dag} can still be carried out in linear time.
Since the statement of Theorem~\ref{thm-balance-TSLP} holds for a signature $\Gamma$ that is part
of the input, this allows to extend Theorem~\ref{cor-balance-dag} to the setting where 
the signature $\Gamma_i$ ($i \in I$) is part of the input.
This situation will be encountered for  
forest algebras (Section~\ref{sec:fslps:forests}) and top dags (Section~\ref{sec-cluster}).
\end{rem}
Before we go into the proof of Theorem~\ref{thm-balance-TSLP}, we first discuss a simple applications
of Theorem~\ref{cor-balance-dag} (further applications for certain tree algebras that yield Theorem~\ref{thm-balancing-TSLP}
can be found in Sections~\ref{sec:fslps} and \ref{sec-cluster}).
Consider straight-line programs over a semiring $\A$.
Such straight-line programs are also known as arithmetic circuits in the literature. 
We view addition and multiplication in $\A$ as binary operations. In other words, 
we consider bounded fan-in arithmetic circuits. We also include arbitrary constants in the algebra $\A$
(this is necessary in order to build expressions).
The following result follows directly from Theorem~\ref{thm-balance-TSLP} and the fact that every 
semiring has a finite subsumption base; see Example~\ref{ex-semirings}.

\begin{cor}
Let $\A$ be an arbitrary semiring with constants (we neither assume that $\A$ is commutative nor that identity elements
with respect to $+$ or $\times$ exist).
Given an arithmetic circuit $\mathcal{G}$ over $\A$ such that the corresponding derivation tree $t$ has $n$ nodes,
one can compute in time $\mathcal{O}(|\calG|)$ an arithmetic circuit
$\mathcal{H}$ over $\A$ such that $\valX{\mathcal{H}}^{\A} = \valX{\calG}^{\A}$, 
$|\mathcal{H}| \in \mathcal{O}(|\calG|)$ and $\depth(\mathcal{H}) \in \mathcal{O}(\log n)$.
\end{cor}
Theorem~\ref{thm-balance-SSLP} (balancing of string straight-line programs) can be deduced in the same way from 
Theorem~\ref{thm-balance-TSLP} by noting that every free monoid has a finite subsumption base.

\subsection{Proof of Theorem~\ref{thm-balance-TSLP}} \label{sec-proof}

Those readers that worked through part I of the paper (Section~\ref{sec-part-I}) 
will notice that our proof of Theorem~\ref{thm-balance-TSLP} is very similar to the proof
of Theorem~\ref{thm-balance-SSLP} in Section~\ref{sec:balancing-of-sslps}.
For the following proof we will use the part I results from Sections~\ref{sec-centroid} and \ref{sec-suffixes}.

Let us fix a signature $\Gamma$ and a standard $\Gamma$-SLP $\calG = (\V, \rho, S)$. Let $t = \valX{\calG}$ be its derivation tree and
$n = |t|$. We view $\calG$ also as a  DAG $\D:=(\V,E)$ with node labels from $\Gamma$. 
The edge relation $E$ contains all edges $(X,i,X_i)$ where $\rho(X)$ is of the form 
$f(X_1, \dots, X_n)$ and $1 \le i \le n$. We can assume that all nodes of the DAG are reachable from 
the start variable $S$. 
All variables from $\V$ also belong to the TSLP $\mathcal{H}$ and produce the same trees in $\calG$ and $\mathcal{H}$.
The right-hand side mapping of $\mathcal{H}$ will be denoted by $\tau$.

We start with the symmetric
centroid decomposition of the DAG $\D$, 
which can be computed in linear time as remarked in Section~\ref{sec-centroid}. Note that the number
$n(\D)$ defined in  Section~\ref{sec-centroid} is the number of leaves of $t$. Hence, we have $n(\D) \le n$.
Consider a symmetric centroid path 
\begin{equation} \label{sym-centroid-path}
(X_0, d_0, X_1), (X_1, d_1, X_2), \ldots, (X_{p-1}, d_{p-1}, X_p)
\end{equation}
in $\D$, where all $X_i$ belong to $\V$ and $d_i \geq 1$.
Thus, for all $0 \le i \le p-1$, the right-hand side of $X_i$ in $\calG$ has the form 
\begin{equation} \label{rho(X_i)}
\rho(X_i) = f_i(X_{i,1}, \ldots, X_{i,d_i-1},X_{i+1}, X_{i,d_i+1}, \ldots, X_{i,n_i})
\end{equation}
for $f_i \in \Gamma_{n_i}$, $X_{i,j} \in \V$ for $1 \le j \le n_i$, $j \neq d_i$.
Figure~\ref{fig-ex-spine} shows such a path. Note that the variables $X_{i,j}$ do not have to be pairwise
different (as Figure~\ref{fig-ex-spine} might suggest). Also note that the variables $X_{i,j}$ from \eqref{rho(X_i)} 
and all variables in $\rho(X_p)$ 
belong to other symmetric centroid paths.

We will introduce $\mathcal{O}(p)$ many 
variables in the TSLP $\mathcal{H}$ to be constructed and the sizes of the corresponding right-hand sides
will sum up to $\sum_{i=0}^p |\rho(X_i)| + \mathcal{O}(p)$. By summing over all  symmetric centroid paths of $\D$, this yields the size 
bound $\mathcal{O}(|\calG|)$ for $\mathcal{H}$. 

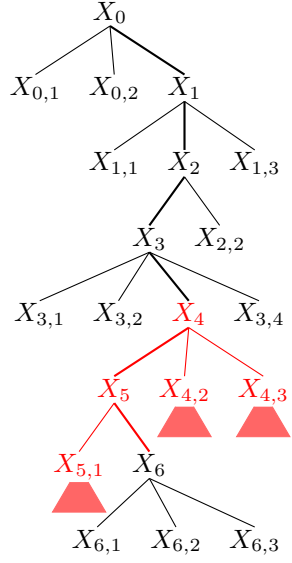
\begin{figure}[t]
\pgfkeys{/pgf/inner sep=.07em}  
  \centering
   \tikzset{
  rect/.append  style={label={[font=\scriptsize]0:#1}}}
  \begin{forest} 
    [$X_0$, s sep=3mm, for tree={parent anchor=south, child anchor=north, l=1cm}
      [$X_{0,1}$, edge label={node[midway,left=1mm,font=\scriptsize]{}}]
      [$X_{0,2}$, edge label={node[midway,left,font=\scriptsize]{}}]
      [$X_1$,  edge label={node[midway,right=1mm,font=\scriptsize]{\textcolor{black}{}}}, s sep=3mm
        [$X_{1,1}$, edge label={node[midway,left=1mm,font=\scriptsize]{}}]
        [$X_2$,  edge label={node[midway,left,font=\scriptsize]{\textcolor{black}{}}}, s sep=3mm
          [$X_3$,  edge label={node[midway,left=.5mm,font=\scriptsize]{\textcolor{black}{}}}, s sep=3mm
            [$X_{3,1}$, edge label={node[midway,left=1.5mm,font=\scriptsize]{}}]
            [$X_{3,2}$, edge label={node[midway,left=.5mm,font=\scriptsize]{}}]
                [\textcolor{red}{$X_4$},  edge label={node[midway,left,font=\scriptsize]{\textcolor{black}{}}}, s sep=3mm
                  [\textcolor{red}{$X_5$},  edge = {red}, edge label={node[midway,left=2mm,font=\scriptsize]{\textcolor{black}{}}}, s sep=3mm
                     [\textcolor{red}{$X_{5,1}$}, name = {X51}, edge = {red},  edge label={node[midway,left=1mm,font=\scriptsize]{}}]
                        [$X_6$,  edge = {red}, edge label={node[midway,left=.5mm,font=\scriptsize]{}}, s sep=3mm, before typesetting nodes={for current and ancestors={edge+={thick}}}
                            [,phantom]
                            [$X_{6,1}$, name = {X61}, edge label={node[pos=.4,left=.5mm,font=\scriptsize]{}}]
                            [$X_{6,2}$, name = {X62}, edge label={node[midway,left,font=\scriptsize]{}}]
                            [$X_{6,3}$, name = {X63}, edge label={node[midway,right=1mm,font=\scriptsize]{}}]
                        ]
                    ]
                  [\textcolor{red}{$X_{4,2}$}, edge = {red}, name = {X42}, edge label={node[midway,right,font=\scriptsize]{}}]
                  [\textcolor{red}{$X_{4,3}$}, edge = {red}, name = {X43}, edge label={node[midway,right=1mm,font=\scriptsize]{}}]
                ]
              [$X_{3,4}$, edge label={node[midway,right=1.5mm,font=\scriptsize]{}}]
            ]
          [$X_{2,2}$, edge label={node[midway,right=.5mm,font=\scriptsize]{}}]
          ]
        [$X_{1,3}$, edge label={node[midway,right=1mm,font=\scriptsize]{}}]
       ]
      ]
      \path[draw, red!60, fill=red!60] (X42) -- ($(X42)+(-.35,-.6)$) -- ($(X42)+(.35,-.6)$) -- (X42) ; 
      \path[draw, red!60, fill=red!60] (X43) -- ($(X43)+(-.35,-.6)$) -- ($(X43)+(.35,-.6)$) -- (X43) ;     
      \path[draw, red!60, fill=red!60] (X51) -- ($(X51)+(-.35,-.6)$) -- ($(X51)+(.35,-.6)$) -- (X51) ; 
\end{forest}
\caption{\label{fig-ex-spine}A symmetric centroid path in the proof of Theorem~\ref{thm-balance-TSLP}.}
\end{figure}

Define the ground terms
$t_i = \valXG{\calG}{X_i}$  for $0 \le i \le p$ and $t_{i,j} = \valXG{\calG}{X_{i,j}}$ for $0 \le i \le p-1$
and $1 \le j \le n_i$, $j \neq d_i$.
Recall that every variable $X_i$ ($0 \le i \le p$) of  $\calG$ also belongs to $\mathcal{H}$. 
For every $0 \le i \le p-1$ we introduce a fresh variable $Y_i$ which will evaluate in $\mathcal{H}$ to the context 
obtained by taking the tree $t_i$ and cutting out the occurrence of the subtree  $t_p$ that 
is reached via the directions $d_i, d_{i+1}, \ldots, d_{p-1}$ from the root of $t_i$. 
In Figure~\ref{fig-ex-spine} this context is visualized for $i=4$ by the red part.
Hence, we set
\begin{equation} \label{rho:X_j}
\tau(X_i) = Y_i[X_p] 
\end{equation}
for $0 \le i \le p$.
For $X_p$ we define 
\begin{equation} \label{rho:X_p}
\tau(X_p) = \rho(X_p) .
\end{equation}
It remains to come up with right-hand sides such that every $Y_i$ derives to the intended context.
For this, we introduce variables $Z_i$ ($0 \le i \le p-1$) and define 
\begin{equation} \label{rho:C_j}
\tau(Z_i) = f_i(X_{i,1}, \ldots, X_{i,d_i-1},x, X_{i,d_i+1}, \ldots, X_{i,n_i})
\end{equation}
for  $0 \le i \le p-1$. It remains to add variables and right-hand sides such that 
every $Y_i$ derives in $\mathcal{H}$ to $Z_i[Z_{i+1}[ \cdots [Z_{p-1}] \cdots]]$.
This is basically a string problem: we want to produce an SSLP for all suffixes
of $Z_0 Z_1 \cdots Z_{p-1}$. This SSLP should have small depth in order to 
keep the total depth of the final TSLP bounded by $\mathcal{O}(\log n)$. Here we use 
Proposition~\ref{prop-all-suffixes}.
For this we have to define the weights of the variables $Z_i$. We set 
$\|Z_i\| = |t_i| - |t_{i+1}|$. We additively extend the weight function to strings 
over the symbols $Z_0, \ldots, Z_{p-1}$.

Using Proposition~\ref{prop-all-suffixes} we can construct
in time $\mathcal{O}(p)$ a single SSLP $\mathcal{I}$ with the following properties:
\begin{itemize}
\item $\mathcal{I}$ has $\mathcal{O}(p)$ many variables and all right-hand sides
have length at most four,
\item $\mathcal{I}$ contains the variables $Y_0, \ldots, Y_{p-1}$, where  
$Y_i$ produces $Z_i Z_{i+1} \cdots Z_{p-1}$ for $0 \le i \le p-1$ and
\item every path from a variable $Y_i$ to a variable $Z_k$ in the derivation tree of $\mathcal{I}$
has length at most $3 + 2\log_2 \|Y_i\| - 2\log_2 \|Z_k\|$ for $i \le k \le p-1$.
\end{itemize}
Note that $\|Y_i\| = |t_i| - |t_p|$.
We finally add to the TSLP $\mathcal{H}$ all right-hand side definitions \eqref{rho:X_j}, \eqref{rho:X_p}, 
\eqref{rho:C_j}, and all right-hand side definitions from the SSLP $\mathcal{I}$. 
Here, we have to replace a concatenation $YZ$ in a right-hand side of $\mathcal{I}$ by $Y[Z]$.

Concerning the number of introduced variables:
for each $X_i$ we introduce $Y_i, Z_i$, so $2p$ in total,
and the $\mathcal{I}$ is guaranteed to have $\mathcal O (p)$ variables as well.
Summed over all paths this yields $\mathcal O (n)$.
For the size of the rules,
each rule introduced in~\eqref{rho:X_p}
is exactly the rule for $X_p$ (i.e., in~\eqref{rho(X_i)})
and similarly a rule for $Z_i$, where $0 \leq i < p$,
corresponds to a rule for $X_i$, in particular,
$|\tau(Z_i)| = |\rho(X_i)|$.
And so the sum of those productions' sizes is $\sum_{i=0}^p |\rho(X_i)|$.
Rules in \eqref{rho:X_j} have size $2$ and there are $p$ of them,
so their productions' size is $2p$.
Lastly, rules introduced as a translation of rules from $\mathcal{I}$
have the same size as those in $\mathcal{I}$,
which is guaranteed to be $\mathcal{O}(p)$.
Thus the sum of rules' sizes is at most
$\sum_{i=0}^p |\rho(X_i)| + \mathcal{O}(p)$.
We make the above construction 
for every symmetric centroid path of $\calG$.
Hence, the total size of the TSLP $\mathcal{H}$ is indeed $\mathcal{O}(|\calG|)$. 
Moreover, the construction of $\mathcal{H}$ needs linear time.
It remains to show that the depth of $\mathcal{H}$ is $\mathcal{O}(\log n)$.

First, we consider the symmetric centroid path \eqref{sym-centroid-path} and a path in $\mathcal{H}$
from a variable $X_i$ ($0 \le i \le p$) to a variable 
$X_{j,k}$ ($i \le j \le p-1$, $1 \le k \le n_j$, $k \neq d_j$)
or a variable from $\rho(X_p)$.
Let us define the weight $\| X \|$ for a variable $X \in \V$ of $\calG$ as the size of the tree $\valXG{\calG}{X}$.
A path from $X_i$ to a variable $Y$ in $\rho(X_p)$
has the form $X_i \to  Y$ or $X_i \to X_p \to Y$ (since $\tau(X_p) = \rho(X_p)$)
and hence has length at most two. Now consider a path from $X_i$ to 
a variable $X_{j,k}$ with $i \le j \le p-1$.
We claim that the length of this path is bounded by $5 + 2\log_2 \| X_i\| - 2\log_2\| X_{j,k} \|$.
The path $X_i \xrightarrow{*} X_{j,k}$ has the form
\begin{equation*}
X_i \to Y_i \xrightarrow{*} Z_{j} \to X_{j,k},
\end{equation*}
where $Y_i \xrightarrow{*} Z_{j}$ is a path in $\mathcal{I}$ and hence has length 
at most $3 + 2\log_2 \|Y_i\| - 2\log_2 \|Z_{j}\|$.
Hence, the length of the path is bounded by
$$
5 + 2\log_2 \|Y_i\| - 2\log_2 \|Z_{j}\| \le 5 + 2\log_2 \|X_i\| - 2\log_2 \|X_{j,k}\|
$$
since $\|Y_i\| = |t_i| - |t_p| \le |t_i| = \|X_i\|$ and 
$\|Z_{j}\| = |t_j| - |t_{j+1}| \ge |t_{j,k}| = \|X_{j,k}\|$.

Finally, we consider a maximal path in the derivation tree of $\mathcal{H}$ that starts in the root $S$ and ends in a leaf.
We can factorize this path as
\begin{equation} \label{full-path}
S = X_0 \xrightarrow{*} X_1 \xrightarrow{*} X_2 \xrightarrow{*} \cdots \xrightarrow{*} X_k
\end{equation}
where all variables $X_i$ belong to the original $\Gamma$-SLP $\calG$, and every 
subpath $X_i \xrightarrow{*} X_{i+1}$ has the form  considered in the last paragraph.
The right-hand side of $X_k$ is a single symbol from $\Gamma_0$ (such a right-hand side can appear in  \eqref{rho:X_p}).
In the $\Gamma$-SLP $\calG$ we have a corresponding path $X_i  \xrightarrow{*} X_{i+1}$ that is contained
in a single symmetric centroid path except for the last edge leading to $X_{i+1}$. 
By the above consideration, the length of the path \eqref{full-path} is bounded by
$$
\sum_{i=0}^{k-1} (5 + 2\log_2 \| X_i\| - 2\log_2 \| X_{i+1}\|) \le 5k + 2\log_2 \| S\|
= 5k+ 2\log_2 n.
$$
By the second claim of 
Lemma~\ref{lem-decomposition} we have $k \le 2 \log_2 n$ which shows that the length of the path 
\eqref{full-path} is bounded by $7\log_2 n$.
This concludes the proof of Theorem~\ref{thm-balance-TSLP}. \qed

\subsection{Forest algebras and forest straight-line programs} \label{sec-FSLP}
\label{sec:fslps}
\newcounter{fslpdef}
\renewcommand{\thefslpdef}{\roman{fslpdef}}
\newcommand{\nextstepdef}[1]{\refstepcounter{fslpdef}\thefslpdef\label{#1}}

\subsubsection{Forest algebra}
\label{sec:fslps:forests} 

Let us fix a finite set $\Sigma$ of node labels.
In this section, we consider $\Sigma$-labelled rooted ordered trees, where ``ordered''
means that the children of a node are totally ordered.  Every node has a label from $\Sigma$.
In contrast to the trees from Section~\ref{sec-trees} we make no rank assumption: the number 
of children of a node (also called its degree)
\emph{is not} determined by its node label.
A {\em forest} is a (possibly empty) sequence of such trees.
The size $|v|$ of a forest is the total number of nodes in $v$.
The set of all $\Sigma$-labelled forests is denoted by $\F_0(\Sigma)$. Formally, $\F_0(\Sigma)$ 
can be inductively defined as the smallest set of strings over the alphabet $\Sigma \cup \{ (, \, )\}$
such that 
\begin{itemize}
\item $\varepsilon \in \F_0(\Sigma)$ (the empty forest),
\item if $u, v \in \F_0(\Sigma)$ then $u v \in \F_0(\Sigma)$, and
\item if $u \in \F_0(\Sigma)$ then $a(u) \in \F_0(\Sigma)$ (this is the forest consisting of a single tree 
whose root is labelled with $a$).
\end{itemize}
Let us fix a distinguished
symbol $\ast \not\in \Sigma$.
The set of forests $u \in \F_0(\Sigma \cup \{\ast\})$ such that $\ast$ has a unique
occurrence in $u$ and this occurrence is at a leaf node is denoted by $\F_1(\Sigma)$.
Elements of $\F_1(\Sigma)$ are called {\em forest contexts}.
Following~\cite{DBLP:conf/birthday/BojanczykW08}, we define the {\em forest algebra}
as the 2-sorted algebra
$$
\mathsf{F}(\Sigma) = ( \F_0(\Sigma), \F_1(\Sigma), \conch_{00},
\conch_{01}, \conch_{10}, \concv_{0}, \concv_{1}, (a(\ast))_{a \in \Sigma}, \varepsilon, \ast)
$$
as follows:
\begin{itemize}
\item $\conch_{ij} \colon \F_i(\Sigma) \times \F_j(\Sigma) \to \F_{i+j}(\Sigma)$ ($ij\in \{00,01,10\}$) 
is a horizontal concatenation operator: for $u \in \F_i(\Sigma)$, $v \in \F_j(\Sigma)$
we set $u \conch_{ij} v = uv$ (i.e., we concatenate the corresponding sequences of trees).
\item $\concv_{i} \colon \F_1(\Sigma) \times \F_i(\Sigma) \to \F_{i}(\Sigma)$ is a vertical concatenation operator: for  
$u \in \F_1(\Sigma)$ and $v \in \F_i(\Sigma)$, $u \concv_{i} v$ is obtained by replacing in $u$ the unique
occurrence of $\ast$ by  $v$.
\item $\varepsilon \in \F_0(\Sigma)$ and $\ast, a(\ast) \in  \F_1(\Sigma)$ ($a \in\Sigma$) are constants of the forest algebra.
\end{itemize}
Note that $(\F_0(\Sigma), \conch_{00}, \varepsilon)$ and $(\F_1(\Sigma), \concv_{1}, \ast)$ are monoids.
In the following we will omit the subscripts $i,j$ in $\conch_{ij}$ and $\concv_i$, since they will be always clear
from the context.
Most of the time, we simply write $uv$ instead of $u \conch v$, $a(u)$ instead of $a(\ast) \concv u$, and 
$a$ instead of $a(\varepsilon)$. With these abbreviations, a forest $u \in \F(\Sigma)$ can be also 
viewed as an algebraic expression over
the algebra $\mathsf{F}(\Sigma)$, which evaluates to $u$ itself (analogously to the free term algebra).

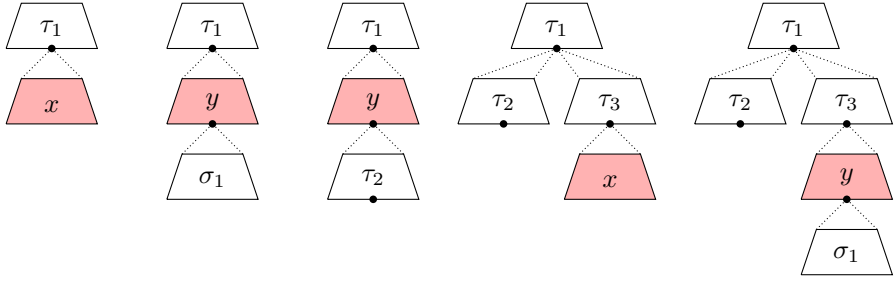
\begin{figure}
\pgfkeys{/pgf/inner sep=.07em}  
  \centering
    \begin{tikzpicture}  
      \node[circle,fill,minimum size=1mm] (a) {} ;
      \draw  ($(a)+(-0.6,0)$) -- node[pos=.5, label={[yshift=1mm]$\tau_1$}] {} ($(a)+(0.6,0)$) --  ($(a)+(0.4,+0.6)$)  -- ($(a)+(-0.4,+0.6)$)  -- ($(a)+(-0.6,0)$) ;
     
      \draw[fill=red!30]  ($(a)+(-0.6,-1)$) -- node[pos=.5, label={[yshift=1mm]$x$}] {} ($(a)+(0.6,-1)$) --  ($(a)+(0.4,-0.4)$)  -- ($(a)+(-0.4,-0.4)$)  -- ($(a)+(-0.6,-1)$) ;
      
      \draw[densely dotted] ($(a)+(-0.4,-0.4)$) -- (a) -- ($(a)+(0.4,-0.4)$);
      
      \node[circle,fill,minimum size=1mm, right = 2cm of a] (b) {} ;
      
      \draw  ($(b)+(-0.6,0)$) -- node[pos=.5, label={[yshift=1mm]$\tau_1$}] {} ($(b)+(0.6,0)$) --  ($(b)+(0.4,+0.6)$)  -- ($(b)+(-0.4,+0.6)$)  -- ($(b)+(-0.6,0)$) ;
     
      \draw[fill=red!30]  ($(b)+(-0.6,-1)$) -- node[circle, fill=black,minimum size=1mm, pos=.5, label={[yshift=1mm]$y$}] (c) {}  ($(b)+(0.6,-1)$) -- ($(b)+(0.4,-0.4)$)  -- ($(b)+(-0.4,-0.4)$)  -- ($(b)+(-0.6,-1)$) ;
      
      \draw ($(c)+(-0.6,-1)$) -- node[pos=.5, label={[yshift=1mm]$\sigma_1$}] {} ($(c)+(0.6,-1)$) --  ($(c)+(0.4,-.4)$)  -- ($(c)+(-0.4,-.4)$)  -- ($(c)+(-0.6,-1)$) ;
      
      \draw[densely dotted] ($(b)+(-0.4,-0.4)$) -- (b) -- ($(b)+(0.4,-0.4)$);
      \draw[densely dotted] ($(c)+(-0.4,-0.4)$) -- (c) -- ($(c)+(0.4,-0.4)$);
      
      \node[circle,fill,minimum size=1mm, right = 2cm of b] (d) {} ;
      
      \draw  ($(d)+(-0.6,0)$) -- node[pos=.5, label={[yshift=1mm]$\tau_1$}] {} ($(d)+(0.6,0)$) --  ($(d)+(0.4,+0.6)$)  -- ($(d)+(-0.4,+0.6)$)  -- ($(d)+(-0.6,0)$) ;
     
      \draw[fill=red!30]  ($(d)+(-0.6,-1)$) -- node[circle, fill=black,minimum size=1mm, pos=.5, label={[yshift=1mm]$y$}] (e) {}  ($(d)+(0.6,-1)$) -- ($(d)+(0.4,-0.4)$)  -- ($(d)+(-0.4,-0.4)$)  -- ($(d)+(-0.6,-1)$) ;
      
      \draw ($(e)+(-0.6,-1)$) -- node[circle, fill=black,minimum size=1mm, pos=.5, label={[yshift=1mm]$\tau_2$}] {} ($(e)+(0.6,-1)$) --  ($(e)+(0.4,-.4)$)  -- ($(e)+(-0.4,-.4)$)  -- ($(e)+(-0.6,-1)$) ;
      
      \draw[densely dotted] ($(d)+(-0.4,-0.4)$) -- (d) -- ($(d)+(0.4,-0.4)$);
      \draw[densely dotted] ($(e)+(-0.4,-0.4)$) -- (e) -- ($(e)+(0.4,-0.4)$);
       
      \node[circle,fill,minimum size=1mm, right = 2.3cm of d] (f) {} ;
      
      \draw  ($(f)+(-0.6,0)$) -- node[pos=.5, label={[yshift=1mm]$\tau_1$}] {} ($(f)+(0.6,0)$) --  ($(f)+(0.4,+0.6)$)  -- ($(f)+(-0.4,+0.6)$)  -- ($(f)+(-0.6,0)$) ;
     
      \draw  ($(f)+(-1.3,-1)$) -- node[circle, fill=black,minimum size=1mm, pos=.5, label={[yshift=1mm]$\tau_2$}] (g) {}  ($(f)+(-0.1,-1)$) -- ($(f)+(-0.3,-0.4)$)  -- ($(f)+(-1.1,-0.4)$)  -- ($(f)+(-1.3,-1)$) ;
      \draw ($(f)+(0.1,-1)$) -- node[circle, fill=black,minimum size=1mm, pos=.5, label={[yshift=1mm]$\tau_3$}] (h) {}  ($(f)+(1.3,-1)$) -- ($(f)+(1.1,-0.4)$)  -- ($(f)+(0.3,-0.4)$)  -- ($(f)+(0.1,-1)$) ;
      
       \draw[fill=red!30]  ($(h)+(-0.6,-1)$) -- node[pos=.5, label={[yshift=1mm]$x$}] {} ($(h)+(0.6,-1)$) --  ($(h)+(0.4,-0.4)$)  -- ($(h)+(-0.4,-0.4)$)  -- ($(h)+(-0.6,-1)$) ;

      \draw[densely dotted] ($(f)+(1.1,-0.4)$) -- (f) -- ($(f)+(0.3,-0.4)$);
     \draw[densely dotted] ($(f)+(-0.3,-0.4)$) -- (f) -- ($(f)+(-1.1,-0.4)$);
      \draw[densely dotted] ($(h)+(-0.4,-0.4)$) -- (h) -- ($(h)+(0.4,-0.4)$);

     \node[circle,fill,minimum size=1mm, right = 3cm of f] (i) {} ;
      
      \draw  ($(i)+(-0.6,0)$) -- node[pos=.5, label={[yshift=1mm]$\tau_1$}] {} ($(i)+(0.6,0)$) --  ($(i)+(0.4,+0.6)$)  -- ($(i)+(-0.4,+0.6)$)  -- ($(i)+(-0.6,0)$) ;
     
      \draw  ($(i)+(-1.3,-1)$) -- node[circle, fill=black,minimum size=1mm, pos=.5, label={[yshift=1mm]$\tau_2$}] (j) {}  ($(i)+(-0.1,-1)$) -- ($(i)+(-0.3,-0.4)$)  -- ($(i)+(-1.1,-0.4)$)  -- ($(i)+(-1.3,-1)$) ;
      \draw ($(i)+(0.1,-1)$) -- node[circle, fill=black,minimum size=1mm, pos=.5, label={[yshift=1mm]$\tau_3$}] (k) {}  ($(i)+(1.3,-1)$) -- ($(i)+(1.1,-0.4)$)  -- ($(i)+(0.3,-0.4)$)  -- ($(i)+(0.1,-1)$) ;
      
       \draw[fill=red!30]  ($(k)+(-0.6,-1)$) -- node[circle, fill=black,minimum size=1mm, pos=.5, label={[yshift=1mm]$y$}] (l) {} ($(k)+(0.6,-1)$) --  ($(k)+(0.4,-0.4)$)  -- ($(k)+(-0.4,-0.4)$)  -- ($(k)+(-0.6,-1)$) ;
       
        \draw ($(l)+(-0.6,-1)$) -- node[pos=.5, label={[yshift=1mm]$\sigma_1$}] {} ($(l)+(0.6,-1)$) --  ($(l)+(0.4,-.4)$)  -- ($(l)+(-0.4,-.4)$)  -- ($(l)+(-0.6,-1)$) ;

      \draw[densely dotted] ($(i)+(1.1,-0.4)$) -- (i) -- ($(i)+(0.3,-0.4)$);
     \draw[densely dotted] ($(i)+(-0.3,-0.4)$) -- (i) -- ($(i)+(-1.1,-0.4)$);
      \draw[densely dotted] ($(k)+(-0.4,-0.4)$) -- (k) -- ($(k)+(0.4,-0.4)$);
      \draw[densely dotted] ($(l)+(-0.4,-0.4)$) -- (l) -- ($(l)+(0.4,-0.4)$);
      
   \end{tikzpicture}
\caption{\label{fig-forest-tame}The shapes of the contexts in $C$ (proof of Lemma~\ref{lem-forest-tame}). 
Forests and forest contexts are represented by trapezoids. The roots of the forests/forest contexts are located
on the top horizontal lines of the trapezoids.
Bullet nodes represent occurrences of $\ast$. Symmetric 
shapes where the roles of $\tau_2$ and $\tau_3$ exchanged are omitted.}

\end{figure}

\begin{lem} \label{lem-forest-tame}
Every forest algebra $\mathsf{F}(\Sigma)$ has a finite subsumption base.
\end{lem}

\begin{proof}
	In the following we denote by $x$ and $y$ the main variables of sorts $\F_0(\Sigma)$ and $\F_1(\Sigma)$, respectively,
	and by $\sigma, \sigma_1, \sigma_2, \dots$ (resp., $\tau, \tau_1, \tau_2, \dots$) auxiliary variables of sorts
	$\F_0(\Sigma)$ (resp., $\F_1(\Sigma)$).
	In the following, subsumption and equivalence of contexts are always meant 
	with respect to the forest algebra $\mathsf{F}(\Sigma)$.
	
	Let $C$ be the set of containing the following contexts (see also Figure~\ref{fig-forest-tame}):
	\begin{enumerate}[(a)]
	\item 
	$\tau_1 \concv x$,
	\item 
	$\tau_1 \concv y \concv \sigma_1$
	and $\tau_1 \concv y \concv \tau_2$,
	\item 
	$\tau_1 \concv (\tau_2 \conch (\tau_3 \concv x))$
	and $\tau_1 \concv ((\tau_2 \concv x) \conch \tau_3)$,
	\item 
	$\tau_1 \concv (\tau_2 \conch (\tau_3 \concv y \concv \sigma_1))$
	and $\tau_1 \concv ((\tau_2 \concv y \concv \sigma_1) \conch \tau_3)$.
	\end{enumerate}
	A context from point (x) (for x = a,b,c,d) will be also called a (x)-context below.
	First notice that every atomic context is of the form
	$\tau \concv x$, $\tau \concv y$,
	$y \concv \sigma$, $y \concv \tau$,
	$\sigma \conch x$, $x \conch \sigma$,
	$\sigma \conch y$, $y \conch \sigma$, 
	$\tau \conch x$, or $x \conch \tau$
	(up to variable renaming). Each of these contexts is subsumed by a context in $C$.
	For the atomic contexts $\tau \concv x$, $\tau \concv y$,
	$y \concv \sigma$, $y \concv \tau$, 
	$\tau \conch x$, and $x \conch \tau$ this is obvious.
	For $\sigma \conch x$ note that $\sigma \conch x$ is equivalent to the context $(\sigma \conch \ast) \concv x$,
	which is subsumed by $\tau_1 \concv x$. A similar argument also applies to $x \conch \sigma$,
	$\sigma \conch y$ and $y \conch \sigma$.
	
	Now consider any context $s \in C$.
	We prove that for any atomic context $s'$ from above, $s'[s]$ is subsumed by some context from $C$.
	\begin{description}
		\item[Case $\tau \concv s$] Since $s$ is of the form $s = \tau_1 \concv s'$
		for some $s'$,
		the context $\tau \concv s = \tau \concv (\tau_1 \concv s')$
		is subsumed by $s \in C$ itself.
		\item[Case $s \concv \sigma$ and $s \concv \tau$] In this case $s$ must be either the (b)-context
		$s = \tau_1 \concv y \concv \tau_2$, a (c)-context or a (d)-context.
		\begin{enumerate}
			\item If $s= \tau_1 \concv y \concv \tau_2$, then $s \concv \sigma$ and 
			$s \concv \tau$ are subsumed by a (b)-context.
			\item Assume that $s$ is a (c)-context, say $s = \tau_1 \concv (\tau_2 \conch (\tau_3 \concv x))$.
			Then $s \concv \sigma$ is equivalent to $(\tau_1 \concv ((\tau_2 \concv \sigma) \conch \tau_3)) \concv x$
			which is subsumed by $\tau_1 \concv x$. Moreover, 
			$s \concv \tau$ is equivalent to $\tau_1 \concv ((\tau_2 \concv \tau) \conch (\tau_3 \concv x))$,
			which is subsumed by $s$ itself.
			\item Assume that $s$ is a (d)-context, say 
			$s = \tau_1 \concv (\tau_2 \conch (\tau_3 \concv y \concv \sigma_1))$.
			Firstly, $s \concv \sigma$ is equivalent to
			$(\tau_1 \concv ((\tau_2 \concv \sigma) \conch \tau_3)) \concv y \concv \sigma_1$,
			which is subsumed by the context $\tau_1 \concv x \concv \sigma_1$.
			Secondly, $s \concv \tau$ is equivalent to
			$\tau_1 \concv ((\tau_2 \concv \tau) \conch (\tau_3 \concv y \concv \sigma_1))$,
			which is subsumed by $s$ itself.
		\end{enumerate}
		\item[Case $\sigma \conch s$ and $s \conch \sigma$] Since $s$ is of the form $s = \tau_1 \concv s'$ for some $s'$
		the context $\sigma \conch s$ is equivalent to $(\sigma \conch \tau_1) \concv s'$, which is subsumed by $s$ itself.
		The case $s \conch \sigma$ is similar.
		\item[Case $\tau \conch s$ and $s \conch \tau$]
		In this case $s$ must be either the (a)-context or the (b)-context $\tau_1 \concv y \concv \sigma_1$.
		If $s$ is the (a)-context $\tau_1 \concv x$ then $\tau \conch s$ is subsumed by
		the (c)-context $\tau_1 \concv (\tau_2 \conch (\tau_3 \concv x))$.
		If $s$ is the (b)-context $\tau_1 \concv y \concv \sigma_1$
		then $\tau \conch s$ is subsumed by the (d)-context 
		$\tau_1 \concv (\tau_2 \conch (\tau_3 \concv y \concv \sigma_1))$.
		The case $s \conch \tau$ is similar.
	\end{description}
	By Lemma~\ref{lem:tame-atomic}, $C$ is a finite subsumption base.
\end{proof}

\begin{rem} \label{remark-func-extension}
Similarly to the proof of Lemma~\ref{lem-forest-tame} one can show that 
for every signature $\Gamma$ the functional extension $\hat\T(\Gamma)$ of the free term algebra $\T(\Gamma)$
has a finite subsumption base as well. Recall from Example~\ref{ex-free} that the free term algebra $\T(\Gamma)$
has no finite subsumption base if $\Gamma$ contains a symbol of rank at least one. 
\end{rem}

\subsubsection{Forest straight-line programs}
\label{sec:fslps:fslps}

A {\em forest straight-line program} over $\Sigma$, FSLP for short, is a straight-line program
$\calG$ over the algebra  $\mathsf{F}(\Sigma)$ such that $\valX{\calG}^{\mathsf{F}(\Sigma)} \in \F_0(\Sigma)$.
Iterated vertical and horizontal concatenations allow to
generate forests, whose depth and width is exponential
in the size of the FSLP.
For an FSLP $\mathcal{G} = (\V,\rho,S)$ and $i \in \{0,1\}$
we define $\V_i = \{ X \in \V \mid \valX{X}^{\mathsf{F}(\Sigma)} \in \F_i(\Sigma) \}$. 
Every right-hand side of a standard FSLP $\mathcal{G}$ must have one of the following forms:
(i) $\varepsilon$ (the empty forest), (ii) $\ast$, (iii) $a(\ast)$ for $a \in \Sigma$,
(iv) $X \conch Y$ (for which we write $XY$)
for $X,Y \in \V$ with $X \in \V_0$ or $Y \in \V_0$, or (v) $X \concv Y$
for $X \in \V_1$ and $Y \in \V$.

\begin{ex}
\label{example:fslp}
Let $n \in \mathbb{N}$. Consider the (non-standard) FSLP 
$$\mathcal{G} = (\{S,X_0, \ldots, X_n,Y_0, \ldots, Y_n\}, \rho, S)$$ over $\{a,b,c\}$
with $\rho$ defined by
$\rho(X_0) = a$, $\rho(X_i) = X_{i-1} X_{i-1}$ for $1 \leq i \leq n$,
$\rho(Y_0) = b(X_n \ast X_n)$, $\rho(Y_i) = Y_{i-1} \concv Y_{i-1}$
for $1 \leq i \leq n$, and $\rho(S) = Y_n \concv c$. We have
$$\valX{\mathcal{G}}^{\mathsf{F}(\Sigma)} = b(a^{2^n} b(a^{2^n} \cdots b (a^{2^n} c \, a^{2^n}) \cdots a^{2^n}) a^{2^n}),$$
where $b$ occurs $2^n$ many times, see Figure~\ref{fig-ex3} for $n=2$.

\pgfkeys{/pgf/inner sep=0.1em}

\begin{figure}
  \centering
  \label{fig:fslp}
  \begin{forest}
    [$b$, s sep=3mm, for tree={parent anchor=south, child anchor=north}
      [$a$]
      [$a$]
      [$a$]
      [$a$]
      [$b$, s sep=3mm
        [$a$]
        [$a$]
        [$a$]
        [$a$]
        [$b$, s sep=3mm
          [$a$]
          [$a$]
          [$a$]
          [$a$]
          [$b$, s sep=3mm
            [$a$]
            [$a$]
            [$a$]
            [$a$]
            [$c$]
            [$a$]
            [$a$]
            [$a$]
            [$a$]
          ]
          [$a$]
          [$a$]
          [$a$]
          [$a$]
        ]
        [$a$]
        [$a$]
        [$a$]
        [$a$]
      ]
      [$a$]
      [$a$]
      [$a$]
      [$a$]
    ]
  \end{forest}
\caption{\label{fig-ex3}Forest $\valX{\mathcal{G}}^{\mathsf{F}(\Sigma)}$ for $n=2$ from Example~\ref{example:fslp}.}
\end{figure}
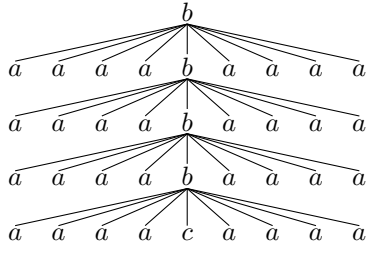
\end{ex}

Let us first show that most occurrences of $\varepsilon$ and $\ast$ can be
eliminated in an FSLP.

\begin{lem} \label{lem-eliminate-eps}
From a given FSLP $\calG$ with $\valX{\mathcal{G}}^{\mathsf{F}(\Sigma)}  \neq \varepsilon$
 one can compute in linear time an FSLP $\mathcal{H}$ such that
$\valX{\mathcal{G}}^{\mathsf{F}(\Sigma)} = \valX{\mathcal{H}}^{\mathsf{F}(\Sigma)}$,
$|\mathcal{H}| \in \mathcal{O}(|\calG|)$,  $\depth(\mathcal{H}) \in \mathcal{O}(\depth(\calG))$, and
$\mathcal{H}$ does not contain occurrences of the constants $\varepsilon$ and $\ast$, except
for right-hand sides of the form $a(\varepsilon)$.\footnote{Constants $a(\ast)$ are allowed as well.
Formally,  $a(\ast)$ is a constant symbol that is interpreted by the forest context $a(\ast)$.}
\end{lem}

\begin{proof}
Let $\calG = (\V,\rho,S)$. We first construct an equivalent FSLP which does
not contain the constant $\ast$.
Let us denote with $\V_\ast \subseteq \V_1$ the set
of all variables $X \in \V_1$ such that $\valX{X}^{\mathsf{F}(\Sigma)}$ is of the form $u_\ell \ast u_r$
for forests $u_\ell, u_r \in \F_0(\Sigma)$. In other words: $\ast$ occurs at a root position in the forest
$\valX{X}^{\mathsf{F}(\Sigma)}$. 
The set $\V_\ast$ can be easily computed in linear time by a single pass over $\calG$.
Every variable $X \in \V_\ast$ with
$\valX{X}^{\mathsf{F}(\Sigma)} = u_\ell \ast u_r$ is replaced in $\mathcal{H}$ by two variables
$X_\ell$ and $X_r$ that produce in $\mathcal{H}$ the forests $u_\ell$ and $u_r$, respectively.  
Every variable $X \in \V_1 \setminus \V_\ast$ is replaced in $\mathcal{H}$ by three variables
$X_t, X_\ell, X_r$. Since $X \in \V_1 \setminus \V_\ast$, $\valX{X}^{\mathsf{F}(\Sigma)}$ contains
a unique subtree of the form $a(u_\ell \ast u_r)$. Let us denote with $u_t$ (the top part of $u$) the 
forest that is obtained from $u$ by replacing the subtree $a(u_\ell \ast u_r)$ by $a(\ast)$.
We then will have $\valX{X_t}^{\mathsf{F}(\Sigma)} = u_t$, $\valX{X_\ell}^{\mathsf{F}(\Sigma)} = u_\ell$, 
and $\valX{X_r}^{\mathsf{F}(\Sigma)} = u_r$. Finally, all variables from $\V_0$ also belong to $\mathcal{H}$
and produce in $\mathcal{H}$ the same forests as in $\calG$.

It is straight-forward to define the right-hand sides of $\mathcal{H}$ such that the variables indeed produce
the desired forests ($\tau$ denotes the right-hand side mapping of $\mathcal{H}$):
\begin{itemize}
\item If $\rho(X) = \ast$ then $\tau(X_\ell) = \tau(X_r) = \varepsilon$.
\item If $\rho(X) = a(\ast)$ then $\tau(X_t) = a(\ast)$ and 
$\tau(X_\ell) = \tau(X_r) = \varepsilon$.
\item If $\rho(X) = \varepsilon$ or $\rho(X) = YZ$ with $X,Y,Z \in \V_0$ then 
$\tau(X) = \rho(X)$.
\item If $\rho(X) = YZ$ with $X,Y \in \V_\ast$ and $Z \in \V_0$ then 
$\tau(X_\ell) = Y_\ell$  and $\tau(X_r) = Y_r Z$, and analogously  for $X,Z \in \V_\ast$ and $Y \in \V_0$.
\item If $\rho(X) = YZ$ with $X,Y \in \V_1 \setminus \V_\ast$ and $Z \in \V_0$ then 
$\tau(X_\ell) = Y_\ell$, $\tau(X_r) = Y_r$, and
$\tau(X_t) = Y_t Z$, and analogously for $X,Z \in \V_1 \setminus \V_\ast$ and $Y \in \V_0$.
\item If $\rho(X) = Y \concv Z$ with $X,Z \in \V_0$ and $Y \in \V_\ast$ then 
$\tau(X) = Y_\ell Z Y_r$.
\item If $\rho(X) = Y \concv Z$ with $X,Z \in \V_0$ and $Y \in \V_1 \setminus \V_\ast$ then 
$\tau(X) = Y_t \concv (Y_\ell Z Y_r)$.
\item If $\rho(X) = Y \concv Z$ with $X,Y,Z \in \V_\ast$ then 
$\tau(X_\ell) = Y_\ell Z_\ell$ and $\tau(X_r) = Z_r Y_r$.
\item If $\rho(X) = Y \concv Z$ with $Y \in \V_\ast$ and $X,Z \in \V_1 \setminus \V_\ast$
then $\tau(X_t) = Y_\ell Z_t Y_r$, $\tau(X_\ell) = Z_\ell$, and 
$\tau(X_r) =Z_r$.
\item If $\rho(X) = Y \concv Z$ with $Z \in \V_\ast$ and $X,Y \in \V_1 \setminus \V_\ast$
then $\tau(X_t) = Y_t$, $\tau(X_\ell) = Y_\ell Z_\ell$, and 
$\tau(X_r) = Z_r Y_r$.
\item If $\rho(X) = Y \concv Z$ with $X,Y,Z \in \V_1 \setminus \V_\ast$
then $\tau(X_t) = Y_t \concv (Y_\ell Z_t Y_r)$, $\tau(X_\ell) = Z_\ell$, and 
$\tau(X_r) = Z_r$.
\end{itemize}
Note that all right-hand sides of the new FSLP have constant length. Variables $X$ such that
$\tau(X)$ is a variable can be eliminated.

Let us finally eliminate occurrences
of the constant $\varepsilon$, except
for right-hand sides of the form $a(\varepsilon)$.
Let us take an FSLP $\calG = (\V,\rho,S)$ with $\valX{\calG}^{\mathsf{F}(\Sigma)} \neq \varepsilon$
and which does not contain occurrences of the constant $\ast$. Let $\V_\varepsilon = \{ X \in \V_0 \mid
\valX{X}^{\mathsf{F}(\Sigma)} = \varepsilon \}$. Note that $S \notin \V_\varepsilon$. 
The set $\V_\varepsilon$ can be easily computed in linear time by a single pass over $\calG$.
We construct an equivalent FSLP $\mathcal{H}$ which neither contains $\ast$ nor $\varepsilon$,
except for right-hand sides of the form $a(\ast)$ and $a(\varepsilon)$.
All variables from $\calG$ are also contained in $\mathcal{H}$, except for 
variables in $\V_\varepsilon$. For every variable $X \in \V_1$,
$\mathcal{H}$ also contains a copy $X_\varepsilon$ that produces
$\valX{X}^{\mathsf{F}(\Sigma)} \concv \varepsilon$. The right-hand side mapping $\tau$
of $\mathcal{H}$ is defined as follows:
\begin{itemize}
\item If $\rho(X) = a(\ast)$ then $\tau(X) = a(\ast)$ and
$\tau(X_\varepsilon) = a(\varepsilon)$.
\item If $\rho(X) = \varepsilon$ then $X$ does not belong to $\mathcal{H}$.
\item If $\rho(X) = YZ$ 
with $Y,Z \in \V_\varepsilon$ then $X \in \V_\varepsilon$ 
does not belong to $\mathcal{H}$.
\item If $\rho(X) = YZ$ or $\rho(X) = ZY$
with $Y \in \V_\varepsilon$ and $Z \in \V_0 \setminus \V_\varepsilon$ then 
$\tau(X) = Z$.
\item If $\rho(X) = YZ$ with $Y,Z \in \V_0 \setminus \V_\varepsilon$ then 
$\tau(X) = YZ$.
\item If $\rho(X) = YZ$ or  $\rho(X) = ZY$ with $Y \in \V_\varepsilon$ and $X,Z\in \V_1$ then 
$\tau(X) = Z$ and $\tau(X_\varepsilon) = Z_\varepsilon$.
\item If $\rho(X) = YZ$ with $Y \in \V_0 \setminus \V_\varepsilon$ and $X,Z\in \V_1$ then 
$\tau(X) = YZ$ and $\tau(X_\varepsilon) = YZ_\varepsilon$, and similarly
if $Z \in \V_0 \setminus \V_\varepsilon$ and $X,Y\in \V_1$.
\item If $\rho(X) = Y \concv Z$ with $X,Y,Z \in \V_1$ then $\tau(X) = Y \concv Z$
and $\tau(X_\varepsilon) = Y \concv Z_\varepsilon$.
\item If $\rho(X) = Y \concv Z$ with $Y \in \V_1$ and $Z \in \V_{\varepsilon}$ then 
$\tau(X) = Y_\varepsilon$.
\item If $\rho(X) = Y \concv Z$ with $Y \in \V_1$ and $Z \in \V_0 \setminus \V_{\varepsilon}$ then 
$\tau(X) = Y \concv Z$.
\end{itemize}
As in the previous case, variables $X$ such that $\tau(X)$ is a variable, can be eliminated.
Note that the construction does not introduce new occurrences of $\ast$.
All variables from $\V \setminus \V_\varepsilon$ produce the same forest in $\calG$ and $\mathcal{H}$,
which implies $\valX{\calG}^{\mathsf{F}(\Sigma)} = \valX{\mathcal{H}}^{\mathsf{F}(\Sigma)}$.
Finally note that both constructions increase the size and depth of the FSLP only by a constant factor. 
\end{proof}

\begin{cor} \label{thm-balance-FSLP}
Given a finite alphabet $\Sigma$ and an FSLP $\calG$ over the forest algebra $\mathsf{F}(\Sigma)$ 
defining the forest $u = \valX{\calG}^{\mathsf{F}(\Sigma)}$,
one can compute in time $\mathcal{O}(|\calG|)$ an FSLP
$\mathcal{H}$ such that $\valX{\mathcal{H}}^{\mathsf{F}(\Sigma)} = u$, $|\mathcal{H}| \in \mathcal{O}(|\calG|)$ 
and $\depth(\mathcal{H}) \in \mathcal{O}(\log |u|)$.
\end{cor}

\begin{proof}
The case $u = \varepsilon$ is trivial. Let us now assume that $u \neq \varepsilon$.
We first apply Lemma~\ref{lem-eliminate-eps} and construct from $\calG$ in linear time an 
equivalent FSLP $\calG'$ which does not contain occurrences of the constants $\ast$ and $\varepsilon$, except
for right-hand sides of the form $a(\varepsilon)$. This ensures that the derivation tree $t = 
\valX{\calG'}$ has size $\mathcal{O}(|u|)$. The size and depth of $\calG'$ are linearly bounded in the size and depth,
respectively, of $\calG$. By Lemma~\ref{lem-forest-tame} we can apply Theorem~\ref{cor-balance-dag} in
order to get the FSLP $\mathcal{H}$ with the desired properties for the situation where the alphabet $\Sigma$
is fixed. For the situation where $\Sigma$ is part of the input one has to use Remark~\ref{rem-variable-algebra}.
The arguments are analogous to the proof of Theorem~\ref{thm-balance-SSLP}. Note in particular that the 
subsumption base from the proof of Lemma~\ref{lem-forest-tame} does not depend on the alphabet $\Sigma$ of the forest
algebra $\mathsf{F}(\Sigma)$.
\end{proof}

\begin{rem}
Using Remark~\ref{remark-func-extension} one can 
show the following variant of Corollary~\ref{thm-balance-FSLP}:
Take a fixed signature $\Gamma$. From a given $\Gamma$-TSLP $\calG$  
defining the tree $t = \valX{\calG}^{\hat\T(\Gamma)}$,
one can compute in time $\mathcal{O}(|\calG|)$ a $\Gamma$-TSLP
$\mathcal{H}$ such that $\valX{\mathcal{H}}^{\hat\T(\Gamma)} = t$, $|\mathcal{H}| \in \mathcal{O}(|\calG|)$ 
and $\depth(\mathcal{H}) \in \mathcal{O}(\log |t|)$. In other words, $\Gamma$-TSLPs can 
be balanced with a linear size increase.
Note that this is a much stronger statement than Theorem~\ref{thm-balance-TSLP}, which
states that a $\Gamma$-SLP can be balanced into an equivalent $\Gamma$-TSLP with
a linear size increase. On the other hand, the above balancing result for $\Gamma$-TSLPs finally uses 
the weaker Theorem~\ref{thm-balance-TSLP} in its proof. 
We have to assume a fixed signature $\Gamma$ in the above argument since the size of the contexts in a 
finite subsumption for $\hat\T(\Gamma)$ depends on the maximal rank of the symbols in $\Gamma$.

Alternatively, the balancing result for $\Gamma$-TSLPs can be deduced from the corresponding
balancing result for FSLPs (Corollary~\ref{thm-balance-FSLP}): A given $\Gamma$-TSLP $\calG$  
can be directly translated into an FSLP $\calG_1$ for the tree $\valX{\calG}^{\hat\T(\Gamma)}$. The size of $\calG_1$ is 
$\mathcal{O}(|\calG|)$. Using Corollary~\ref{thm-balance-FSLP} one can compute from $\calG_1$
a balanced FSLP $\calG_2$ of size $\mathcal{O}(|\calG_1|)$. Finally, the FSLP $\calG_2$ can be easily transformed
back into a $\Gamma$-TSLP of size $\mathcal{O}(r \cdot |\calG_2|)$, where $r$ is the maximal rank of a symbol in
$\Gamma$. For this one has to eliminate horizontal concatenations in the FSLP. Since we assumed $\Gamma$
to be a fixed signature, $r$ is a constant.
\end{rem}

\subsection{Cluster algebras and top dags} \label{sec-cluster}

FSLPs are very similar to {\em top dags} that were introduced in \cite{BilleGLW15} 
and further studied in \cite{BilleFG17,DudekG18,Hubschle-Schneider15}.
In fact, top dags can be defined in the same way as FSLPs, one only has to slightly change
the two concatenation operations $\conch$ and $\concv$, which yields the so called cluster algebra defined
below.
 
 Let us fix an alphabet $\Sigma$ of node labels and define for $a \in \Sigma$ the set
 $\mathcal{K}_a(\Sigma) = \{ a(u) \mid u \in \F_0(\Sigma) \setminus \{\varepsilon\}\}$.
 Note that $\mathcal{K}_a(\Sigma)$ consists of unranked $\Sigma$-labelled trees of size at least two,
 where the root is labeled with $a$.
 Elements of $\mathcal{K}_a(\Sigma)$ (for any $a$) are also called {\em clusters of rank $0$}.
 For $a,b \in \Sigma$ let $\mathcal{K}_{ab}(\Sigma)$ be the set of all
 trees $t \in\mathcal{K}_{a}(\Sigma)$ together with a distinguished $b$-labelled leaf of $t$,
 which is called the {\em bottom boundary node} of $t$.
 Elements of $\mathcal{K}_{ab}(\Sigma)$ (for any $a,b$) are called
 {\em clusters of rank one}. The root node of a cluster $t$ (of rank zero or one) is called the 
 {\em top boundary node} of $t$. 
When writing a cluster of rank one,
we underline the bottom boundary node. For instance $a(b c(\underline{b} a))$
is an element of $\mathcal{K}_{ab}(\Sigma)$. An {\em atomic cluster} is of the form $a(b)$ or $a(\underline{b})$ for $a,b \in \Sigma$.

We define the {\em cluster algebra} $\mathsf{K}(\Sigma)$ as an algebra over a 
$(\Sigma \cup \Sigma^2)$-sorted signature.
The universe of sort $a \in \Sigma$ is 
$\mathcal{K}_a(\Sigma)$ and the universe of sort $ab \in \Sigma^2$ is
$\mathcal{K}_{ab}(\Sigma)$.
The operations of $\mathsf{K}(\Sigma)$ are the following:
\begin{itemize}
\item There are $|\Sigma| + 2|\Sigma|^2$ many horizontal merge operators; we denote all of them with the same symbol $\mergh$.
Their domains and ranges are specified by: 
$\mergh \colon \mathcal{K}_{a}(\Sigma) \times \mathcal{K}_{a}(\Sigma) \to \mathcal{K}_{a}(\Sigma)$,
$\mergh \colon \mathcal{K}_{a}(\Sigma) \times \mathcal{K}_{ab}(\Sigma) \to \mathcal{K}_{ab}(\Sigma)$, and
$\mergh \colon \mathcal{K}_{ab}(\Sigma) \times \mathcal{K}_{a}(\Sigma) \to \mathcal{K}_{ab}(\Sigma)$, where $a,b \in \Sigma$.
All of these merge operators are defined by $a(u) \mergh a(v) = a(uv)$, where sorts of the clusters $u,v$ must match 
the input sorts for one of the merge operators.
\item There are $|\Sigma|^2 + |\Sigma|^3$ many vertical merge operators; we denote all of them with the same symbol $\mergv$.
Their domains and ranges are specified by: $\mergv \colon \mathcal{K}_{ab}(\Sigma) \times \mathcal{K}_{b}(\Sigma) \to 
\mathcal{K}_{a}(\Sigma)$ and $\mergv \colon \mathcal{K}_{ab}(\Sigma) \times \mathcal{K}_{bc}(\Sigma) \to 
\mathcal{K}_{ac}(\Sigma)$ for $a,b,c\in\Sigma$. For clusters $s\in \mathcal{K}_{ab}(\Sigma)$ and $t \in \mathcal{K}_{b}(\Sigma)
\cup \mathcal{K}_{bc}(\Sigma)$ we obtain $s \mergv t$ by replacing in $s$ the bottom boundary node 
by $t$. For instance,
\begin{displaymath}
  a(b c(\underline{b} a)) \mergv b(a c) =
  a(b c(b(ac)a)) .
\end{displaymath}
\item The atomic clusters $a(b)$ and $a(\underline{b})$ are constants of the  cluster algebra.
\end{itemize}
In the following, 
we just write $\mathcal{K}_a$ and $ \mathcal{K}_{ab}$ for $\mathcal{K}_a(\Sigma)$ and $ \mathcal{K}_{ab}(\Sigma)$, respectively.
A {\em top dag} over $\Sigma$ is an SLP $\calG$ over the algebra  $\mathsf{K}(\Sigma)$
such that $\valX{\calG}^{\mathsf{K}(\Sigma)}$ is a cluster of rank zero.\footnote{Note that the definition of a top dag
in~\cite{BilleGLW15} refers to the outcome of a particular top dag construction. In other words: for every
 tree $t$ a very specific SLP over the cluster algebra is constructed and 
 this SLP is called the top dag of $t$. Here, as in \cite{GLMRS18}, we call any SLP over the cluster
 algebra a top dag.}
In our terminology, cluster straight-line program would be a more appropriate name,
but we prefer to use the original term ``top dag''.

\begin{ex}
Consider the top dag $\calG = (\{S,X_0, \ldots, X_n,Y_0, \ldots, Y_n\},\rho,S)$ with
$\rho(X_0) = b(a)$, $\rho(X_i) = X_{i-1} \mergh X_{i-1}$ for $1 \leq i \leq n$,
$\rho(Y_0) = X_n \mergh b(\underline{b}) \mergh X_n$, $\rho(Y_i) = Y_{i-1} \mergv Y_{i-1}$
for $1 \leq i \leq n$, and $\rho(S) = Y_n \mergv b(c)$.
We have
\begin{displaymath}
  \valX{\calG}^{\mathsf{K}(\Sigma)} = b(a^{2^n} b(a^{2^n} \cdots b(a^{2^n} b(c)a^{2^n})\cdots a^{2^n}) a^{2^n}),
\end{displaymath}
where $b$ occurs $2^n+1$ many times.
\end{ex}
In  \cite{GLMRS18} it was shown that from a top dag $\calG$ one can compute in linear time an
equivalent FSLP of size $\mathcal{O}(|\calG|)$. Vice versa, from an FSLP $\mathcal{H}$ for a tree $t \in \calC_a$ (for some $a\in\Sigma$)
one can compute in time $\mathcal{O}(|\Sigma| \cdot |\mathcal{H}|)$ an equivalent top dag of 
size $\mathcal{O}(|\Sigma| \cdot |\mathcal{H}|)$. The additional factor $|\Sigma|$ in the transformation from FSLPs to top dags
is unavoidable; see \cite{GLMRS18} for an example.

\begin{figure}
\pgfkeys{/pgf/inner sep=.07em}  
  \centering
    \begin{tikzpicture}  
      \node[circle,fill,minimum size=1mm] (r1) {} ;
      \draw (r1) -- ($(r1)+(-.5,-.8)$) -- node[circle,minimum size=1mm, fill, pos=.5, label={[yshift=0.2cm]$\tau_1$}] (a){} ($(r1)+(.5,-.8)$) -- (r1) ; 
      \draw (a) -- ($(a)+(-1.8,-.8)$) --  node[pos=.8, label={[yshift=0cm]$\sigma_1$}] {} ($(a)+(-.8,-.8)$) -- (a) ; 
      \draw (a) -- ($(a)+(.8,-.8)$) --  node[pos=.2, label={[yshift=0cm]$\sigma_2$}] {} ($(a)+(1.8,-.8)$) -- (a) ; 
      \draw[fill=red!30] (a) -- ($(a)+(-.5,-.8)$) --  node[pos=.5, label={[yshift=0.2cm]$x$}] {} ($(a)+(.5,-.8)$) -- (a) ; 
      
      \node[circle,fill,minimum size=1mm,right = 4cm of r1] (r2) {} ;
      \draw (r2) -- ($(r2)+(-.5,-.8)$) -- node[circle, fill,minimum size=1mm, pos=.5, label={[yshift=0.2cm]$\tau_1$}] (a){} ($(r2)+(.5,-.8)$) -- (r2) ; 
      \draw (a) -- ($(a)+(-1.8,-.8)$) --  node[pos=.8, label={[yshift=0cm]$\sigma_1$}] {} ($(a)+(-.8,-.8)$) -- (a) ; 
      \draw (a) -- ($(a)+(.8,-.8)$) --  node[pos=.2, label={[yshift=0cm]$\sigma_2$}] {} ($(a)+(1.8,-.8)$) -- (a) ; 
      \draw[fill=red!30] (a) -- ($(a)+(-.5,-.8)$) --  node[circle, fill=black,minimum size=1mm, pos=.5, label={[yshift=0.2cm]$y$}]  (b) {} ($(a)+(.5,-.8)$) -- (a) ; 
      \draw (b) -- ($(b)+(-.5,-.8)$) -- node[pos=.5, label={[yshift=0.2cm]$\sigma_3$}] {} ($(b)+(.5,-.8)$) -- (b) ; 
      
      \node[circle,fill,minimum size=1mm,right = 4cm of r2] (r3) {} ;
      \draw (r3) -- ($(r3)+(-.5,-.8)$) -- node[circle, fill,minimum size=1mm, pos=.5, label={[yshift=0.2cm]$\tau_1$}] (a){} ($(r3)+(.5,-.8)$) -- (r3) ; 
      \draw (a) -- ($(a)+(-1.8,-.8)$) --  node[pos=.8, label={[yshift=0cm]$\sigma_1$}] {} ($(a)+(-.8,-.8)$) -- (a) ; 
      \draw (a) -- ($(a)+(.8,-.8)$) --  node[pos=.2, label={[yshift=0cm]$\sigma_2$}] {} ($(a)+(1.8,-.8)$) -- (a) ; 
      \draw[fill=red!30] (a) -- ($(a)+(-.5,-.8)$) --  node[circle, fill=black,minimum size=1mm, pos=.5, label={[yshift=0.2cm]$y$}]  (b) {} ($(a)+(.5,-.8)$) -- (a) ; 
      \draw (b) -- ($(b)+(-.5,-.8)$) -- node[circle, fill,minimum size=1mm, pos=.5, label={[yshift=0.2cm]$\tau_2$}] {} ($(b)+(.5,-.8)$) -- (b) ; 
      
      \node[circle,fill,minimum size=1mm,below right = 3cm and 1.5cm of r1] (r4) {} ;
       \draw (r4) -- ($(r4)+(-.5,-.8)$) -- node[circle, fill,minimum size=1mm, pos=.5, label={[yshift=0.2cm]$\tau_4$}] (x){} ($(r4)+(.5,-.8)$) -- (r4) ; 
      \draw (x) -- ($(x)+(-.5,-.8)$) -- node[circle, fill,minimum size=1mm, pos=.5, label={[yshift=0.2cm]$\tau_1$}] (a){} ($(x)+(.5,-.8)$) -- (x) ; 
       \draw (x) -- ($(x)+(-1.8,-.8)$) --  node[circle, fill,minimum size=1mm, pos=.5] {} node[pos=.8, label={[yshift=0cm]$\tau_3$}] {} ($(x)+(-.8,-.8)$) -- (x) ; 
      \draw (a) -- ($(a)+(-1.8,-.8)$) --  node[pos=.8, label={[yshift=0cm]$\sigma_1$}] {} ($(a)+(-.8,-.8)$) -- (a) ; 
      \draw (a) -- ($(a)+(.8,-.8)$) --  node[pos=.2, label={[yshift=0cm]$\sigma_2$}] {} ($(a)+(1.8,-.8)$) -- (a) ; 
      \draw[fill=red!30] (a) -- ($(a)+(-.5,-.8)$) --  node[pos=.5, label={[yshift=0.2cm]$x$}]  (b) {} ($(a)+(.5,-.8)$) -- (a) ;

       \node[circle,fill,minimum size=1mm, right = 5cm of r4] (r5) {} ;
       \draw (r5) -- ($(r5)+(-.5,-.8)$) -- node[circle, fill,minimum size=1mm, pos=.5, label={[yshift=0.2cm]$\tau_4$}] (x){} ($(r5)+(.5,-.8)$) -- (r5) ; 
      \draw (x) -- ($(x)+(-.5,-.8)$) -- node[circle, fill,minimum size=1mm, pos=.5, label={[yshift=0.2cm]$\tau_1$}] (a){} ($(x)+(.5,-.8)$) -- (x) ; 
       \draw (x) -- ($(x)+(-1.8,-.8)$) --  node[circle, fill,minimum size=1mm, pos=.5] {} node[pos=.8, label={[yshift=0cm]$\tau_3$}] {} ($(x)+(-.8,-.8)$) -- (x) ; 
      \draw (a) -- ($(a)+(-1.8,-.8)$) --  node[pos=.8, label={[yshift=0cm]$\sigma_1$}] {} ($(a)+(-.8,-.8)$) -- (a) ; 
      \draw (a) -- ($(a)+(.8,-.8)$) --  node[pos=.2, label={[yshift=0cm]$\sigma_2$}] {} ($(a)+(1.8,-.8)$) -- (a) ; 
      \draw[fill=red!30] (a) -- ($(a)+(-.5,-.8)$) --  node[circle, fill=black,minimum size=1mm, pos=.5, label={[yshift=0.2cm]$y$}]  (b) {} ($(a)+(.5,-.8)$) -- (a) ; 
      \draw (b) -- ($(b)+(-.5,-.8)$) -- node[pos=.5, label={[yshift=0.2cm]$\sigma_3$}] {} ($(b)+(.5,-.8)$) -- (b) ; 
      
\end{tikzpicture}
\caption{\label{fig-cluster}The shapes of the contexts in $C$ (proof of Lemma~\ref{lem-cluster-tame}). Bullet nodes represent boundary nodes. Symmetric 
shapes where $\tau_3$ is to the right of $\tau_1$ are omitted.}
\end{figure}
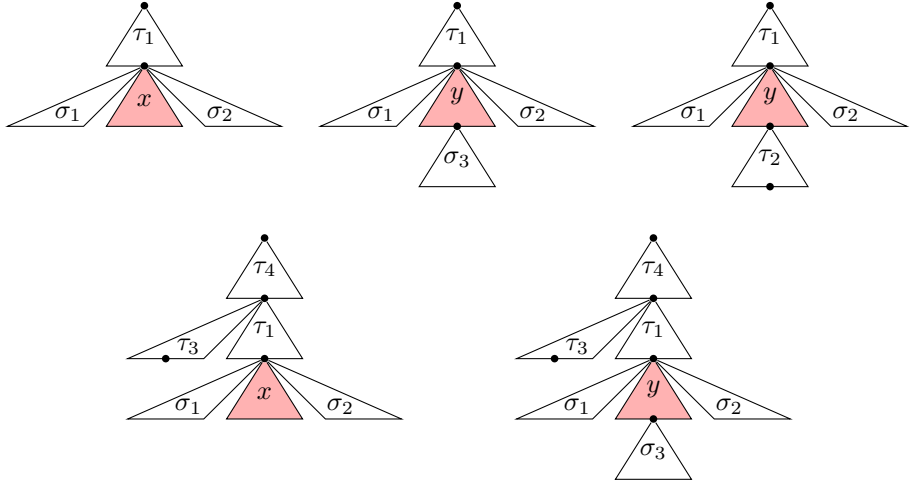

\begin{lem} \label{lem-cluster-tame}
Every cluster algebra $\mathsf{K}(\Sigma)$ has a finite subsumption base.
\end{lem}

\begin{proof}
The proof is similar to the proof of Lemma~\ref{lem-forest-tame}. 
Let the set $C$ contain the following contexts, where in each context, 
each of the auxiliary variables $\sigma_1, \sigma_2, \sigma_3$, $\tau_1, \tau_2, \tau_3, \tau_4$
can be also missing (this is necessary since in the cluster algebra,
the merge operations have no neutral elements).
The main variable $x$ and the auxiliary variables $\sigma_1, \sigma_2, \sigma_3$ must have sorts from $\Sigma$ (rank zero),
whereas the main variable $y$ and the auxiliary variables $\tau_1, \tau_2, \tau_3, \tau_4$ must have sorts from 
$\Sigma\Sigma$ (rank one). The concrete sorts must be chosen 
such that all horizontal and vertical merge operations are defined.
\begin{enumerate}[(a)]
\item $\tau_1 \concv (\sigma_1 \conch x \conch \sigma_2)$ 
\item $\tau_1 \concv (\sigma_1 \conch y \conch \sigma_2) \concv \sigma_3$ 
\item $\tau_1 \concv (\sigma_1 \conch y \conch \sigma_2) \concv \tau_2$ 
\item $\tau_4 \concv (\tau_3 \conch (\tau_1 \concv (\sigma_1 \conch x \conch \sigma_2)))$ 
\item $\tau_4 \concv ((\tau_1 \concv (\sigma_1 \conch x \conch \sigma_2)) \conch \tau_3)$ 
\item $\tau_4 \concv (\tau_3 \conch (\tau_1 \concv (\sigma_1 \conch y \conch \sigma_2) \concv \sigma_3))$ 
\item $\tau_4 \concv ((\tau_1 \concv (\sigma_1 \conch y \conch \sigma_2) \concv \sigma_3) \conch \tau_3)$ 
\end{enumerate}
Note that these forms are very similar to the forms (a)--(g) for forest algebras from the proof of 
Lemma~\ref{lem-forest-tame}. Only the variables $\sigma_1$ and $\sigma_2$ that are horizontally merged with $x$ (resp., $y$)
are new. 

Figure~\ref{fig-cluster} shows the shapes of the above contexts. Let us explain the intuition behind these shapes.
Take a cluster $s$ (of rank zero or one) and cut out from $s$ a subcluster $x$ of rank zero or a subcluster $y$ of rank one. 
We do not give
a formal definition of subclusters (see \cite{BilleGLW15}), but roughly speaking this means that $x$ (resp., $y$) is a 
cluster that occurs somewhere in $s$. In 
Figure~\ref{fig-cluster}, these subclusters are the red triangles. The part of $s$ that does not belong to the subcluster $x$ (resp., $y$)
can be partitioned into finitely many subclusters, and these are the white triangles in 
Figure~\ref{fig-cluster}.

\begin{figure}
\pgfkeys{/pgf/inner sep=.07em}  
  \centering
    \begin{tikzpicture}  
      \node[circle,fill,minimum size=1mm] (a) {} ;
      \draw (a) -- ($(a)+(-1.8,-.8)$) --  node[circle, fill,minimum size=1mm, pos=.5] {} node[pos=.8, label={[yshift=0cm]$\tau$}] {} ($(a)+(-.8,-.8)$) -- (a) ; 
      
       \draw (a) -- ($(a)+(-.5,-.8)$) -- node[circle, fill,minimum size=1mm, pos=.5, label={[yshift=0.2cm]$\tau_1$}] (b){} ($(a)+(.5,-.8)$) -- (a) ; 
      \draw (b) -- ($(b)+(-1.8,-.8)$) --  node[pos=.8, label={[yshift=0cm]$\sigma_1$}] {} ($(b)+(-.8,-.8)$) -- (b) ; 
      \draw (b) -- ($(b)+(.8,-.8)$) --  node[pos=.2, label={[yshift=0cm]$\sigma_2$}] {} ($(b)+(1.8,-.8)$) -- (b) ; 
      \draw[fill=red!30] (b) -- ($(b)+(-.5,-.8)$) --  node[circle, fill=black,minimum size=1mm, pos=.5, label={[yshift=0.2cm]$y$}]  (c) {} ($(b)+(.5,-.8)$) -- (b) ; 
      \draw (c) -- ($(c)+(-.5,-.8)$) -- node[pos=.5, label={[yshift=0.2cm]$\sigma_3$}] (d) {} ($(c)+(.5,-.8)$) -- (c) ;

      \node[circle,fill,minimum size=1mm, right = 6cm of a] (r5) {} ;
   
       \draw (r5) -- ($(r5)+(-.5,-.8)$) -- node[circle, fill,minimum size=1mm, pos=.5, label={[yshift=0.2cm]$\tau_4$}] (x){} ($(r5)+(.5,-.8)$) -- (r5) ; 
          
      \draw (x) -- ($(x)+(-.5,-.8)$) -- node[circle, fill,minimum size=1mm, pos=.5, label={[yshift=0.2cm]$\tau_1$}] (a){} ($(x)+(.5,-.8)$) -- (x) ; 
       \draw (x) -- ($(x)+(-1.8,-.8)$) --  node[circle, fill,minimum size=1mm, pos=.5] (d) {} node[pos=.8, label={[yshift=0cm]$\tau_3$}] {} ($(x)+(-.8,-.8)$) -- (x) ;     
      \draw (d) -- ($(d)+(-1.8,-.8)$) --  node[pos=.8, label={[yshift=0cm]$\sigma$}] {} ($(d)+(-.8,-.8)$) -- (d) ;
      \draw (a) -- ($(a)+(-1.8,-.8)$) --  node[pos=.8, label={[yshift=0cm]$\sigma_1$}] {} ($(a)+(-.8,-.8)$) -- (a) ; 
      \draw (a) -- ($(a)+(.8,-.8)$) --  node[pos=.2, label={[yshift=0cm]$\sigma_2$}] {} ($(a)+(1.8,-.8)$) -- (a) ; 
      \draw[fill=red!30] (a) -- ($(a)+(-.5,-.8)$) --  node[circle, fill=black,minimum size=1mm, pos=.5, label={[yshift=0.2cm]$y$}]  (b) {} ($(a)+(.5,-.8)$) -- (a) ; 
      \draw (b) -- ($(b)+(-.5,-.8)$) -- node[pos=.5, label={[yshift=0.2cm]$\sigma_3$}] {} ($(b)+(.5,-.8)$) -- (b) ;
\end{tikzpicture}
\caption{\label{fig-comp1}Two example cases from the proof of Lemma~\ref{lem-cluster-tame}.}
\end{figure}
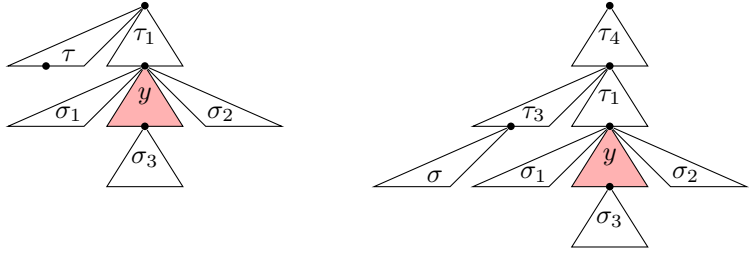

Using Lemma~\ref{lem:tame-atomic} we can show that $C$ is a finite subsumption base for 
the cluster algebra $\mathsf{K}(\Sigma)$.
The atomic clusters are $\tau \conch x$, $\sigma \conch x$, $\sigma \conch y$, 
$x \conch \tau$, $x \conch \sigma$, $y \conch \sigma$, $\tau \concv x$, $\tau \concv y$,
$x \concv \tau$, $x \concv \sigma$ (where $x$ and  $\sigma$ have sorts from $\Sigma$ and 
$y$ and $\tau$ have sorts from $\Sigma\Sigma$). Each of these atomic contexts belongs to $C$ 
 up to renaming of auxiliary variables. For this it is important that every context from the above list (a)--(g), where some 
 of the auxiliary variables are omitted, belongs to $C$ as well. 

Let us now consider a context $s'[s]$, where $s \in C$ and $s'$ is atomic.
We have to show that $s'[s]$ is subsumed in $\mathsf{K}(\Sigma)$ by a context
from $C$. The case distinction is very similar to the proof of 
Lemma~\ref{lem-forest-tame}. Two examples are shown in Figure~\ref{fig-comp1}.
The left figure shows the case 
$s = \tau_1 \concv (\sigma_1 \conch y \conch \sigma_2) \concv \sigma_3$ and $s' = \tau \mergh x$. In this case
$s'[s] = \tau \mergh (\tau_1 \concv (\sigma_1 \conch y \conch \sigma_2) \concv \sigma_3)$
is subsumed in $\mathsf{K}(\Sigma)$ by
$\tau_3 \conch (\tau_1 \concv (\sigma_1 \conch y \conch \sigma_2) \concv \sigma_3)$
(the latter is obtained from the context in (f) by removing $\tau_4$).

Figure~\ref{fig-comp1} on the right shows the case 
$s = \tau_4 \mergv (\tau_3 \mergh (\tau_1 \mergv (\sigma_1 \mergh y \mergh \sigma_2) \mergv \sigma_3))$ and $s' = y \mergv \sigma$.
We have
$s'[s] = \tau_4 \mergv ((\tau_3 \mergv \sigma) \conch (\tau_1 \mergv (\sigma_1 \mergh y \mergh \sigma_2) \mergv \sigma_3))$,
which is equivalent in $\mathsf{K}(\Sigma)$ to 
$(\tau_4 \mergh ( (\tau_3 \mergv \sigma) \mergh \tau_1)) \mergv  (\sigma_1 \mergh y \mergh \sigma_2) \mergv \sigma_3$.
The latter context is subsumed in $\mathsf{K}(\Sigma)$ by $\tau_1 \concv (\sigma_1 \conch y \conch \sigma_2) \concv \sigma_3 \in C$.
\end{proof}
We can now show the main result for top dags:

\begin{cor} \label{coro-top-dag}
Given a finite alphabet $\Sigma$ and a top dag $\calG$
over the cluster algebra $\mathsf{K}(\Sigma)$ producing the tree $t = \valX{\calG}^{\mathsf{K}(\Sigma)}$,
one can compute in time $\mathcal{O}(|\calG|)$ a top dag
$\mathcal{H}$ for $t$ of size $\mathcal{O}(|\calG|)$ and depth $\mathcal{O}(\log |t|)$.
\end{cor}

\begin{proof}
Note that in the derivation tree $\valX{\calG}$ of a top dag $\calG$, all
leaves are labelled with atomic clusters and all internal nodes have rank two.
Hence, the size of the derivation tree $\valX{\calG}$ is linearly bounded in the size of the generated tree
$\valX{\calG}^{\mathsf{K}(\Sigma)}$ (in the forest algebra, we needed Lemma~\ref{lem-eliminate-eps}
to enforce this property). For the case of a fixed alphabet $\Sigma$, the statement of the corollary 
follows from Lemma~\ref{lem-cluster-tame} and Theorem~\ref{cor-balance-dag}
analogously to Corollary~\ref{thm-balance-FSLP} for FSLPs. 
For the general case of a variable-size alphabet $\Sigma$ we have to use again Remark~\ref{rem-variable-algebra}.
As for SSLPs and FSLPs we need the natural assumption that symbols from the input alphabet fit into a single
machine word of the RAM. All operations from a cluster algebra have rank zero and two, and 
the subsumption base $C$ from the proof of Lemma~\ref{lem-cluster-tame}
has the property that every context $s \in \Sigma$ has constant size. In contrast to free monoids
and forest algebras, the subsumption base depends on the alphabet $\Sigma$.
Basically, we need to choose
the sorts of the variables $x,y,\tau_1, \tau_2, \tau_3$, $\sigma_1, \sigma_2, \sigma_3, \sigma_4$ in each of the contexts from
$C$.  This implies that every context $s \in C$ can be represented by a constant number of symbols from $\Sigma$ 
and hence can be stored in a constant number of machine words. The constant time algorithms from point (iii) and (iv)
from Remark~\ref{rem-variable-algebra} make a constant number of comparisons between the 
$\Sigma$-symbols representing the input contexts.
\end{proof}
In~\cite{GaJoMoWe19} top dags have been used for compressed range minimum queries (RMQs).
It is well known that for a string $s$ of integers one can reduce RMQs to lowest common ancestor
queries on the Cartesian tree corresponding to $s$.
Two compressed data structures for answering RMQs for $s$ are proposed in \cite{GaJoMoWe19}: 
one is based on an SSLP for $s$, we commented on it already in Section~\ref{sec-applications},
the other one uses a top dag for the Cartesian tree corresponding to $s$.
The following result has been shown, see \cite[Corollary~1.4]{GaJoMoWe19}:\\
\noindent Given a string $s$ of length $n$ over an alphabet of $\sigma$ many integers,
let $m_{\text{opt}}$ denote the size of a smallest SSLP for $s$.
There is a top dag $\calG$ for  the Cartesian tree corresponding to $s$
of size $|\calG| \le \min(\mathcal{O}(n/\log n), \mathcal{O}(m_{\text{opt}} \cdot \log n\cdot \sigma))$,
and there is a data structure of size $\mathcal{O}(|\calG|)$ that answers range
minimum queries on $s$ in time $\mathcal{O}(\log \sigma \cdot \log n)$.%
\noindent

As the time bound $\mathcal{O}(\log \sigma \cdot \log n)$ comes from the height of the constructed top dag,
using Corollary~\ref{coro-top-dag} we can enforce the bound $\mathcal{O}(\log n)$ on the height 
of the constructed top dag and ensure that the transformation can be applied to any input SSLP.
This yields the following improvement of the result of \cite{GaJoMoWe19}:

\begin{thm}
\label{thm:RMQ_via_top-dag}
Given an SSLP of size $m$ generating a string $s$ of length $n$ over an alphabet of $\sigma$ many integers
one can compute a top dag $\calG$ for the Cartesian tree corresponding to $s$
of size $|\calG| \le \min(\mathcal{O}(n/\log n), \mathcal{O}(m \cdot \sigma))$
and depth $\mathcal{O}(\log n)$,
and there is a data structure of size $\mathcal{O}(|\calG|)$ that answers RMQs on $s$ in time $\mathcal{O}(\log n)$.
If $m_{\text{opt}}$ denotes the size of a smallest SSLP generating $s$
then, using Rytter's algorithm, we can assume that $m \le \mathcal {O}(m_{\text{opt}} \cdot \log n)$.
\end{thm}


\section{Open problems}

For SSLPs one may require a strong notion of balancing. Let us say that an SSLP $\calG$ is 
$c$-balanced if (i) the length of every right-hand side
is at most $c$ and (ii) if a variable $Y$ occurs in $\rho(X)$ then $|\valXG{\calG}{Y}| \leq |\valXG{\calG}{X}|/2$.
It is open, whether there is a constant $c$ such that for every SSLP of size $m$ there exists an equivalent
$c$-balanced SSLP of size $\mathcal{O}(m)$.

Another important open problem is whether the query time bound in Theorem~\ref{cor-random-access} 
(random access to grammar-compressed strings)
can be improved from $\mathcal{O}(\log n)$ to $\mathcal{O}(\log n / \log \log n)$.
If we allow space $\mathcal{O}(m \cdot \log^{\epsilon} n)$ (for any small $\epsilon>0$) then such an improvement is possible
by Corollary~\ref{cor-fusion}, but it is open whether query time $\mathcal{O}(\log n / \log \log n)$ can be achieved with space
$\mathcal{O}(m)$. By the lower bound from \cite{VerbinY13} this would be an optimal
random-access data structure for grammar-compressed strings.


\end{document}